\documentclass[11pt]{article}

\usepackage[hyphens]{url}
\usepackage[pagebackref,colorlinks]{hyperref}
\usepackage[hyphenbreaks]{breakurl}
\usepackage{amsmath} 
\usepackage{amsthm} 
\usepackage{amssymb}	
\usepackage{graphicx} 
\usepackage{multicol} 
\usepackage{multirow}
\usepackage{color}
\usepackage[dvips,letterpaper,margin=1in,bottom=1in]{geometry}
\usepackage[capitalize,noabbrev]{cleveref}

\usepackage{bm}
\usepackage{bbm}

\usepackage{ytableau}
\usepackage{longtable}

\usepackage[utf8]{inputenc}
\usepackage[english]{babel}

\usepackage[T1]{fontenc}
\AtBeginDocument{%
  \DeclareFontShape{T1}{cmr}{m}{scit}{<->ssub*cmr/m/sc}{}%
}

\usepackage{diagbox}
\usepackage{mathtools}

\newtheorem{theorem}{Theorem}[section]

\newtheorem{lemma}[theorem]{Lemma}
\newtheorem{corollary}[theorem]{Corollary}

\newtheorem{definition}[theorem]{Definition}

\newcommand{\braket}[2]{\langle #1 | #2 \rangle}

\DeclarePairedDelimiter\rbra{\lparen}{\rparen}
\DeclarePairedDelimiter\sbra{\lbrack}{\rbrack}
\DeclarePairedDelimiter\cbra{\{}{\}}
\DeclarePairedDelimiter\abs{\lvert}{\rvert}
\DeclarePairedDelimiter\Abs{\lVert}{\rVert}
\DeclarePairedDelimiter\ceil{\lceil}{\rceil}
\DeclarePairedDelimiter\floor{\lfloor}{\rfloor}
\DeclarePairedDelimiter\ket{\lvert}{\rangle}
\DeclarePairedDelimiter\bra{\langle}{\rvert}

\DeclareMathOperator*{\E}{\mathbb{E}}

\newcommand{\set}[2] {\left\{\, #1 \colon #2 \,\right\}}

\newcommand{\tr} {\operatorname{tr}}
\newcommand{\poly} {\operatorname{poly}}

\newcommand{\rank} {\operatorname{rank}}
\newcommand{\supp} {\operatorname{supp}}
\newcommand{\spanspace} {\operatorname{span}}

\newcommand{\ketbra}[2]{\ensuremath{\ket{#1}\!\bra{#2}}}

\usepackage{latexsym}
\usepackage{CJK}

\usepackage{enumerate}

\usepackage{algorithm}
\usepackage{algpseudocode}

\usepackage{stmaryrd}
\usepackage{tabularx}
\usepackage{booktabs}
\usepackage{adjustbox}

\newcommand{\footremember}[2]{%
    \footnote{#2}
    \newcounter{#1}
    \setcounter{#1}{\value{footnote}}%
}

\usepackage{tikz}
\newcommand*\circled[1]{\tikz[baseline=(char.base)]{\node[shape=circle,draw,inner sep=2pt] (char) {#1};}}
\usetikzlibrary{quantikz2}

\begin{document}

\title{A Unified Complexity Framework for Quantum Property Testing}
\author{Qisheng Wang \footremember{1}{Qisheng Wang is with the School of Computer Science, Shanghai Jiao Tong University (e-mail: \url{QishengWang1994@gmail.com}).}}
\date{}

\maketitle
\pagenumbering{roman}
\thispagestyle{empty}

\begin{abstract}
We develop a unified framework for analyzing the complexity of quantum property testing through functionals of the form $\mathcal{L}_{\phi}(\rho) = \operatorname{tr}(\phi(d\rho))/d$, where $\rho$ is an unknown $d$-dimensional quantum state and $\phi$ is a given function. 
A master theorem is established that derives sample complexity lower bounds for estimating $\mathcal{L}_\phi(\rho)$ from properties of $\phi$, combining Haar-random moment encoding with moment matching and best polynomial approximation. 
Corresponding query complexity lower bounds follow from quantum sample-to-query lifting. 

The framework yields (nearly) tight bounds for a broad class of problems, including entropy estimation (von Neumann, R\'enyi, and Tsallis), closeness estimation (trace distance and Uhlmann fidelity), spectrum estimation, rank testing (operator rank, Schmidt rank, and matrix product states).
In particular, combined with known upper bounds, we show that the sample complexity of von Neumann entropy estimation is 
\[
\Theta\!\left(\frac{d^2}{\varepsilon\log^2(d)\max\{1, \varepsilon\log^2(d)\}} + \frac{\log^2(d)}{\varepsilon^2}\right),
\]
and the sample complexity of spectrum estimation is 
\[
\Theta\!\left(\frac{d^2}{\varepsilon^2\log^2(d)\max\{1, \varepsilon^2\log^2(d)\}}\right).
\]
In summary, our results resolve several open problems and establish the optimality of 37 quantum algorithms since 2015, up to constant or polylogarithmic factors.
\end{abstract}

\newpage
\tableofcontents
\thispagestyle{empty}
\newpage
\pagenumbering{arabic}

\section{Introduction}

Quantum property testing \cite{MdW16} plays an important role in quantum computing, with potential applications in the verification of quantum devices, data analysis, and the discovery of quantum advantages. 
However, only a few quantum property testing tasks have previously been known to have optimal approaches, e.g., quantum state tomography \cite{HHJ+17,OW16,OW17,Yue23,vAAGN23,SSW25,PSW26,PSTW26a} and quantum state certification \cite{OW21,BOW19,OW26}. 

In this paper, we develop a unified complexity framework for quantum property testing, thus establishing nearly tight lower bounds for a large class of quantum property testing problems from various lines of research: \textbf{entropy estimation} (von Neumann/R\'enyi/Tsallis), \textbf{closeness estimation} (trace distance/Uhlmann fidelity/Hellinger distance), \textbf{spectrum estimation}, and \textbf{rank testing} (operator rank/Schmidt rank/matrix product states/tree tensor network states). 
As a result, we show that 37 quantum algorithms since 2015 are nearly optimal, which we list in \cref{tab:list} numbered from \circled{1} to \circled{37}.
In particular, combining the upper bounds known in the literature, we establish the following sample complexities that are optimal up to constant factors or up to polylogarithmic factors:
\begin{itemize}
    \item Von Neumann entropy estimation: ${\Theta}\rbra{\frac{d^2}{\varepsilon\log^2\rbra{d} \max\cbra{1,  \varepsilon\log^2\rbra{d}}} + \frac{\log^2\rbra{d}}{\varepsilon^2}}$. 
    \item Spectrum estimation: ${\Theta}\rbra{\frac{d^2}{\varepsilon^2\log^2\rbra{d} \max\cbra{1,  \varepsilon^2\log^2\rbra{d}}}}$.
    \item $\alpha$-R\'enyi entropy estimation: $\widetilde{\Theta}\rbra{\frac{d^{1+1/\alpha}}{\varepsilon^{1/\alpha}}+\frac{d^{1/\alpha-1}}{\varepsilon^2}}$ for $0 < \alpha < 1$ and $\widetilde{\Theta}\rbra{\frac{d^2}{\varepsilon^{1/\alpha}}+\frac{d^{1-1/\alpha}}{\varepsilon^2}}$ for non-integer $\alpha > 1$.
    \item $\alpha$-Tsallis entropy estimation: $\widetilde{\Theta}\rbra{\frac{d^{1+1/\alpha}}{\varepsilon^{1/\alpha}}+\frac{d^{2-2\alpha}}{\varepsilon^2}}$ for $0 < \alpha < 1$ and $\widetilde{\Theta}\rbra{{1}/{\varepsilon^{\frac{2}{\alpha-1}}}}$ for $1 < \alpha < 2$.
    \item Trace distance/Uhlmann fidelity/squared Hellinger distance: $\widetilde{\Theta}\rbra{\frac{d^2}{\varepsilon^2}}$. 
    \item Rank testing: $\widetilde{\Theta}\rbra{\frac{r^2}{\varepsilon}}$.
    \item Schmidt rank testing: $\widetilde{\Theta}\rbra{\frac{r^2}{\varepsilon^2}}$.
    \item Matrix product state/tree tensor network state testing: $\widetilde{\Theta}\rbra{\frac{nr^2}{\varepsilon^2}}$. 
\end{itemize}

\paragraph{Roadmap for readers.}
We present our master theorem in \cref{sec:master}. 
For the readers who are interested in the detail of the numerous applications, you can skip directly to \cref{sec:app}.
The techniques for both the master theorem and the applications are discussed in \cref{sec:tech}. 

\subsection{Master theorem} \label{sec:master}

Our framework builds on the key observation that a large class of quantum property testing problems essentially estimate functionals of quantum states of the form 
\[
\mathcal{L}_{\phi}\rbra{\rho} = \frac{1}{d} \tr\rbra*{\phi\rbra{d\rho}},
\]
where $\rho$ is an unknown $d$-dimensional quantum state and $\phi \colon [0, +\infty) \to \mathbb{R}$ is a known function. 
For example, 
\begin{itemize}
    \item The function $\phi\rbra{x} = -x\log\rbra{x}$ gives the von Neumann entropy $\mathrm{S}\rbra{\rho} = -\tr\rbra{\rho\log\rbra{\rho}} = \mathcal{L}_{-x\log\rbra{x}}\rbra{\rho} + \log\rbra{d}$.
    \item The function $\phi\rbra{x} = x^\alpha$ gives the R\'enyi entropy $\mathrm{S}_{\alpha}^{\mathrm{R}}\rbra{\rho} = \frac{1}{1-\alpha} \log\rbra{\tr\rbra{\rho^\alpha}} = \frac{1}{1-\alpha}\log\rbra{\mathcal{L}_{x^\alpha}\rbra{\rho}} + \log\rbra{d}$ and the Tsallis entropy $\mathrm{S}_{\alpha}^{\mathrm{T}}\rbra{\rho} = \frac{1}{1-\alpha} \rbra{\tr\rbra{\rho^\alpha} - 1} = \frac{1}{1-\alpha}\rbra{d^{1-\alpha}\mathcal{L}_{x^\alpha}\rbra{\rho} - 1}$.
    \item The function $\phi\rbra{x} = \max\cbra{1-x, 0}$ gives the trace distance $\mathrm{T}\rbra{\rho, \mathbb{I}/d} = \mathcal{L}_{\max\cbra{1-x, 0}}\rbra{\rho}$, where $\mathrm{T}\rbra{\rho, \sigma} = \frac{1}{2} \tr\rbra{\abs{\rho - \sigma}}$.
    \item The function $\phi\rbra{x} = \sqrt{x}$ gives the Uhlmann fidelity $\mathrm{F}\rbra{\rho, \mathbb{I}/d} = \mathcal{L}_{\sqrt{x}}\rbra{\rho}$, where $\mathrm{F}\rbra{\rho, \sigma} = \tr\rbra{\sqrt{\sqrt{\sigma}\rho\sqrt{\sigma}}}$.
    \item The function $\phi\rbra{x} = \mathbf{1}_{\cbra{x > 0}}$ gives the rank $\rank\rbra{\rho} = d \mathcal{L}_{\mathbf{1}_{\cbra{x > 0}}}\rbra{\rho}$.
\end{itemize}
Inspired by these, we aim to characterize the difficulty in estimating the functional $\mathcal{L}_{\phi}\rbra{\rho}$ using the properties of $\phi$. 

\begin{theorem}[Master theorem, informal and simplified \cref{thm:master-functional}] \label{thm:main-intro}
    Let $\rbra{\nu_0, \nu_1}$ be a $K$-moment matching on an interval $I = \sbra{a_0, b_0} \subseteq [0, +\infty)$ with mean $\leq 1$. 
    Consider the following parameters:
    \[
    A_\phi = \sup_{x,y\in I} \abs*{\phi\rbra{x} - \phi\rbra{y}}, \qquad L_\phi = \sup_{\substack{x, y \in \sbra{1/2, 4/3} \\ x \neq y}} \abs*{\frac{\phi\rbra{x}-\phi\rbra{y}}{x - y}}, \qquad g_\phi = \abs*{ \E_{X \sim \nu_0} \sbra*{\phi\rbra{X}} - \E_{X \sim \nu_1} \sbra*{\phi\rbra{X}} }.
    \]
    For sufficiently large $d$, if the following three conditions hold:
    \begin{enumerate}
        \item $K = \Omega\rbra{\log\rbra{d}}$, 
        \item $\rbra{1 + M}^2 = O\rbra{d/\log\rbra{d}}$, where $M = \max\cbra{\abs{a_0-1}, \abs{b_0-1}}$, 
        \item $\rbra{A_\phi + \rbra{b_0 - a_0} L_\phi} / g_\phi = O\rbra{\sqrt{d}}$,
    \end{enumerate}
    then there are two probability distributions $\Pi_0$ and $\Pi_1$ of $d$-dimensional quantum states and intervals $I_0, I_1$ with $\textup{dist}\rbra{I_0, I_1} = \Omega\rbra{g_\phi}$ such that 
    \begin{align*}
        \mathrm{T}\rbra*{\E_{\rho \sim \Pi_0}\sbra*{\rho^{\otimes n}}, \E_{\rho \sim \Pi_1}\sbra*{\rho^{\otimes n}}} \leq \frac{1}{\poly\rbra{d}}, \qquad & 1 \leq n \leq O\rbra*{\frac{d^2}{M^2}}, \\
        \Pr_{\rho \sim \Pi_b}\sbra*{\mathcal{L}_\phi\rbra{\rho} \in I_b} \geq 0.99, \qquad & b \in \cbra{0, 1}. 
    \end{align*}
    Therefore, estimating $\mathcal{L}_\phi\rbra{\rho}$ to within additive error $\Theta\rbra{g_\phi}$ requires sample complexity $\Omega\rbra{d^2/M^2}$. 
\end{theorem}

To apply \cref{thm:main-intro}, we only have to find appropriate probability distributions $\rbra{\nu_0, \nu_1}$ with $K$ matching moments and interval $I$ for the function $\phi$ of interest. 
Then, an estimator for $\mathcal{L}_\phi\rbra{\rho}$ can distinguish $\rho \sim \Pi_0$ from $\rho \sim \Pi_1$, whereas this is impossible when samples of $\rho$ are insufficient. 

\subsection{Applications} \label{sec:app}

For each application, we provide both sample and query complexity lower bounds.
The sample complexity is measured by the number of samples of quantum states used.
The query complexity is measured by the number of queries to the state-preparation oracle (also known as the purified quantum query access oracle) for the quantum state; specifically, purified quantum query access to a quantum state $\rho$ means access to a unitary oracle $U$ that prepares the purification $\ket{\psi} = U\ket{0}$ of $\rho$, i.e., $\rho = \tr_{\textup{env}}\rbra{\ketbra{\psi}{\psi}}$ is obtained by tracing out the environment system. 
The sample lower bound is obtained through the master theorem (\cref{thm:main-intro}) and the query lower bound is further obtained by combining the quantum sample-to-query lifting \cite{WZ25a,TWZ25} (cf.\ \cite{CWZ25}) that if a property testing problem requires sample complexity $\Omega\rbra{S}$, then its query complexity is $\Omega\rbra{\sqrt{S}}$. 

In the following, we use $d$ and $r$ to denote the dimension and rank of quantum states, respectively. 
Note that an upper bound in $r$ implies the same upper bound in $d$ by replacing $r$ with $d$, and a lower bound in $d$ implies the same lower bound in $r$ by replacing $d$ with $r$. 

We list nearly optimal quantum algorithms in \cref{tab:list} at the end of this subsection.

\subsubsection{Entropy estimation}

We close all the gaps for quantum entropy estimation up to polylogarithmic factors, including the von Neumann entropy, R\'enyi entropy, and Tsallis entropy. 

\paragraph{Von Neumann entropy estimation.}
Given an unknown $d$-dimensional quantum state $\rho$, the goal is to estimate the von Neumann entropy $\mathrm{S}\rbra{\rho} = -\tr\rbra{\rho \log\rbra{\rho}}$ to within additive error $\varepsilon$.

The investigation of von Neumann entropy estimation dates back to \cite{AISW20} \circled{1}, where they provided a plug-in estimator, also known as the Empirical Young Diagram (EYD) estimator, with sample complexity $O\rbra{\frac{d^2}{\varepsilon^2}}$, using weak Schur sampling \cite{CHW07}.
The analysis of the plug-in estimator was later improved to a sample complexity of $O\rbra{\frac{d^2}{\varepsilon}+\frac{\log^2\rbra{d}}{\varepsilon^2}}$ by \cite{BMW16} \circled{2} (cf.\ \cite[Theorem 1.27]{OW17}).
A time-efficient estimator was proposed in \cite{WZ25b} \circled{3} with sample and time complexity $\widetilde{O}\rbra{\frac{r^2}{\varepsilon^5}}$ for rank-$r$ quantum states; a better sample complexity of $\widetilde{O}\rbra{\frac{r^2}{\varepsilon^4}}$ was given in \cite{CWZ25} \circled{4}. 
Very recently (after this work), an estimator was proposed in \cite{GW26} \circled{5} with sample complexity $O\rbra{\frac{d^2 \log^2\rbra{\log\rbra{d}}\log\rbra{1/\varepsilon}}{\varepsilon^2\log^2\rbra{d}} + \frac{\log^2\rbra{d/\varepsilon}}{\varepsilon^2}}$ subquadratic in $d$, and later improved to $O\rbra{\frac{d^2}{\varepsilon\log^2\rbra{d} \max\cbra{1,  \varepsilon\log^2\rbra{d}}} + \frac{\log^2\rbra{d}}{\varepsilon^2}}$ in \cite{CGW26} \circled{6}, thereby refuting a conjecture in \cite{CWZ25} that the lower bound is $\Omega\rbra{d^2}$.
However, the prior best sample lower bound remains $\Omega\rbra{\frac{d}{\varepsilon} + \frac{\log^2\rbra{d}}{\varepsilon^2}}$ due to \cite{WZ25b} (for the first term) and \cite{JVHW15,WY16} (for the second term).

The query upper bound for von Neumann entropy estimation is $\widetilde{O}\rbra{\frac{d}{\varepsilon^{1.5}}}$ due to \cite{GL20} \circled{7} and $\widetilde{O}\rbra{\frac{r}{\varepsilon^2}}$ for rank-$r$ quantum states due to \cite{WGL+24} \circled{8}, which was recently improved to $\widetilde{O}\rbra{\frac{r}{\varepsilon}}$ in \cite{HLZ+26} \circled{9}. 
The prior best query lower bound is $\Omega\rbra{\frac{\sqrt{d}}{\sqrt{\varepsilon}}+\frac{\log\rbra{d}}{\varepsilon}}$ (cf.\ \cite{CWZ25}); before that, only a query lower bound of $\widetilde{\Omega}\rbra{\sqrt{d}}$ is known due to \cite{BKT20} from the Shannon entropy estimation. 

In this paper, we show (nearly) tight lower bounds for von Neumann entropy estimation. 

\begin{theorem}[\cref{cor:entropy-lower-bound} simplified] \label{thm:von-intro}
    For von Neumann entropy estimation, the sample complexity is ${\Omega}\rbra{\frac{d^2}{\varepsilon\log^2\rbra{d} \max\cbra{1,  \varepsilon\log^2\rbra{d}}} + \frac{\log^2\rbra{d}}{\varepsilon^2}}$ and the query complexity is ${\Omega}\rbra{\frac{d}{\sqrt{\varepsilon}\log\rbra{d} + \varepsilon\log^2\rbra{d}} + \frac{\log\rbra{d}}{\varepsilon}}$.
\end{theorem}

\cref{thm:von-intro} shows that for von Neumann entropy estimation, the sample complexity is exactly $\Theta\rbra{\frac{d^2}{\varepsilon\log^2\rbra{d} \max\cbra{1,  \varepsilon\log^2\rbra{d}}} + \frac{\log^2\rbra{d}}{\varepsilon^2}}$, combined with the upper bound in \cite{CGW26} \circled{6}; the query complexity is $\widetilde{\Theta}\rbra{d}$ for constant $\varepsilon$, combined with the upper bounds in \cite{GL20} \circled{7}, \cite{WGL+24} \circled{8}, and \cite{HLZ+26} \circled{9}.
In addition, \cite{BMW16} \circled{2} is optimal up to polylogarithmic factors; \cite{WZ25b} \circled{3} is nearly sample-optimal and time-optimal in $r$; \cite{AISW20} \circled{1} and \cite{GW26} \circled{5} are nearly sample-optimal in $d$; \cite{CWZ25} \circled{4} is nearly sample-optimal in $r$. 

\paragraph{R\'enyi entropy estimation.}
Given an unknown $d$-dimensional quantum state $\rho$, the goal is to estimate the R\'enyi entropy $\mathrm{S}_{\alpha}^{\mathrm{R}}\rbra{\rho} = \frac{1}{1-\alpha} \log\rbra{\tr\rbra{\rho^\alpha}}$ to within additive error $\varepsilon$.

Plug-in estimators for R\'enyi entropy were provided with sample complexity $O\rbra{\frac{d^{2/\alpha}}{\varepsilon^{2/\alpha}}}$ for $0 < \alpha < 1$ in \cite{AISW20} and $O\rbra{\frac{d^{2}}{\varepsilon^{2}}}$ for non-integer $\alpha > 1$ in \cite{AISW20} \circled{12}.
In particular, sample-optimal estimators are known for integer $\alpha \geq 2$ with sample complexity $\Theta\rbra{\frac{d^{2-2/\alpha}}{\varepsilon^{2/\alpha}}+\frac{d^{1-1/\alpha}}{\varepsilon^2}}$ in \cite{AISW20}. 
Time-efficient estimators were proposed in \cite{WZ25b} with sample and time complexity $\poly\rbra{r/\varepsilon}$; a better sample complexity of $\widetilde{O}\rbra{\frac{r^{\frac{2}{\alpha}}}{\varepsilon^{\frac{2}{\alpha}+2}}}$ for $0 < \alpha < 1$ and $\widetilde{O}\rbra{\frac{r^{2}}{\varepsilon^{2+\frac{2}{\alpha}}}}$ for $\alpha > 1$ was given in \cite{CWZ25} \circled{13}. 
Recently (after this work), estimators were proposed with sample complexity $O\rbra{\frac{d^{1+1/\alpha}}{\varepsilon^{1/\alpha}}+\frac{d^{1/\alpha-1}}{\varepsilon^2}}$ for $0 < \alpha < 1$ in \cite{CW26} \circled{10} and $O\rbra{\frac{d^{2}}{\varepsilon^{1/\alpha}}+\frac{d^{1-1/\alpha}}{\varepsilon^2}}$ for $\alpha > 1$ in \cite{CW26} \circled{14}.
However, the prior best sample lower bound is $\Omega\rbra{\frac{d}{\varepsilon}+\frac{d^{1/\alpha-1}}{\varepsilon^{1/\alpha}}}$ for all constant $\alpha > 0$ due to \cite{WZ25b}. 

Quantum query algorithms for R\'enyi entropy estimation were studied in \cite{SH21} and later a query upper bound of $\poly\rbra{r/\varepsilon}$ was given in \cite{WGL+24} for rank-$r$ quantum states. 
The current best query upper bound is $\widetilde{O}\rbra{\frac{d^{\frac{1}{2\alpha}+\frac{1}{2}}}{\varepsilon^{\frac{1}{2\alpha}+1}}}$ for $0 < \alpha < 1$ in \cite{WZL24} \circled{11} and $\widetilde{O}\rbra{\frac{r}{\varepsilon^{1+\frac{1}{\alpha}}}}$ for $\alpha > 1$ in \cite{WZL24} \circled{15}. 
In particular, recently, an improved query upper bound of $\widetilde{O}\rbra{\frac{r}{\varepsilon}}$ was given in \cite{HLZ+26} \circled{16} for $1 < \alpha \leq 3$.
However, the prior best query lower bound is $\Omega\rbra{d^{\frac{1}{2\alpha}-o\rbra{1}}+\frac{d^{\frac{1}{2\alpha}-\frac{1}{2}}}{\varepsilon^{\frac{1}{2\alpha}}}+\frac{\sqrt{d}}{\sqrt{\varepsilon}}}$ for $0 < \alpha < 1$ and $\Omega\rbra{\frac{d^{\frac{1}{2}-\frac{1}{2\alpha}}}{\varepsilon}+\frac{\sqrt{d}}{\sqrt{\varepsilon}}}$ for non-integer $\alpha > 1$ due to \cite{CWZ25}. 

In this paper, we show nearly tight lower bounds for R\'enyi entropy estimation. 

\begin{theorem}[\cref{cor:renyi-high-accuracy} simplified] \label{thm:ren-intro}
    For R\'enyi entropy estimation, the sample complexity is $\widetilde{\Omega}\rbra{\frac{d^{1+1/\alpha}}{\varepsilon^{1/\alpha}} + \frac{d^{1/\alpha-1}}{\varepsilon^{2}}}$ for $0 < \alpha < 1$ and $\widetilde{\Omega}\rbra{\frac{d^{2}}{\varepsilon^{1/\alpha}} + \frac{d^{1-1/\alpha}}{\varepsilon^{2}}}$ for non-integer $\alpha > 1$; the query complexity is $\widetilde{\Omega}\rbra{\frac{d^{\frac{1}{2}+\frac{1}{2\alpha}}}{\varepsilon^{\frac{1}{2\alpha}}} + \frac{d^{\frac{1}{2\alpha}-\frac{1}{2}}}{\varepsilon}}$ for $0 < \alpha < 1$ and $\widetilde{\Omega}\rbra{\frac{d}{\varepsilon^{\frac{1}{2\alpha}}} + \frac{d^{\frac{1}{2}-\frac{1}{2\alpha}}}{\varepsilon}}$ for non-integer $\alpha > 1$.
\end{theorem}

\cref{thm:ren-intro} shows that for R\'enyi entropy estimation, the sample complexity is $\widetilde{\Theta}\rbra{\frac{d^{1+1/\alpha}}{\varepsilon^{1/\alpha}} + \frac{d^{1/\alpha-1}}{\varepsilon^{2}}}$ for $0 < \alpha < 1$, combined with the upper bound in \cite{CW26} \circled{10}, and $\widetilde{\Theta}\rbra{\frac{d^{2}}{\varepsilon^{1/\alpha}} + \frac{d^{1-1/\alpha}}{\varepsilon^{2}}}$ for non-integer $\alpha > 1$, combined with the upper bound in \cite{CW26} \circled{14}; 
for constant $\varepsilon$, the query complexity is $\widetilde{\Theta}\rbra{d^{\frac{1}{2}+\frac{1}{2\alpha}}}$ for $0 < \alpha < 1$, combined with the upper bound in \cite{WZL24} \circled{11}, and $\widetilde{\Theta}\rbra{r}$ for non-integer $\alpha > 1$, combined with the upper bounds in \cite{WZL24} \circled{15} and \cite{HLZ+26} \circled{16}.
In addition, for non-integer $\alpha > 1$, \cite{AISW20} \circled{12} is nearly sample-optimal in $d$ and \cite{CWZ25} \circled{13} is nearly sample-optimal in $r$. 

\paragraph{Tsallis entropy estimation.}

Given an unknown $d$-dimensional quantum state $\rho$, the goal is to estimate the Tsallis entropy $\mathrm{S}_{\alpha}^{\mathrm{T}}\rbra{\rho} = \frac{1}{1-\alpha} \rbra{\tr\rbra{\rho^\alpha} - 1}$ to within additive error $\varepsilon$.

For integer $\alpha \geq 2$, Tsallis entropy estimation can be done via the generalized SWAP test \cite{BCWdW01,EAO+02} with query complexity $\Theta\rbra{{1}/{\varepsilon}}$ and sample complexity $\Theta\rbra{{1}/{\varepsilon^2}}$, where the lower bound was shown in \cite{CWLY23,GHYZ26}; for asymptotically large integer $\alpha$, the sample complexity is $\Theta\rbra{\frac{1}{\alpha\varepsilon^2}}$ (cf.\ \cite{CWLZ26}) and the query complexity is $\widetilde{\Theta}\rbra{\frac{1}{\sqrt{\alpha}\varepsilon}}$ \cite{Wan25}.
For the non-integer case, quantum algorithms for Tsallis entropy estimation were provided in \cite{WGL+24} with query complexity $\poly\rbra{r/\varepsilon}$ for rank-$r$ quantum states for all non-integer $\alpha > 0$ and $\alpha \neq 1$. 
Later, the query complexity was improved to $O\rbra{1/\varepsilon^{1+\frac{1}{\alpha-1}}}$ for $\alpha > 1$ in \cite{LW26}, removing the $r$-dependence, where the sample complexity was shown to be $\widetilde{O}\rbra{1/\varepsilon^{3+\frac{2}{\alpha-1}}}$ for $\alpha > 1$. 
For non-integer $\alpha > 2$, the sample complexity $\widetilde{\Theta}\rbra{{1}/{\varepsilon^2}}$ was given in \cite{CW25}.
For $1 < \alpha < 2$, the sample upper bound is $\widetilde{O}\rbra{{1}/{\varepsilon^{\frac{2}{\alpha-1}}}}$ due to \cite{CW25} \circled{18}, whereas the prior best lower bound is only $\Omega\rbra{1/\varepsilon^{\frac{1}{\alpha-1}}}$ \cite{CW25}; the query upper bound is $\widetilde{O}\rbra{1/\varepsilon^{\frac{1}{\alpha-1}}}$ due to \cite{HLZ+26} \circled{19}, whereas the prior best lower bound is only $\Omega\rbra{1/\varepsilon^{\frac{1}{2\rbra{\alpha-1}}}}$ by lifting the sample lower bound in \cite{CW25}.
For $0 < \alpha < 1$, the sample complexity was shown to be $O\rbra{\rbra{d/\varepsilon}^{2/\alpha}}$ in \cite{CLW26} and was recently (after this work) improved to $O\rbra{\frac{d^{1+1/\alpha}}{\varepsilon^{1/\alpha}}+\frac{d^{2-2\alpha}}{\varepsilon^2}}$ in \cite{CW26} \circled{17}, whereas the prior best lower bound is $\Omega\rbra{\rbra{d/\varepsilon}^{1/\alpha}}$ \cite{CLW26}. 

In this paper, we show nearly tight lower bounds for Tsallis entropy estimation. 

\begin{theorem}[\cref{cor:tsallis-regimes} simplified] \label{thm:tsa-intro}
    For Tsallis entropy estimation, the sample complexity is $\widetilde{\Omega}\rbra{\frac{d^{1+1/\alpha}}{\varepsilon^{1/\alpha}}+\frac{d^{2-2\alpha}}{\varepsilon^2}}$ for $0 < \alpha < 1$ and $\widetilde{\Omega}\rbra{{1}/{\varepsilon^{\frac{2}{\alpha-1}}}}$ for $1 < \alpha < 2$; the query complexity is $\widetilde{\Omega}\rbra{\frac{d^{1/2+1/2\alpha}}{\varepsilon^{1/2\alpha}}+\frac{d^{1-\alpha}}{\varepsilon}}$ for $0 < \alpha < 1$ and $\widetilde{\Omega}\rbra{{1}/{\varepsilon^{\frac{1}{\alpha-1}}}}$ for $1 < \alpha < 2$.
\end{theorem}

\cref{thm:tsa-intro} shows that for Tsallis entropy estimation, the sample complexity is $\widetilde{\Theta}\rbra{\frac{d^{1+1/\alpha}}{\varepsilon^{1/\alpha}}+\frac{d^{2-2\alpha}}{\varepsilon^2}}$ for $0 < \alpha < 1$, combined with the upper bound in \cite{CW26} \circled{17}, and $\widetilde{\Theta}\rbra{{1}/{\varepsilon^{\frac{2}{\alpha-1}}}}$ for $1 < \alpha < 2$, combined with the upper bound in \cite{CW25} \circled{18};
the query complexity is $\widetilde{\Theta}\rbra{1/\varepsilon^{\frac{1}{\alpha-1}}}$ for $1 < \alpha < 2$, combined with the upper bound in \cite{HLZ+26} \circled{19}.

\subsubsection{Closeness and spectrum estimation}

We close all gaps for quantum closeness estimation (including trace distance and Uhlmann fidelity) and spectrum estimation up to polylogarithmic factors. 

\paragraph{Trace distance estimation.}
Given two unknown $d$-dimensional quantum states $\rho$ and $\sigma$, the goal is to estimate the trace distance $\mathrm{T}\rbra{\rho, \sigma} = \frac{1}{2}\tr\rbra{\abs{\rho - \sigma}}$ to within additive error $\varepsilon$. 

The query complexity was first shown to be $\poly\rbra{r/\varepsilon}$ for rank-$r$ quantum states in \cite{WGL+24}, which was later improved to $O\rbra{\frac{r}{\varepsilon^2}\log\rbra{\frac{1}{\varepsilon}}}$ in \cite{WZ24} \circled{20}. 
However, the prior best query lower bound is $\widetilde{\Omega}\rbra{\sqrt{d} + \frac{1}{\varepsilon}}$ due to \cite{BKT20,Wan24} (cf.\ \cite{CWZ25}).

A time-efficient estimator was proposed in \cite{WZ24} \circled{21} with sample and time complexity $\widetilde{O}\rbra{\frac{r^2}{\varepsilon^5}}$. 
The sample complexity was later improved to $O\rbra{\frac{r^2}{\varepsilon^4}\log^2\rbra{\frac{1}{\varepsilon}}}$ in \cite{CWZ25} \circled{22} and then to $O\rbra{\frac{r^2}{\varepsilon^2}}$ recently in \cite{LT26} \circled{23}. 
However, the prior best sample lower bound is $\Omega\rbra{\frac{d}{\varepsilon^2}}$ due to \cite{OW21} (cf.\ \cite{BOW19}). 

In this paper, we show nearly tight lower bounds for trace distance estimation. 

\begin{theorem}[\cref{corollary:td} simplified] \label{thm:td-intro}
    For trace distance estimation, the sample complexity is $\Omega\rbra{\frac{d^2}{\varepsilon^2\log^2\rbra{d}\max\cbra{1,\varepsilon^2\log^2\rbra{d}}}}$ and the query complexity is $\Omega\rbra{\frac{d}{\varepsilon\log\rbra{d} + \varepsilon^2\log^2\rbra{d}}}$. 
\end{theorem}

\cref{thm:td-intro} shows that for trace distance estimation, the sample complexity is $\widetilde{\Theta}\rbra{\frac{r^2}{\varepsilon^2}}$, combined with the upper bound in \cite{LT26} \circled{23}; the query complexity is $\widetilde{\Theta}\rbra{r}$ for constant $\varepsilon$, combined with the upper bound in \cite{WZ24} \circled{20}. 
Moreover, for constant $\varepsilon$, \cite{WZ24} \circled{21} is both nearly sample-optimal and time-optimal in $r$, and \cite{CWZ25} \circled{22} is nearly sample-optimal in $r$. 

\paragraph{Spectrum estimation.}
Given an unknown $d$-dimensional quantum state $\rho$, the goal is to estimate the nonincreasingly ordered spectrum of $\rho$ to precision $\varepsilon$ in the total variation distance. 

The spectrum estimation task was early studied in \cite{ARS88,KW01,HM02,CM06}.
The sample complexity was analyzed as $O\rbra{\frac{d^2}{\varepsilon^2}\log\rbra{\frac{d}{\varepsilon}}}$ in \cite{OW21} \circled{24}, which was improved to $O\rbra{\frac{d^2}{\varepsilon^2}}$ in \cite{OW17} \circled{25}, recently to $O\rbra{\frac{d^2 \log^2\rbra{\log\rbra{d}}}{\varepsilon^4 \log^2\rbra{d}}}$ in \cite{PSTW26b} \circled{26}, and to $O\rbra{\frac{d^2}{\varepsilon^2\log^2\rbra{d}\max\cbra{1,\varepsilon^2\log^2\rbra{d}}}}$ in \cite{BST26} \circled{27}.
However, the prior best lower bound is $\Omega\rbra{\frac{d}{\varepsilon^2}}$ from the mixedness testing lower bound in \cite{OW21}. 

In this paper, we show a tight lower bound for spectrum estimation through a reduction from trace distance estimation. 

\begin{theorem}[\cref{cor:spectrum-lower-bound} simplified] \label{thm:spectrum-intro}
    For spectrum estimation, the sample complexity is \\ $\Omega\rbra{\frac{d^2}{\varepsilon^2\log^2\rbra{d}\max\cbra{1,\varepsilon^2\log^2\rbra{d}}}}$. 
\end{theorem}

\cref{thm:spectrum-intro} shows that the sample complexity of spectrum estimation is $\Theta\rbra{\frac{d^2}{\varepsilon^2\log^2\rbra{d}\max\cbra{1,\varepsilon^2\log^2\rbra{d}}}}$, combined with the upper bound in \cite{BST26} \circled{27}. 
In addition, the upper bounds in 
\cite{OW21} \circled{24} and \cite{OW17} \circled{25} are optimal up to polylogarithmic factors; the upper bound in \cite{PSTW26b} \circled{26} with subquadratic dependence on $d$ is near-optimal in $d$. 

\paragraph{Uhlmann fidelity estimation.}

Given two unknown $d$-dimensional quantum states $\rho$ and $\sigma$, the goal is to estimate the Uhlmann fidelity $\mathrm{F}\rbra{\rho, \sigma} = \tr\rbra{\sqrt{\sqrt{\sigma}\rho\sqrt{\sigma}}}$ to within additive error $\varepsilon$. 

For pure states, fidelity estimation can be done via the SWAP test \cite{BCWdW01}. 
The query complexity was shown to be $\Theta\rbra{\frac{1}{\varepsilon}}$ in \cite{Wan24,FW25,LW26b}.
The sample complexity was shown to be $\Theta\rbra{\frac{1}{\varepsilon^2}}$ in \cite{WZ26,FW26}, where the lower bound is due to \cite{ALL22}. 

For the general case, the query complexity was first shown to be $\widetilde{O}\rbra{\frac{r^{12.5}}{\varepsilon^{13.5}}}$ in \cite{WZC+23}, which was later improved to $\widetilde{O}\rbra{\frac{r^{6.5}}{\varepsilon^{7.5}}}$ in \cite{WGL+24}, to $\widetilde{O}\rbra{\frac{r^{2.5}}{\varepsilon^{5}}}$ in \cite{GP22}, and to $\widetilde{O}\rbra{\frac{r}{\varepsilon^{2}}}$ in \cite{UNWT25} \circled{28}, where $r$ is the lower rank of $\rho$ and $\sigma$.
However, the query lower bound was shown to be $\widetilde{\Omega}\rbra{d^{1/3} + \frac{1}{\varepsilon}}$ in \cite{UNWT25} by combining the results in \cite{CFMdW10} and \cite{BBC+01,NW99}, which was later improved to $\widetilde{\Omega}\rbra{\sqrt{d} + \frac{1}{\varepsilon}}$ in \cite{CWZ25}, and to $\widetilde{\Omega}\rbra{\frac{\sqrt{r}}{\varepsilon}}$ in \cite{Wan26}. 

In \cite{UNWT25}, they showed that the sample complexity is $O\rbra{\frac{r^2}{\varepsilon^4}\log^2\rbra{\frac{1}{\varepsilon}}}$ \circled{29} and also provided a time-efficient estimator with sample and time complexity $\widetilde{O}\rbra{\frac{r^2}{\varepsilon^5}}$ \circled{30}. 
Recently, \cite{LT26} \circled{31} improved the sample complexity to $O\rbra{\frac{r^2}{\varepsilon^2}}$ when both $\rho$ and $\sigma$ are of rank $r$; \cite{Wan26} \circled{32} improved the sample complexity to $O\rbra{\frac{r^2}{\varepsilon^2}}$ when $\sigma$ of rank $r$ is known in advance. 
However, the sample lower bound is $\Omega\rbra{\frac{r}{\varepsilon}+\frac{1}{\varepsilon^2}}$ due to \cite{OW21,ALL22}, which was recently improved to $\Omega\rbra{\frac{r}{\varepsilon^2}}$ in \cite{Wan26}. 

In this paper, we show nearly tight lower bounds for Uhlmann fidelity estimation. 

\begin{theorem}[\cref{cor:fidelity-lower-bound} simplified] \label{thm:fi-intro}
    For Uhlmann fidelity estimation, the sample complexity is $\Omega\rbra{\frac{d^2}{\varepsilon^2\log^2\rbra{d}\max\cbra{1,\varepsilon^2\log^2\rbra{d}}}}$ and the query complexity is $\Omega\rbra{\frac{d}{\varepsilon\log\rbra{d} + \varepsilon^2\log^2\rbra{d}}}$. 
\end{theorem}

\cref{thm:fi-intro} shows that for Uhlmann fidelity estimation, the sample complexity is $\widetilde{\Theta}\rbra{\frac{r^2}{\varepsilon^2}}$, combined with the upper bound in \cite{LT26} \circled{31} and \cite{Wan26} \circled{32}; the query complexity is $\widetilde{\Theta}\rbra{r}$ for constant $\varepsilon$, combined with the upper bound in \cite{UNWT25} \circled{28}. 
Moreover, for constant $\varepsilon$, \cite{UNWT25} \circled{30} is nearly sample-optimal and time-optimal; \cite{UNWT25} \circled{29} is nearly sample-optimal.

\paragraph{(Squared) Hellinger distance estimation.}
Given two unknown $d$-dimensional quantum states $\rho$ and $\sigma$, the goal is to estimate the squared Hellinger distance $d_\mathrm{H}^2\rbra{\rho, \sigma} = \tr\rbra{\rbra{\sqrt{\rho} - \sqrt{\sigma}}^2}$ to within additive error $\varepsilon$. 

An estimator for (squared) Hellinger distance between two rank-$r$ quantum states was presented in \cite{BGW25} with sample complexity $\widetilde{O}\rbra{\frac{r^{3.5}}{\varepsilon^{10}}}$, which was later improved to $O\rbra{\frac{r^2}{\varepsilon^2}}$ in \cite{LT26} \circled{33}. 

In this paper, we show nearly tight lower bounds for (squared) Hellinger distance estimation. 

\begin{theorem}[\cref{cor:hellinger-lower-bound} simplified] \label{thm:he-intro}
    For (squared) Hellinger distance estimation, the sample complexity is $\Omega\rbra{\frac{d^2}{\varepsilon^2\log^2\rbra{d}\max\cbra{1,\varepsilon^2\log^2\rbra{d}}}}$. 
\end{theorem}

\cref{thm:he-intro} shows that the sample complexity of squared Hellinger distance estimation is $\widetilde{\Theta}\rbra{\frac{r^2}{\varepsilon^2}}$, combined with the upper bound in \cite{LT26} \circled{33}. 

\subsubsection{Rank testing}

We close all gaps for rank testing, Schmidt rank testing, matrix product state testing, and tree tensor network state testing, up to polylogarithmic factors.

\paragraph{Rank testing.}
Given an unknown quantum state $\rho$, the goal is to determine whether $\rho$ is of rank at most $r$ or $\varepsilon$-far in trace distance from any quantum state of rank at most $r$. 

The rank testing task was raised in \cite[Question 8]{MdW16} with a sample lower bound of $\Omega\rbra{r}$ due to \cite{CHW07}.
A tester with sample complexity $O\rbra{\frac{r^2}{\varepsilon}}$ with one-sided error was later provided in \cite{OW21} \circled{34}. 
However, the prior best lower bound is $\Omega\rbra{\frac{r}{\varepsilon}}$ given in \cite{OW21}, where they only showed a matching lower bound of $\Omega\rbra{\frac{r^2}{\varepsilon}}$ for the special case of one-sided error. 

In this paper, we show a nearly tight lower bound for rank testing. 

\begin{theorem}[\cref{cor:rank-lower-bound} simplified] \label{thm:rk-intro}
    For rank testing, the sample complexity is $\Omega\rbra{\frac{r^2}{\varepsilon\log^2\rbra r\max\cbra{1,\varepsilon\log^2\rbra r}}}$.
\end{theorem}

\cref{thm:rk-intro} shows that the sample complexity of rank testing is $\widetilde{\Theta}\rbra{\frac{r^2}{\varepsilon}}$, combined with the upper bound in \cite{OW21} \circled{34}, thereby almost settling the question raised in \cite{MdW16}.

\paragraph{Schmidt rank testing.}
Given an unknown bipartite pure state $\ket{\psi}$, the goal is to determine whether $\ket{\psi}$ has Schmidt rank at most $r$ or $\varepsilon$-far in trace distance from any bipartite pure state of Schmidt rank at most $r$. 

The Schmidt rank testing was raised in \cite[Question 8]{MdW16}. 
The rank tester in \cite{OW21} \circled{35} implies a Schmidt rank tester with sample complexity $O\rbra{\frac{r^2}{\varepsilon^2}}$ with one-sided error (cf.\ \cite{SW22}). 
The hardness of Schmidt rank testing was studied in \cite[Lemma 13]{ABF+24}, assuming the existence of quantum secure one way functions. 
However, the prior best lower bound is $\Omega\rbra{\frac{r}{\varepsilon^2}}$ also due to \cite{CWZ26}, where they only showed a matching lower bound of $\Omega\rbra{\frac{r^2}{\varepsilon^2}}$ for the special case of one-sided error. 

In this paper, we show a nearly tight lower bound for Schmidt rank testing. 

\begin{theorem}[\cref{cor:mps-lower-bound} simplified] \label{thm:schmidt-rk-intro}
    For Schmidt rank testing, the sample complexity is $\Omega\rbra{\frac{r^2}{\varepsilon\log^2\rbra r\max\cbra{1,\varepsilon\log^2\rbra r}}}$.
\end{theorem}

\cref{thm:schmidt-rk-intro} shows that the sample complexity of Schmidt rank testing is $\widetilde{\Theta}\rbra{\frac{r^2}{\varepsilon^2}}$, combined with the upper bound due to \cite{OW21} \circled{35}, thereby almost settling the question raised in \cite{MdW16}.

\paragraph{Matrix product state testing.} 
Given an unknown $n$-partite pure state $\ket{\psi}$, the goal is to determine whether $\ket{\psi}$ is a matrix product state (MPS) of bond dimension $r$ or $\varepsilon$-far in trace distance from any MPS of bond dimension $r$. 

A tester for MPS was proposed in \cite{SW22} \circled{36} with sample complexity $O\rbra{\frac{nr^2}{\varepsilon^2}}$ with one-sided error, where they also showed a lower bound of $\Omega\rbra{\frac{\sqrt{n}}{\varepsilon^2}}$. 
Later, a lower bound of $\Omega\rbra{\sqrt{r}}$ was provided in \cite{ABF+24}. 
The prior best lower bound $\Omega\rbra{\frac{\sqrt{nr}+r}{\varepsilon^2}}$ is due to \cite{CWZ26}.
There are also lower bounds for the special case of one-sided error: a lower bound of $\Omega\rbra{\frac{\sqrt{nr}+r^2}{\varepsilon^2}}$ was given in \cite{CWZ26} and later an almost matching lower bound of $\Omega\rbra{\frac{nr^2}{\varepsilon^2 \log\rbra{n}}}$ was provided in \cite{LL26}.

In this paper, we show a nearly tight lower bound for MPS testing. 

\begin{theorem}[\cref{cor:mps-n-partite-lower-bound} simplified] \label{thm:mps-intro}
    For matrix product state testing, the sample complexity is $\Omega\rbra{\frac{nr^2}{\varepsilon^2 \log^2\rbra{nr}}}$ when $\Omega\rbra{\log\rbra{n}} \leq r \leq \exp\rbra{O\rbra{\sqrt{n}/\varepsilon}}$.
\end{theorem}

\cref{thm:mps-intro} shows that the sample complexity of MPS testing is $\widetilde{\Theta}\rbra{\frac{nr^2}{\varepsilon^2}}$, combined with the upper bound in \cite{SW22} \circled{36}. 
Moreover, this also shows that the sample complexity of tree tensor network state (TTNS) testing is $\widetilde{\Theta}\rbra{\frac{nr^2}{\varepsilon^2}}$, combined with the upper bound $O\rbra{\frac{nr^2}{\varepsilon^2}}$ in \cite{LL26} \circled{37}; this is because the MPS class is a special TTNS class where the tree is a line.

\begin{longtable}{cccccc}
    \caption{List of (nearly) optimal quantum algorithms.}
    \label{tab:list}\\
    \toprule
    Problem & \# & Reference & Type & Complexity & Optimality \\
    \midrule
    \multirow{13}{*}{\begin{tabular}{c}
         Von \\ Neumann
    \end{tabular}} & \circled{1} & \cite{AISW20} & Sample & $O\rbra{\frac{d^2}{\varepsilon^2}}$ & In $d$ \\
    \cmidrule{2-6}
    & \circled{2} & \cite{BMW16} & Sample & $O\rbra{\frac{d^2}{\varepsilon}+\frac{\log^2\rbra{d}}{\varepsilon^2}}$ & Near-Optimal \\
    \cmidrule{2-6}
    & \circled{3} & \cite{WZ25b} & Sample \& Time & $\widetilde{O}\rbra{\frac{r^2}{\varepsilon^5}}$ & In $r$ \\
    \cmidrule{2-6}
    & \circled{4} & \cite{CWZ25} & Sample & $\widetilde{O}\rbra{\frac{r^2}{\varepsilon^4}}$ & In $r$ \\
    \cmidrule{2-6}
    & \circled{5} & \cite{GW26} & Sample & $\widetilde{O}\rbra{\frac{d^{2-o\rbra{1}}}{\varepsilon^2}}$ & In $d$ \\
    \cmidrule{2-6}
    & \circled{6} & \cite{CGW26} & Sample & ${O}\rbra{\frac{d^2}{\varepsilon\log^2\rbra{d} + \varepsilon^2\log^4\rbra{d}} + \frac{\log^2\rbra{d}}{\varepsilon^2}}$ & Optimal \\
    \cmidrule{2-6}
    & \circled{7} & \cite{GL20} & Query & $\widetilde{O}\rbra{\frac{d}{\varepsilon^{1.5}}}$ & In $d$ \\
    \cmidrule{2-6}
    & \circled{8} & \cite{WGL+24} & Query & $\widetilde{O}\rbra{\frac{r}{\varepsilon^{2}}}$ & In $r$ \\
    \cmidrule{2-6}
    & \circled{9} & \cite{HLZ+26} & Query & $\widetilde{O}\rbra{\frac{r}{\varepsilon}}$ & In $r$ \\
    \midrule 
    \multirow{3}{*}{\begin{tabular}{c}
         R\'enyi \\ ($0 < \alpha < 1$)
    \end{tabular}} & \circled{10} & \cite{CW26} & Sample & $O\rbra{\frac{d^{1+1/\alpha}}{\varepsilon^{1/\alpha}}+\frac{d^{1/\alpha-1}}{\varepsilon^2}}$ & Near-Optimal \\
    \cmidrule{2-6}
    & \circled{11} & \cite{WZL24} & Query & $\widetilde{O}\rbra{\frac{d^{\frac{1}{2\alpha}+\frac{1}{2}}}{\varepsilon^{\frac{1}{2\alpha}+1}}}$ & In $d$ \\
    \midrule
    \multirow{8}{*}{\begin{tabular}{c}
         R\'enyi \\ (non-integer \\ $\alpha > 1$)
    \end{tabular}} & \circled{12} & \cite{AISW20} & Sample & $\widetilde{O}\rbra{\frac{d^2}{\varepsilon^2}}$ & In $d$ \\
    \cmidrule{2-6}
    & \circled{13} & \cite{CWZ25} & Sample & $\widetilde{O}\rbra{\frac{r^2}{\varepsilon^{\frac{2}{\alpha}+2}}}$ & In $r$ \\
    \cmidrule{2-6}
    & \circled{14} & \cite{CW26} & Sample & $O\rbra{\frac{d^{2}}{\varepsilon^{1/\alpha}}+\frac{d^{1-1/\alpha}}{\varepsilon^2}}$ & Near-Optimal \\
    \cmidrule{2-6}
    & \circled{15} & \cite{WZL24} & Query & $\widetilde{O}\rbra{\frac{r}{\varepsilon^{1+\frac{1}{\alpha}}}}$ & In $r$ \\
    \cmidrule{2-6}
    & \circled{16} & \cite{HLZ+26} & Query & $\widetilde{O}\rbra{\frac{r}{\varepsilon}}$ for $1 < \alpha < 3$ & In $r$ \\
    \midrule 
    \begin{tabular}{c}
         Tsallis \\ ($0 < \alpha < 1$)
    \end{tabular} & \circled{17} & \cite{CW26} & Sample & $O\rbra{\frac{d^{1+1/\alpha}}{\varepsilon^{1/\alpha}}+\frac{d^{2-2\alpha}}{\varepsilon^2}}$ & Near-Optimal \\
    \midrule 
    \multirow{3}{*}{\begin{tabular}{c}
         Tsallis \\ ($1 < \alpha < 2$)
    \end{tabular}} & \circled{18} & \cite{CW25} & Sample & $\widetilde{O}\rbra{1/\varepsilon^{\frac{2}{\alpha-1}}}$ & Near-Optimal \\
    \cmidrule{2-6}
    & \circled{19} & \cite{HLZ+26} & Query & $\widetilde{O}\rbra{1/\varepsilon^{\frac{1}{\alpha-1}}}$ & Near-Optimal \\
    \midrule 
    \multirow{6}{*}{\begin{tabular}{c}
         Trace \\ Distance
    \end{tabular}} & \circled{20} & \cite{WZ24} & Query & $\widetilde{O}\rbra{\frac{r}{\varepsilon^2}}$ & In $r$ \\
    \cmidrule{2-6}
    & \circled{21} & \cite{WZ24} & Sample \& Time & $\widetilde{O}\rbra{\frac{r^2}{\varepsilon^{5}}}$ & In $r$ \\
    \cmidrule{2-6}
    & \circled{22} & \cite{CWZ25} & Sample & $\widetilde{O}\rbra{\frac{r^2}{\varepsilon^{4}}}$ & In $r$ \\
    \cmidrule{2-6}
    & \circled{23} & \cite{LT26} & Sample & $O\rbra{\frac{r^2}{\varepsilon^{2}}}$ & Near-Optimal \\
    \midrule 
    \multirow{6}{*}{\begin{tabular}{c}
         Spectrum
    \end{tabular}} & \circled{24} & \cite{OW21} & Sample & $O\rbra{\frac{d^2}{\varepsilon^2}\log\rbra{\frac{d}{\varepsilon}}}$ & Near-Optimal \\
    \cmidrule{2-6}
    & \circled{25} & \cite{OW17} & Sample & ${O}\rbra{\frac{d^2}{\varepsilon^{2}}}$ & Near-Optimal \\
    \cmidrule{2-6}
    & \circled{26} & \cite{PSTW26b} & Sample & $O\rbra{\frac{d^2 \log^2\rbra{\log\rbra{d}}}{\varepsilon^4 \log^2\rbra{d}}}$ & In $d$ \\
    \cmidrule{2-6}
    & \circled{27} & \cite{BST26} & Sample & $O\rbra{\frac{d^2}{\varepsilon^2\log^2\rbra{d} + \varepsilon^4\log^4\rbra{d}}}$ & Optimal \\
    \midrule 
    \multirow{8}{*}{\begin{tabular}{c}
         Uhlmann \\ Fidelity
    \end{tabular}} & \circled{28} & \cite{UNWT25} & Query & $\widetilde{O}\rbra{\frac{r}{\varepsilon^2}}$ & In $r$ \\
    \cmidrule{2-6}
    & \circled{29} & \cite{UNWT25} & Sample & $\widetilde{O}\rbra{\frac{r^2}{\varepsilon^{4}}}$ & In $r$ \\
    \cmidrule{2-6}
    & \circled{30} & \cite{UNWT25} & Sample \& Time & $\widetilde{O}\rbra{\frac{r^2}{\varepsilon^{5}}}$ & In $r$ \\
    \cmidrule{2-6}
    & \circled{31} & \cite{LT26} & Sample & $O\rbra{\frac{r^2}{\varepsilon^{2}}}$ & Near-Optimal \\
    \cmidrule{2-6}
    & \circled{32} & \cite{Wan26} & Sample & $O\rbra{\frac{r^2}{\varepsilon^{2}}}$ & Near-Optimal \\
    \midrule 
    \begin{tabular}{c}
         (Squared) \\ Hellinger
    \end{tabular} & \circled{33} & \cite{LT26} & Sample & $O\rbra{\frac{r^2}{\varepsilon^2}}$ & Near-Optimal \\
    \midrule 
    \begin{tabular}{c}
         Rank \\ Testing
    \end{tabular} & \circled{34} & \cite{OW21} & Sample & $O\rbra{\frac{r^2}{\varepsilon}}$ & Near-Optimal \\
    \midrule 
    \begin{tabular}{c}
         Schmidt \\ Rank
    \end{tabular} & \circled{35} & \cite{OW21} & Sample & $O\rbra{\frac{r^2}{\varepsilon^2}}$ & Near-Optimal \\
    \midrule 
    \multirow{1}{*}{\begin{tabular}{c}
         MPS
    \end{tabular}} & \circled{36} & \cite{SW22} & Sample & $O\rbra{\frac{nr^2}{\varepsilon^2}}$ & Near-Optimal \\
    \midrule 
    \multirow{1}{*}{\begin{tabular}{c}
         TTNS
    \end{tabular}} & \circled{37} & \cite{LL26} & Sample & $O\rbra{\frac{nr^2}{\varepsilon^2}}$ & Near-Optimal \\
    \bottomrule
\end{longtable}

\subsection{Techniques} \label{sec:tech}

\paragraph{Step 1: Haar-random one-direction moment encoding.} 
In the classical literature of property testing (cf.\ \cite{JVHW15,WY16,AOST17}), the hard instance is usually chosen to be the discrete distribution of the form
\[
p_Z = \rbra*{ \frac{1+Z}{D}, \underbrace{\frac{1 - \dfrac{Z}{D-1}}{D}, \dots, \frac{1 - \dfrac{Z}{D-1}}{D}}_{D-1 \textup{ equal entries}} },
\]
where $Z \in \sbra{-1, D-1}$ is a parameter. 
However, directly using this type of hard instance does not give strong enough lower bounds (cf.\ \cite{WZL24,WZ25b}). 

To overcome this issue, we strengthen the hard instance with a Haar-random direction. 
Specifically, let $\ket{u}$ be a Haar-random state, and consider the quantum state 
\begin{equation*} 
    \varsigma_{Z, u} = \frac{1+Z}{D} \ketbra{u}{u} + \frac{1 - \dfrac{Z}{D-1}}{D} \rbra*{I - \ketbra{u}{u}}.
\end{equation*}
Note that the spectrum of $\varsigma_{Z, u}$ is exactly $p_Z$. 
Inspired by the moment matching arguments in the classical literature (cf.\ \cite{JVHW17,WY16}), we aim to encode the moment matching into a random quantum state $\varsigma_{Z, u}$, with $\ket{u}$ a Haar-random state and $Z$ distributed from distributions with matching moments. 
Specifically, suppose $\rbra{\nu_0, \nu_1}$ is an $L$-moment matching on $\sbra{-B, B}$, and consider the $n$-copy average state $\sigma_b^{\rbra{n}} = \E_{Z \sim \nu_b, \ket{u} \sim \textup{Haar}} \sbra{\varsigma_{Z, u}^{\otimes n}}$ for each $b \in \cbra{0, 1}$. 
One can show that $\sigma_0^{\rbra{n}}$ and $\sigma_1^{\rbra{n}}$ are indistinguishable when $n = O\rbra{D^2/B^2}$ (see \cref{lemma:q-moment-matching}). 

\paragraph{Step 2: Haar-random multi-direction moment encoding.}
Inspired by the indistinguishability revealed in Step 1, we further generalize it to multiple Haar-random directions. 
Specifically, let $\ket{u_i}$ be orthogonal Haar-random states, and for each $b \in \cbra{0, 1}$, consider the quantum state
\begin{equation} \label{eq:multi-direct}
\rho_b = \frac{1}{d} \sum_{i=1}^q X_i \ketbra{u_i}{u_i} + \frac{d-S}{rd} \rbra*{ \mathbb{I} - \sum_{i=1}^q \ketbra{u_i}{u_i} },
\end{equation}
where $X_i \sim \nu_b$, $d/6 \leq q \leq d/3$, $S = \sum_{i=1}^q X_i$, and $r = d - q$. 
In particular, we define $\rho_b = \mathbb{I}/d$ if $S > d/2$. 
Let $\Pi_b$ be the distribution of $\rho_b$ and consider the $n$-copy average state $\Omega_{b}^{\rbra{n}} = \E_{\rho \sim \Pi_b}\sbra{\rho^{\otimes n}}$. 
One can show that $\Omega_{0}^{\rbra{n}}$ and $\Omega_{1}^{\rbra{n}}$ are indistinguishable when $n = O\rbra{d^2/M^2}$, where $M = \sup_{x \in \supp\rbra{\nu_0} \cup \supp\rbra{\nu_1}} \abs{x-1}$ (see \cref{lemma:q-moment-matching-multi-direction}). 
The proof reduces to the one-direction moment encoding via a hybrid argument. 

The key construction is the hard instance defined in \cref{eq:multi-direct}. 

\paragraph{Step 3: A master theorem.}
To establish our master theorem, note that on the event $S \leq d/2$ with high probability, the eigenvalues of $\rho_b$ are 
\[
    \frac{X_1}{d}, \dots, \frac{X_q}{d}, \underbrace{\frac{d-S}{rd}, \dots, \frac{d-S}{rd}}_{r\textup{ equal entries}}.
\]
For any function $\phi$, the corresponding functional is
\[
\mathcal{L}_\phi\rbra{\rho_b} = \frac{1}{d} \sum_{i=1}^q \phi\rbra{X_i} + \frac{r}{d} \phi\rbra*{\frac{d-S}{r}},
\]
which can be shown to concentrate at 
\[
\theta_b = \frac{q}{d} \E_{X \sim \nu_b} \sbra*{\phi\rbra{X}} + \frac{r}{d} \phi\rbra*{\frac{d-q \E_{X \sim \nu_b}\sbra{X}}{r}}.
\]
Then, one can take $I_b = \sbra{\theta_b - \frac{qg_\phi}{8d}, \theta_b + \frac{qg_\phi}{8d}}$ and verify that under the assumptions required in \cref{thm:main-intro}, $\Omega_0^{\rbra{n}}$ and $\Omega_1^{\rbra{n}}$ are indistinguishable for $n = O\rbra{d^2/M^2}$ while $\mathcal{L}_\phi\rbra{\rho} \in I_b$ are distinguishable. 
See \cref{thm:master-functional} for more detail and further generalizations. 

\paragraph{An illustrative application: trace distance estimation and spectrum estimation.}

For convenience of the readers, in \cref{tab1}, we collect the choice of $\phi$ for the problems considered in this paper. 

\begin{table}[h]
    \centering
    \caption{The choice of $\phi$ for different problems} 
    \label{tab1}
    \begin{tabular}{ccc}
        \toprule
        Problem & Functional & $\phi$ \\
        \midrule
        Von Neumann Entropy Estimation & $\mathrm{S}\rbra{\rho}$ & $-x\log\rbra{x}$ \\
        $\alpha$-R\'enyi/Tsallis Entropy Estimation & $\mathrm{S}_\alpha^{\mathrm{R}}\rbra{\rho}$ \& $\mathrm{S}_\alpha^{\mathrm{T}}\rbra{\rho}$ & $x^\alpha$ \\
        Trace Distance Estimation & $\mathrm{T}\rbra{\rho, \mathbb{I}/d}$ & $\max\cbra{1-x,0}$ \\
        Uhlmann Fidelity Estimation & $\mathrm{F}\rbra{\rho, \mathbb{I}/d}$ & $\sqrt{x}$ \\
        Rank Testing & $\rank\rbra{\rho}$ & $\mathbf{1}_{\{x > 0\}}$ \\
        \bottomrule
    \end{tabular}
\end{table}

As an illustrative example, we show how to apply our master theorem to derive lower bounds for the problem of trace distance estimation. 
Here, we take $\phi\rbra{x} = \max\cbra{1-x,0}$ and note that $\mathcal{L}_{\max\cbra{1-x, 0}}\rbra{\rho} = \mathrm{T}\rbra{\rho, \mathbb{I}/d}$.
Based on the best polynomial approximation lower bound in \cite{JHW18}, we can construct a $\rbra{K+1}$-moment matching $\rbra{\nu_0, \nu_1}$ on $\sbra{1-\frac{3t}{4}, 1 + \frac{3t\rbra{1-a}}{4a}}$ with mean $\leq 1$ for $K \geq 2$, $1/K^2 \leq a \leq 1/2$, and $t = \Theta\rbra{K \sqrt{a} \varepsilon}$. 
Under the choice of $\rbra{\nu_0, \nu_1}$, one can verify that $A_{\phi} \leq \frac{3t}{4}$, $L_{\phi} \leq 1$, $g_\phi \geq 24\varepsilon$, and $M \leq \frac{3t}{4a} = b_0 - a_0$. 
With $K = \Theta\rbra{\log\rbra{d}}$ and choosing the parameter $a$ appropriately, we can then apply \cref{thm:main-intro} and obtain a sample complexity lower bound of $\Omega\rbra{d^2/M^2} = \widetilde{\Omega}\rbra{d^2/\varepsilon^2}$ for trace distance estimation, even if one quantum state is the maximally mixed state known in advance. 

As a direct implication, since spectrum estimation can be straightforwardly used to estimate $\mathrm{T}\rbra{\rho, \mathbb{I}/d}$, the lower bound $\widetilde{\Omega}\rbra{d^2/\varepsilon^2}$ for trace distance estimation is also a lower bound for spectrum estimation. 

\subsection{Independent and Simultaneous Work}

The independent and simultaneous works \cite{FOW26,LJ26} also considered some of the problems investigated in this work.
In comparison, our lower bounds are better than their stated lower bounds. 
Specifically,
\begin{itemize}
    \item Von Neumann entropy estimation (compared to \cref{thm:von-intro}): In \cite{FOW26}, they showed a sample lower bound of $\Omega\rbra{d^{2-\gamma}}$ for any constant $\gamma > 0$. 
    In comparison, our lower bound $\Omega\rbra{d^2/\log^4\rbra{d}}$ is better than their lower bound in the same regime of constant $\varepsilon$. 
    Moreover, our lower bound also considers the $\varepsilon$-dependence. 

    \item Spectrum estimation (compared to \cref{thm:spectrum-intro}): In \cite{FOW26}, they showed a sample lower bound of $\Omega\rbra{d^{2-\gamma}}$ for any constant $\gamma > 0$; in \cite{LJ26}, they showed a sample lower bound of $\Omega\rbra{{d^2}/{\log^C\rbra{d}}}$ for some constant $C > 0$, where a detailed calculation will give $C > 30$. 
    In comparison, our lower bound $\Omega\rbra{d^2/\log^4\rbra{d}}$ is better than theirs in the same regime of constant $\varepsilon$. 
    Moreover, our lower bound also considers the $\varepsilon$-dependence. 

    \item Uhlmann fidelity estimation (compared to \cref{thm:fi-intro}): In \cite{LJ26}, they showed a sample lower bound of $\Omega\rbra{d^2/\rbra{\varepsilon^2 \log^C\rbra{d}}}$. 
    In comparison our lower bound $\Omega\rbra{d^2/\rbra{\varepsilon^2 \log^4\rbra{d}}}$ is better than theirs in all regimes. 

    \item Rank testing (compared to \cref{thm:rk-intro}): In \cite{FOW26}, they showed a sample lower bound of $\Omega\rbra{r^{2-\gamma}}$ for any constant $\gamma > 0$; in \cite{LJ26}, they showed a sample lower bound of $\Omega\rbra{r^2/\log^C\rbra{r}}$. 
    In comparison, our lower bound $\Omega\rbra{r^2/\log^4\rbra{r}}$ is better than theirs in the same regime of constant $\varepsilon$. 
    Moreover, our lower bound also considers the $\varepsilon$-dependence. 
\end{itemize}

\subsection{Recent Developments}

After the work described in this paper, there have been (near-)optimal estimators proposed. 

\begin{itemize}
    \item In \cite{GW26}, they proposed a von Neumann entropy estimator with sample complexity $O\rbra{\frac{d^2 \log^2\rbra{\log\rbra{d}}\log\rbra{1/\varepsilon}}{\varepsilon^2\log^2\rbra{d}} + \frac{\log^2\rbra{d/\varepsilon}}{\varepsilon^2}}$, which was later improved to the optimal $\Theta\rbra{\frac{d^2}{\varepsilon\log^2\rbra{d} \max\cbra{1,  \varepsilon\log^2\rbra{d}}} + \frac{\log^2\rbra{d}}{\varepsilon^2}}$ in \cite{CGW26}. 
    \item In \cite{CW26}, they proposed a R\'enyi entropy estimator with sample complexity $O\rbra{\frac{d^{1+1/\alpha}}{\varepsilon^{1/\alpha}}+\frac{d^{1/\alpha-1}}{\varepsilon^2}}$ for $0 < \alpha < 1$ and $O\rbra{\frac{d^{2}}{\varepsilon^{1/\alpha}}+\frac{d^{1-1/\alpha}}{\varepsilon^2}}$ for $\alpha > 1$, and a Tsallis entropy estimator with sample complexity $O\rbra{\frac{d^{1+1/\alpha}}{\varepsilon^{1/\alpha}}+\frac{d^{2-2\alpha}}{\varepsilon^2}}$ for $0 < \alpha < 1$.
    \item In \cite{BST26}, they proposed a spectrum estimator with optimal sample complexity ${\Theta}\rbra{\frac{d^2}{\varepsilon^2\log^2\rbra{d} \max\cbra{1,  \varepsilon^2\log^2\rbra{d}}}}$.
    \item In \cite{HLZ+26}, they proposed quantum query algorithms for von Neumann entropy estimation with query complexity $\widetilde{O}\rbra{\frac{r}{\varepsilon}}$, for Tsallis entropy estimation of order $\alpha$ with query complexity $\widetilde{O}\rbra{1/\varepsilon^{\frac{1}{\alpha-1}}}$ for $1 < \alpha < 2$, and for R\'enyi entropy estimation of order $\alpha$ with query complexity $\widetilde{O}\rbra{\frac{r}{\varepsilon}}$ for $1 < \alpha \leq 3$. 
\end{itemize}

\subsection{AI Use Disclosure}

The author used Large Language Models as AI-assisted research and writing tools throughout the
preparation of this paper. The tools were used to help brainstorm ideas and explore proof
strategies. Portions of the paper text were redrafted or modified with AI assistance across all
sections. All final mathematical claims, algorithms, proofs, citations, and wording were reviewed,
edited, and validated by the author. The author assumes responsibility for all content of the paper.

\section{Preliminaries}

\subsection{Moment Matching}

\begin{definition}[Moment matching]
    For two probability distributions $\nu_0$ and $\nu_1$ on $\mathbb{R}$, $\rbra{\nu_0, \nu_1}$ is said to be a $K$-moment matching, if
    \[
    \E_{X \sim \nu_0}\sbra*{X^j} = \E_{X \sim \nu_1}\sbra*{X^j}, \qquad j \in \cbra{0, 1, 2, \dots, K}.
    \]
\end{definition}

\begin{lemma}[Cf.\ {\cite[Lemma 10]{JVHW15}}, {\cite[Appendix E]{WY16}}, and {\cite[Lemma 7]{CC19}}]\label{fact:mm}
For every continuous \(f\) on a compact interval \(I\), there is a $K$-moment matching \(\rbra{\nu_0,\nu_1}\) on \(I\) such that
\begin{align*}
    \int_I f\rbra x\,\mathrm d\nu_1\rbra x
    -\int_I f\rbra x\,\mathrm d\nu_0\rbra x
    =2E_K\rbra{f,I},
\end{align*}
where
\[
    E_K\rbra{f,I}
    =\inf_{\deg\rbra{P}\le K}\sup_{x\in I}\abs{f\rbra x-P\rbra x}.
\]
\end{lemma}

\begin{lemma}[Moment lifting]\label{lemma:T-eta-alpha}\label{lemma:moment-lifting}
Fix \(0<\eta\le1\) and \(a>0\).  Let \(X\sim\nu\) on \([\eta,1]\), and define \(Y\) conditionally on \(X\) by
\[
    Y=
    \begin{cases}
        aX/\eta,&\text{with probability }\eta/X,\\
        0,&\text{with probability }1-\eta/X.
    \end{cases}
\]
Write \(\mathsf T_{\eta,a}\rbra\nu\) for the probability distribution of \(Y\).  Then, for every integer \(j\ge1\),
\begin{equation}
    \E\sbra{Y^j}
    =a^j\eta^{1-j}\E\sbra{X^{j-1}}.
    \label{eq:change-measure-moments}
\end{equation}
More generally, for every function \(h\) for which the expectations exist,
\begin{equation}
    \E\sbra{h\rbra Y}
    =h\rbra0+
    \eta\E\sbra*{
    \frac{h\rbra{aX/\eta}-h\rbra0}{X}}.
    \label{eq:change-measure-functional}
\end{equation}
Consequently, if \(\rbra{\nu_0,\nu_1}\) is a $K$-moment matching, then \(\rbra{\mathsf T_{\eta,a}\rbra{\nu_0}, \mathsf T_{\eta,a}\rbra{\nu_1}}\) is a $\rbra{K+1}$-moment matching with mean \(a\).
\end{lemma}

\begin{proof}
Conditioning on \(X\) gives
\[
\begin{aligned}
    \E\sbra{Y^j}
    &=\E\sbra*{\frac\eta X\left(\frac{aX}{\eta}\right)^j}
      =a^j\eta^{1-j}\E\sbra{X^{j-1}},
\end{aligned}
\]
which is \cref{eq:change-measure-moments}.  The same conditioning gives
\[
\begin{aligned}
    \E\sbra{h\rbra Y}
    &=\E\sbra*{\frac\eta Xh\rbra{aX/\eta}
       +\left(1-\frac\eta X\right)h\rbra0}\\
    &=h\rbra0+
      \eta\E\sbra*{\frac{h\rbra{aX/\eta}-h\rbra0}{X}},
\end{aligned}
\]
proving \cref{eq:change-measure-functional}.  The final assertion follows from \cref{eq:change-measure-moments}.
\end{proof}

\subsection{Haar-Random Bounds}

\begin{lemma}
\label{lemma:haar-overlap-moments}
Let \(\ket{u}, \ket{v}\) be independent Haar-random unit vectors in \(\mathbb C^D\).
Then, $\E\sbra{\abs{D \abs{\braket{u}{v}}^2 - 1}^k} \leq 2^kk!$ for every integer \(k\ge0\).
\end{lemma}

\begin{proof}
Note that \(S \coloneqq \abs{\braket{u}{v}}^2\) obeys the beta distribution \(\operatorname{Beta}(1,D-1)\) (cf.\ {\cite[Equations (7)--(9)]{ZS00}}), where the probability density function of $\operatorname{Beta}\rbra{\alpha, \beta}$ is defined by
\[
f\rbra{x; \alpha, \beta} = \frac{1}{\operatorname{B}\rbra{\alpha, \beta}} x^{\alpha - 1} \rbra{1 - x}^{\beta - 1}, \quad 0 \leq x \leq 1, \quad \operatorname{B}\rbra{\alpha, \beta} = \int_0^1 t^{\alpha-1}\rbra{1-t}^{\beta-1}\mathrm{d}t.
\]
Therefore, cf.\ \cite[Equation (16.8)]{BN03},
\[
\E\sbra{S^k} = \frac{k!\rbra{D-1}!}{\rbra{D+k-1}!} \leq \frac{k!}{D^k}.
\]
Finally, using $\abs{x-y}^k \leq 2^{k-1} \rbra{\abs{x}^k + \abs{y}^k}$, we have 
\[
\E\sbra*{\abs[\big]{D \abs{\braket{u}{v}}^2 - 1}^k} \leq 2^{k-1} \rbra*{ \E\sbra{\rbra{DS}^k} + 1 } \leq 2^{k-1}\rbra*{ k! + 1 } \leq 2^k k!.
\]
\end{proof}

\section{Haar-Random Moment Encoding}

In this section, we present an approach to encode a moment matching into a quantum state using Haar-randomness. 
First, we encode with one Haar-random direction in \cref{sec:one-direct}. 
Then, we generalize this idea to multiple Haar-random directions in \cref{sec:mm-qfe}.

\subsection{Haar-Random One-Direction Moment Encoding} \label{sec:one-direct}

We consider the $D$-dimensional quantum state $\varsigma_{z, u}$ defined as follows. 

\begin{definition}[Haar-random one-direction encoding] \label{def:varsigma}
    For an integer \(D\ge2\), a unit vector \(\ket u\in\mathbb C^D\), and a real number \(-1\le z\le D-1\), define
\begin{equation*}
    \varsigma_{z,u}=\frac{1+z}{D}\ketbra u u +\frac{1-\dfrac{z}{D-1}}{D} \rbra*{\mathbb{I}-\ketbra u u}.
\end{equation*}
\end{definition}

Note that $\varsigma_{z, u}$ has non-negative eigenvalues \begin{equation*}
    \rbra*{
    \frac{1+z}{D},
    \underbrace{
    \frac{1-\dfrac{z}{D-1}}{D},\ldots,
    \frac{1-\dfrac{z}{D-1}}{D}}_{D-1\textup{ equal eigenvalues}}}
\end{equation*}
and $\tr\rbra{\varsigma_{z, u}} = 1$. 
Therefore, $\varsigma_{z, u}$ is a valid quantum state (i.e., density operator). 

The following lemma shows that if $z$ is randomly distributed from either distribution of a moment matching, then $\varsigma_{z, u}$ can be indistinguishable under some conditions depending on the dimension $D$, the moment matching, and the number $n$ of samples of $\varsigma_{z, u}$.

\begin{lemma}[Haar-random one-direction moment encoding]
\label{lemma:q-moment-matching}
Let \(\rbra{\nu_0,\nu_1}\) be an $L$-moment matching on $[-1,D-1]\cap[-B,B]$.
Define
\[
    \sigma_b^{\rbra n}
    =\E_{\substack{Z\sim\nu_b\\\ket{u}\sim\operatorname{Haar}\rbra{\mathbb{C}^D}}}
      \sbra*{\varsigma_{Z,u}^{\otimes n}},
    \qquad b\in\{0,1\}.
\]
If $R \coloneqq 2nB^2/\rbra{D-1}^2 < 1$, then 
\begin{equation*}
    \mathrm T\rbra*{\sigma_0^{\rbra n},\sigma_1^{\rbra n}}
    \le\sqrt{\frac{R^{L+1}}{1-R}}.
\end{equation*}
In particular, if \(R\le1/16\), then $\mathrm T\rbra{\sigma_0^{\rbra n},\sigma_1^{\rbra n}} \le4^{-L}$.
\end{lemma}

\begin{proof}
Using the fact that $\Abs{A}_1 \leq \sqrt{d} \Abs{A}_2$ for any $d$-dimensional operator $A$, we have
\begin{equation}
    \mathrm T\rbra*{\sigma_0^{\rbra n},\sigma_1^{\rbra n}}
    \leq\frac{\sqrt{D^n}}2
    \sqrt{\tr\rbra*{\rbra*{\sigma_0^{\rbra n}-\sigma_1^{\rbra n}}^2}}.
    \label{eq:one-direction-trace-to-hs}
\end{equation}
The definition of $\varsigma_{z,u}$ can be rewritten as
\[
    \varsigma_{z,u}
    =\frac {\mathbb{I}}D+\frac z{D-1}\rbra*{\ketbra u u-\frac {\mathbb{I}}D}.
\]
Thus, for unit vectors $\ket u$ and $\ket v$, direct calculation shows
\begin{align*}
    \tr\rbra*{\varsigma_{z,u}\varsigma_{z',v}}
    =\frac1D
      +\frac{zz'}{D(D-1)^2}
       X, \qquad X = D\abs*{\braket{u}{v}}^2 - 1.
\end{align*}
Therefore, we have
\begin{align}
\tr\rbra*{\rbra*{\sigma_0^{\rbra n}-\sigma_1^{\rbra n}}^2}
&=\frac1{D^n}\sum_{k=0}^n\binom nk\frac1{(D-1)^{2k}}
\rbra*{
    \E_{Z\sim\nu_0}\sbra*{Z^k}
    -\E_{Z\sim\nu_1}\sbra*{Z^k}
}^2
\E\sbra*{X^k} \nonumber \\
& = \frac1{D^n}\sum_{k=L+1}^n\binom nk\frac1{(D-1)^{2k}}
\rbra*{
    \E_{Z\sim\nu_0}\sbra*{Z^k}
    -\E_{Z\sim\nu_1}\sbra*{Z^k}
}^2
\E\sbra*{X^k}, \label{eq:tr-sigma-diff-1}
\end{align}
where the second equality uses that $\rbra{\nu_0, \nu_1}$ is an $L$-moment matching. 

Since both distributions are defined on $\sbra{-1, D-1} \cap \sbra{-B,B}$, the triangle inequality gives
\begin{equation}
    \abs*{
      \E_{Z\sim\nu_0}\sbra*{Z^k}
      -\E_{Z\sim\nu_1}\sbra*{Z^k}}
    \leq
    \E_{Z\sim\nu_0}\sbra*{\abs{Z}^k}
    +\E_{Z\sim\nu_1}\sbra*{\abs{Z}^k}
    \leq2B^k,
\end{equation}
which, together with \cref{lemma:haar-overlap-moments}, implies 
\begin{align}
\eqref{eq:tr-sigma-diff-1}
& \leq \frac1{D^n}\sum_{k=L+1}^n\binom nk\frac1{(D-1)^{2k}}
\rbra*{2B^k}^2
\rbra*{2^k k!} \nonumber \\
& = \frac{4}{D^n} \sum_{k=L+1}^n\binom nk k! \rbra*{\frac{2B^2}{(D-1)^{2}}}^k \nonumber \\
& \leq \frac{4}{D^n} \sum_{k=L+1}^n R^k \label{eq:tr-sigma-diff-2} \\
& \leq \frac{4 R^{L+1}}{D^n \rbra{1-R}}, \label{eq:tr-sigma-diff-3}
\end{align}
where \cref{eq:tr-sigma-diff-2} uses the fact that $\binom{n}{k} k! \leq n^k$ and \cref{eq:tr-sigma-diff-3} uses the condition that $R < 1$. 
Finally, combining \cref{eq:one-direction-trace-to-hs,eq:tr-sigma-diff-3} yields the proof.
\end{proof}

\subsection{Haar-Random Multi-Direction Moment Encoding}\label{sec:mm-qfe}

We now extend the one-direction encoding in \cref{def:varsigma} to multiple Haar-random directions through the following process. 

\begin{definition} \label{def:pi-01}
    Fix an integer $d \geq 2$. 
    For an $L$-moment matching $\rbra{\nu_0, \nu_1}$ on $I = \sbra{a_0, b_0} \subseteq \sbra{0, d/12}$ with mean $\nu \leq 1$, define a pair of distributions of $d$-dimensional quantum states, $\rbra{\Pi_0, \Pi_1}$, as follows. 
    Let $d/6 \leq q \leq d/3$ be an integer and write $r = d - q$. 
    For each $b \in \cbra{0, 1}$, draw independent random variables $X_1, X_2, \dots, X_q \sim \nu_b$ and a Haar-random orthonormal basis $\ket{u_1}, \ket{u_2}, \dots, \ket{u_d} \in \mathbb{C}^d$.\footnote{Precisely, let $U$ be a Haar-random unitary operator on a $d$-dimensional Hilbert space. Then, setting $\ket{u_i} = U\ket{i}$ for $1 \leq i \leq d$ gives a Haar-random orthonormal basis.} 
    Define $\rho_b \coloneqq \rho\rbra{X_1, \dots, X_q, \ket{u_1}, \dots, \ket{u_q}}$ as follows. 
    \begin{enumerate}
        \item If $S \coloneqq\sum_{i=1}^q X_i \leq d/2$, set 
        \begin{equation} \label{eq:def-rho-b}
        \rho_b = \frac{1}{d} \sum_{i=1}^q X_i \ketbra{u_i}{u_i} + \frac{d-S}{rd} \rbra*{\mathbb{I} - \sum_{i=1}^q \ketbra{u_i}{u_i}}.
        \end{equation}
        \item Otherwise, set $\rho_b = \mathbb{I}/d$. 
    \end{enumerate}
    We write $\Pi_b$ to denote the probability distribution of $\rho_b$ generated from the above process. 
\end{definition}

The following lemma shows the indistinguishability between multiple identical and independent copies of a quantum state from the two probability distributions $\Pi_0$ and $\Pi_1$. 

\begin{lemma}[Haar-random multi-direction moment encoding] \label{lemma:q-moment-matching-multi-direction}
    Let $d \geq 2$, $d/6 \leq q \leq d/3$, and $r = d-q$. 
    For an $L$-moment matching $\rbra{\nu_0, \nu_1}$ on $I = \sbra{a_0, b_0} \subseteq \sbra{0, d/12}$ with mean $\nu \leq 1$, let $\rbra{\Pi_0, \Pi_1}$ be the pair of probability distributions of $d$-dimensional quantum states defined in \cref{def:pi-01}.
    Let $n \geq 1$ be an integer. 
    Denote
    \[
    \Omega_b^{\rbra{n}} = \E_{\rho \sim \Pi_b} \sbra*{\rho^{\otimes n}}, \qquad b \in \cbra{0, 1}.
    \]
    Set
    \[
    R = \frac{128nM^2}{25r^2}, \qquad M = \sup_{x \in I} \abs*{x - 1}. 
    \]
    If $R < 1$, then 
    \[
    \mathrm{T}\rbra*{\Omega_0^{\rbra{n}}, \Omega_1^{\rbra{n}}} \leq q \rbra*{\sqrt{\frac{R^{L+1}}{1-R}} + \exp\rbra*{-\frac{d}{96\rbra{b_0-a_0}^2}}}.
    \]
\end{lemma}

\begin{proof}
    Let $\Gamma_i$ be the probability distribution of $\rho\rbra{X_1, \dots, X_q, \ket{u_1}, \dots, \ket{u_q}}$ with $X_1, \dots, X_i \sim \nu_1$ and $X_{i+1}, \dots, X_q \sim \nu_0$. 
    Define
    \[
    \Psi_{i}^{\rbra{n}} = \E_{\rho \sim \Gamma_i} \sbra*{\rho^{\otimes n}}, \qquad 0 \leq i \leq q. 
    \]
    Note that $\Psi_0^{\rbra{n}} = \Omega_0^{\rbra{n}}$ and $\Psi_q^{\rbra{n}} = \Omega_1^{\rbra{n}}$. 
    By the triangle inequality, we have
    \begin{equation} \label{eq:td-sum-td}
        \mathrm{T}\rbra*{\Omega_0^{\rbra{n}}, \Omega_1^{\rbra{n}}} \leq \sum_{i=1}^q \mathrm{T}\rbra*{\Psi_{i-1}^{\rbra{n}}, \Psi_i^{\rbra{n}}}.
    \end{equation}

    Fix $1 \leq i \leq q$. 
    Consider the event $E_i = \cbra{S_i \leq 3d/8}$, where $S_i \coloneqq \sum_{j \neq i} X_j$ with $\E\sbra{S_i} = \rbra{q-1}\nu < d/3$. 
    Then, by Hoeffding's inequality \cite[Theorem 2]{Hoe63}, 
    \begin{equation} \label{eq:hoeffding}
    \Pr\sbra*{E_i^{\mathsf{c}}} \leq \exp\rbra*{-\frac{2\rbra{3d/8 - d/3}^2}{\rbra{q-1}\rbra{b_0-a_0}^2}} \leq \exp\rbra*{-\frac{d}{96\rbra{b_0-a_0}^2}}.
    \end{equation}
    
    Now consider the case when the event $E_i$ happens. 
    Conditioned on $Y = \cbra{\rbra{X_j, \ket{u_j}}: 1 \leq j \leq q, ~j \neq i}$, let $\mathcal{V}_i \coloneqq \rbra{\spanspace\cbra{\ket{u_j}: 1 \leq j \leq q, ~j \neq i}}^{\perp}$, then $\ket{u_i}$ is Haar-random in the $\rbra{r+1}$-dimensional space $\mathcal{V}_i$. 
    Define the random variable
    \[
    Z = \frac{\rbra{r+1} X_i}{d-S_i} - 1.
    \]
    It can be verified that (i) $-1 \leq Z \leq r$, (ii) $\abs{Z} \leq 8M/5$, and (iii) $\rbra{\zeta_0, \zeta_1}$ is an $L$-moment matching, where $\zeta_b$ is the probability distribution of $Z$ when $X_i \sim \nu_b$ for $b \in \cbra{0, 1}$. 
    Consider the quantum state $\varsigma_{Z, u_i}$ (defined in \cref{def:varsigma}) on the space $\mathcal{V}_i$ and the quantum channel
    \[
    \Phi_{i}\rbra{A} = \rbra*{1-\frac{S_i}{d}} J_i A J_i^\dag + \frac{\tr\rbra{A}}{d} \sum_{\substack{1 \leq j \leq q \\ j \neq i}} X_j \ketbra{u_j}{u_j},
    \]
    where $J_i \colon \mathcal{V}_i \hookrightarrow \mathbb{C}^d$ is an isometry defined by $J_i \ket{v} = \ket{v}$ for all $\ket{v} \in \mathcal{V}_i$.
    Direct calculation shows that on event $E_i$ (note that $S \leq d/2$ in this case), 
    \[
    \Phi_i\rbra*{\varsigma_{Z, u_i}} = \rho_b \coloneqq \rho\rbra{X_1, \dots, X_q, \ket{u_1}, \dots, \ket{u_q}}.
    \]

    Define 
    \[
    \tau_{b, Y}^{\rbra{n}} \coloneqq \E_{\substack{X_i \sim \nu_b \\ \ket{u_i} \sim \operatorname{Haar}\rbra{\mathcal{V}_i}}} \sbra*{ \rho_b^{\otimes n} \middle| Y }, \qquad b \in \cbra{0, 1}.
    \]
    Note that $\Psi_{i-1}^{\rbra{n}} = \E_{Y} \sbra{\tau_{0, Y}^{\rbra{n}}}$ and $\Psi_i^{\rbra{n}} = \E_{Y} \sbra{\tau_{1, Y}^{\rbra{n}}}$.
    Joint convexity of trace distance gives
    \begin{equation} \label{eq:td-EY}
    \mathrm{T}\rbra*{\Psi_{i-1}^{\rbra{n}}, \Psi_{i}^{\rbra{n}}} \leq \E_{Y} \sbra*{\mathrm{T}\rbra*{\tau_{0, Y}^{\rbra{n}}, \tau_{1, Y}^{\rbra{n}}}}.
    \end{equation}
    The trace distance in the expectation can be split into two cases: on $E_i$ and on $E_i^{\mathsf{c}}$, which gives
    \begin{equation} \label{eq:EY-td}
    \E_Y\sbra*{\mathrm{T}\rbra*{\tau_{0, Y}^{\rbra{n}}, \tau_{1, Y}^{\rbra{n}}}}
    \leq \E_Y \sbra*{\mathrm{T}\rbra*{\tau_{0, Y}^{\rbra{n}}, \tau_{1, Y}^{\rbra{n}}} \middle| E_i} + \Pr_Y \sbra*{E_i^{\mathsf{c}}}.
    \end{equation}
    Note that conditioned on $E_i$, 
    \[
    \tau_{b, Y}^{\rbra{n}} = \E_{\substack{X_i \sim \nu_b \\ \ket{u_i} \sim \operatorname{Haar}\rbra{\mathcal{V}_i}}} \sbra*{\Phi_i^{\otimes n}\rbra*{\varsigma_{Z, u_i}^{\otimes n}}},
    \]
    which, by contractivity of trace distance, gives
    \begin{align*}
    \mathrm{T}\rbra*{\tau_{0, Y}^{\rbra{n}}, \tau_{1, Y}^{\rbra{n}}}
    \leq \mathrm{T}\rbra*{\E_{\substack{X_i \sim \nu_0 \\ \ket{u_i} \sim \operatorname{Haar}\rbra{\mathcal{V}_i}}}\sbra*{\varsigma_{Z, u_i}^{\otimes n}}, \E_{\substack{X_i \sim \nu_1 \\ \ket{u_i} \sim \operatorname{Haar}\rbra{\mathcal{V}_i}}} \sbra*{\varsigma_{Z, u_i}^{\otimes n}}} 
    \leq \sqrt{\frac{R^{L+1}}{1-R}},
    \end{align*}
    where the last inequality is by \cref{lemma:q-moment-matching} with $B \coloneqq 8M/5$.
    Therefore, \cref{eq:td-EY}, together with \cref{eq:EY-td,eq:hoeffding}, gives
    \[
    \mathrm{T}\rbra*{\Psi_{i-1}^{\rbra{n}}, \Psi_{i}^{\rbra{n}}} \leq \sqrt{\frac{R^{L+1}}{1-R}} + \exp\rbra*{-\frac{d}{96\rbra{b_0-a_0}^2}},
    \]
    which completes the proof by \cref{eq:td-sum-td}.
\end{proof}

\section{Master Theorem for Quantum Functional Estimation}

To describe our master theorem, we need to focus on the following quantities. 

\begin{definition}[Important quantities]\label{def:master-functional-parameters}
Let $\nu_0$ and $\nu_1$ be probability distributions on an interval $I$. 
For a function $\phi \colon [0, +\infty) \to \mathbb{R}$, define the following parameters:
\begin{align*}
    M & = \sup_{x \in I} \abs*{x-1}, \\
    A_\phi & = \sup_{x, y \in I} \abs*{\phi\rbra{x} - \phi\rbra{y}}, \\
    L_\phi & = \sup_{\substack{x, y \in \sbra{1/2, 4/3} \\ x \neq y}} \abs*{\frac{\phi\rbra{x} - \phi\rbra{y}}{x - y}}, \\
    g_\phi & = \abs*{\E_{X \sim \nu_0} \sbra*{\phi\rbra{X}} - \E_{X \sim \nu_1} \sbra*{\phi\rbra{X}}}.
\end{align*}
\end{definition}

Our master theorem is provided in the following. 

\begin{theorem}[Master theorem]
\label{thm:master-functional}
    There exist universal constants $C, c > 0$ such that the following holds. 
    Suppose $\rbra{\nu_0, \nu_1}$ is a $K$-moment matching on an interval $I = \sbra{a_0, b_0} \subseteq [0, +\infty)$ with mean $\nu \leq 1$, and $\phi \colon [0, +\infty) \to \mathbb{R}$ is a function with $A_\phi, L_\phi, g_\phi \in \rbra{0, \infty}$.
    For sufficiently large $d \geq 1$, if 
    \begin{align}
        K & \geq C \log\rbra{d}, \label{cond1} \\
        \rbra{1+M}^2 & \leq cd/\log\rbra{d}, \label{cond2} \\
        \rbra{A_\phi + \rbra{b_0 - a_0} L_\phi}/g_\phi & \leq c\sqrt{d}, \label{cond3}
    \end{align}
    then there are two probability distributions $\Pi_0, \Pi_1$ of $d$-dimensional quantum states and intervals $I_0, I_1 \subseteq \mathbb{R}$ with $\operatorname{dist}\rbra{I_0, I_1} \geq g_\phi/8$ such that
    \begin{align}
        \mathrm{T}\rbra*{\E_{\rho \sim \Pi_0}\sbra*{\rho^{\otimes n}}, \E_{\rho \sim \Pi_1}\sbra*{\rho^{\otimes n}}} \leq \frac{1}{\poly\rbra{d}}, \qquad & 1 \leq n \leq n_0 \coloneqq \floor*{\frac{cd^2}{M^2}}, \label{outcond1} \\
        \Pr_{\rho \sim \Pi_b}\sbra*{\mathcal{L}_\phi\rbra{\rho} \in I_b} \geq 0.99, \qquad & b \in \cbra{0, 1}. \label{outcond2}
    \end{align}
    Furthermore, we have the following lower bounds:
    \begin{enumerate}[(1)]
    \item Estimating $\mathcal{L}_{\phi}\rbra{\rho}$ to within additive error $g_\phi/24$ with success probability at least $2/3$ for an unknown $d$-dimensional quantum state ${\rho}$ requires $\Omega\rbra{d^2/M^2}$ samples of $\rho$. \label{main-item1}
    \item Let $0 < p \leq 1$ and $\sigma$ be a fixed $d'$-dimensional quantum state such that
    \begin{equation} \label{outcond3}
     \Pr_{\rho\sim\Pi_b}
     \sbra{\Phi\rbra{p\rho \oplus \rbra{1-p}\sigma}\in J_b}
     \geq0.99,
     \qquad b\in\cbra{0,1},
    \end{equation}
    where $\Phi$ is a functional, and $J_0, J_1$ are two intervals with $\operatorname{dist}\rbra{J_0, J_1} > 2\varepsilon$.
    Then, estimating $\Phi\rbra{\widetilde{\rho}}$ to within additive error $\varepsilon$ with success probability at least $2/3$ for an unknown $\rbra{d+d'}$-dimensional quantum state $\widetilde{\rho}$ requires $\Omega\rbra{d^2/\rbra{pM^2}}$ samples of $\widetilde{\rho}$. \label{main-item2}
    \end{enumerate}
\end{theorem}

\begin{proof}
    Let $q = \floor{d/4}$ and $r = d - q \geq 3d/4$. 
    Note that $\nu \leq 1$, $a_0 \leq 1$, $b_0 \leq 1 + M$, and $b_0 - a_0 \leq 2M$. 
    For sufficiently small $c$ and sufficiently large $d$, \cref{cond2} implies that $b_0 \leq \sqrt{cd/\log\rbra{d}} \leq d/12$. 
    Therefore, the process in \cref{def:pi-01} applies here and we take the probability distributions $\Pi_0$ and $\Pi_1$ obtained there. 

    For $n \leq n_0 \leq cd^2/M^2$, applying \cref{lemma:q-moment-matching-multi-direction} with $L \coloneqq K$ gives
    \begin{align}
    \mathrm{T}\rbra*{\E_{\rho \sim \Pi_0}\sbra*{\rho^{\otimes n}}, \E_{\rho \sim \Pi_1}\sbra*{\rho^{\otimes n}}} 
    & \leq q \rbra*{\sqrt{\frac{R^{K+1}}{1-R}} + \exp\rbra*{-\frac{d}{96\rbra{b_0-a_0}^2}}} \nonumber \\
    & \leq q \rbra*{ 4^{-K} + \exp\rbra*{-\frac{d}{384M^2}} } \label{eq:second-neq} \\
    & \leq d^{1-C} + d^{1-1/384c}. \nonumber
    \end{align}
    where in \cref{eq:second-neq}, the first term uses \cref{cond1} and the fact that $R \leq 2048c/225 \leq 1/16$ for sufficiently small $c$, and the second term uses \cref{cond2}.
    Therefore, \cref{outcond1} holds for sufficiently large $C$, sufficiently small $c$, and sufficiently large $d$.

    For $b \in \cbra{0, 1}$, let $\rho_b \sim \Pi_b$. 
    Consider the event $E = \cbra{S \leq d/2}$, and using \cref{eq:hoeffding} gives
    \begin{equation} \label{eq:hoeffding1}
    \Pr\sbra*{E^{\mathsf{c}}} \leq \Pr\sbra*{E_i^{\mathsf{c}}} \leq \exp\rbra*{-\frac{d}{96\rbra{b_0-a_0}^2}}.
    \end{equation}
    On event $E$, the eigenvalues of $\rho_b$ are 
    \[
    \frac{X_1}{d}, \dots, \frac{X_q}{d}, \underbrace{\frac{d-S}{rd}, \dots, \frac{d-S}{rd}}_{r\textup{ times}},
    \]
    which gives
    \[
    \mathcal{L}_\phi\rbra{\rho_b} = \frac{1}{d} \sum_{i=1}^q \phi\rbra{X_i} + \frac{r}{d} \phi\rbra*{\frac{d-S}{r}}.
    \]
    Let
    \[
    \theta_b \coloneqq \frac{q}{d} \E_{X \sim \nu_b} \sbra*{\phi\rbra{X}} + \frac{r}{d} \phi\rbra*{\frac{d-q\nu}{r}}, \qquad I_b = \sbra*{\theta_b - \frac{qg_\phi}{8d}, \theta_b + \frac{qg_\phi}{8d}}.
    \]
    Direct calculation gives
    \begin{align*}
        \abs*{ \mathcal{L}_\phi\rbra{\rho_b} - \theta_b }
        & \leq \frac{1}{d} \abs*{ \sum_{i=1}^q \rbra*{ \phi\rbra{X_i} - \E_{X \sim \nu_b} \sbra*{\phi\rbra{X}}}} + \frac{r}{d} \abs*{ \phi\rbra*{\frac{d-S}{r}} - \phi\rbra*{\frac{d-q\nu}{r}} }.
    \end{align*}
    For the second term, using the Lipschitz constant $L_\phi$, we have 
    \begin{align*}
        \abs*{ \phi\rbra*{\frac{d-S}{r}} - \phi\rbra*{\frac{d-q\nu}{r}} }
        & \leq L_\phi \abs*{\frac{S-q\nu}{r}} = \frac{L_\phi}{r} \abs*{ \sum_{i=1}^q \rbra*{ X_i - \nu } }.
    \end{align*}
    Therefore, 
    \[
    \abs*{ \mathcal{L}_\phi\rbra{\rho_b} - \theta_b } \leq \frac{1}{d} \abs*{ \sum_{i=1}^q \rbra*{ \phi\rbra{X_i} - \E_{X \sim \nu_b} \sbra*{\phi\rbra{X}}}} + \frac{L_\phi}{d} \abs*{ \sum_{i=1}^q \rbra*{ X_i - \nu } }.
    \]
    By Hoeffding's inequality \cite[Theorem 2]{Hoe63}, we have
    \begin{equation} \label{eq:hoeffding2}
    \Pr\sbra*{\abs*{ \sum_{i=1}^q \rbra*{ \phi\rbra{X_i} - \E_{X \sim \nu_b} \sbra*{\phi\rbra{X}}} } > \frac{qg_\phi}{16}} \leq 2\exp\rbra*{-\frac{2\rbra{qg_\phi/16}^2}{qA_\phi^2}} = 2\exp\rbra*{-\frac{qg_\phi^2}{128A_\phi^2}},
    \end{equation}
    \begin{equation} \label{eq:hoeffding3}
    \Pr\sbra*{L_\phi \abs*{ \sum_{i=1}^q \rbra*{ X_i - \nu } } > \frac{qg_\phi}{16}} \leq 2\exp\rbra*{-\frac{2\rbra{qg_\phi/16L_\phi}^2}{q\rbra{b_0-a_0}^2}} = 2\exp\rbra*{-\frac{qg_\phi^2}{128\rbra{b_0-a_0}^2L_\phi^2}}.
    \end{equation}
    Combining \cref{eq:hoeffding1,eq:hoeffding2,eq:hoeffding3}, we have
    \begin{align*}
    \Pr\sbra*{\mathcal{L}_\phi\rbra{\rho_b} \in I_b} 
    & \geq 1 - \exp\rbra*{-\frac{d}{96\rbra{b_0-a_0}^2}} - 2\exp\rbra*{-\frac{qg_\phi^2}{128A_\phi^2}} - 2\exp\rbra*{-\frac{qg_\phi^2}{128\rbra{b_0-a_0}^2L_\phi^2}} \\
    & \geq 1 - d^{1-1/384c} - 4e^{-1/768c^2},
    \end{align*}
    where the second inequality is by \cref{cond3}. 
    This yields \cref{outcond2} for sufficiently small $c$ and sufficiently large $d$. 

    For \cref{main-item1}, consider the hypothesis testing $b \in \cbra{0, 1}$ that $b = 0$ and $b = 1$ with equal probability. 
    Suppose there is an estimator that estimates $\mathcal{L}_{\phi}\rbra{\rho}$ to within additive error $\varepsilon \coloneqq g_\phi/24$ with success probability at least $2/3$ using $n \leq n_0$ samples of $\rho$. 
    Then, \cref{outcond2} implies that this estimator can be used to distinguish $\rho \sim \Pi_0$ from $\rho \sim \Pi_1$ with success probability $p_{\textup{succ}} \geq 0.99 \cdot \frac{2}{3} = 0.66$, since $\operatorname{dist}\rbra{I_0, I_1} \geq g_\phi/8 > 2 \varepsilon$.
    On the other hand, the Helstrom--Holevo bound \cite{Hel67,Hol73} gives 
    \[
    p_{\textup{succ}} \leq \frac{1}{2} + \frac{1}{2} \mathrm{T}\rbra*{\E_{\rho \sim \Pi_0}\sbra*{\rho^{\otimes n}}, \E_{\rho \sim \Pi_1}\sbra*{\rho^{\otimes n}}} \leq 0.51,
    \]
    which gives a contradiction. 
    Therefore, any estimator that estimates $\mathcal{L}_{\phi}\rbra{\rho}$ to within additive error $\varepsilon$ requires $n_0 = \Omega\rbra{d^2/M^2}$ samples of $\rho$. 

    For \cref{main-item2}, we also consider the hypothesis testing $b \in \cbra{0, 1}$ that $b = 0$ and $b = 1$ with equal probability. 
    Suppose there is an estimator that estimates $\Phi\rbra{\widetilde\rho}$ to within additive error $\varepsilon$ using $n \leq cn_0/p$ samples of $\widetilde\rho$. 
    Then, \cref{outcond3} implies that this estimator can be used to distinguish $\rho \sim \Pi_0$ from $\rho \sim \Pi_1$ through the state $p\rho \oplus \rbra{1-p}\sigma$ with success probability $p_{\textup{succ}} \geq 0.99 \cdot \frac{2}{3} = 0.66$, since $\operatorname{dist}\rbra{J_0, J_1} \geq g_\phi/8 > 2 \varepsilon$.
    On the other hand, the Helstrom--Holevo bound \cite{Hel67,Hol73} gives 
    \begin{align}
    p_{\textup{succ}} 
    & \leq \frac{1}{2} + \frac{1}{2} \mathrm{T}\rbra*{\E_{\rho \sim \Pi_0}\sbra*{\rbra{p\rho \oplus \rbra{1-p}\sigma}^{\otimes n}}, \E_{\rho \sim \Pi_1}\sbra*{\rbra{p\rho \oplus \rbra{1-p}\sigma}^{\otimes n}}} \nonumber \\
    & = \frac{1}{2} + \frac{1}{2} \E_{N \sim \operatorname{Bin}\rbra{n, p}} \sbra*{ \mathrm{T}\rbra*{\E_{\rho \sim \Pi_0}\sbra*{\rho^{\otimes N}}, \E_{\rho \sim \Pi_1}\sbra*{\rho^{\otimes N}}} }, \label{eq:psucc}
    \end{align}
    where $\operatorname{Bin}\rbra{n, p}$ means the binomial distribution of $n$ items, each with probability $p$. 
    We consider the two cases: $N \leq n_0$ and $N > n_0$. 
    For $N \leq n_0$, we use \cref{outcond1}. 
    For $N > n_0$, the probability is at most $\Pr\sbra{N > n_0} \leq \E\sbra{N}/n_0 = np/n_0$. 
    Therefore, 
    \begin{align*}
        \eqref{eq:psucc} \leq \frac{1}{2} + \frac{1}{2} \rbra*{ 0.01 + \frac{np}{n_0} } \leq 0.51 + \frac{c}{2},
    \end{align*}
    which contradicts $p_{\textup{succ}} \geq 0.66$ for sufficiently small $c > 0$. 
    Therefore, any estimator that estimates $\Phi\rbra{\widetilde\rho}$ to within additive error $\varepsilon$ requires $cn_0/p = \Omega\rbra{d^2/\rbra{pM^2}}$ samples of $\widetilde\rho$. 
\end{proof}

\section{Applications}
\label{sec:applications}

\subsection{Von Neumann Entropy Estimation}

\begin{lemma}[{\cite[Lemma 5]{WY16}}]
\label{lemma:log-approximation}
There are universal constants \(c_{\log},c_0>0\) such that, for every integer \(K\geq2\) with $\eta={c_0}/{K^2}$, it holds that $E_K\rbra{\log\rbra{1/x},\sbra{\eta,1}} \geq c_{\log}$.
\end{lemma}

\begin{lemma}[Moment matching for von Neumann entropy]
\label{lemma:distri-entropy}
Let $c_0, c_{\log} > 0$ be the constants in \cref{lemma:log-approximation}.
Let \(K\geq2\), \(\eta=c_0/K^2\), $1\leq M\leq 4K^2$, and $\alpha=\eta M$.
Then, there is a $\rbra{K+1}$-moment matching $\rbra{\nu_0, \nu_1}$ on $\sbra{0, M}$ with mean $\alpha$ satisfying
\begin{align*}
    \E_{Y\sim\nu_1}\sbra{-Y\log \rbra{Y}}
    -\E_{Y\sim\nu_0}\sbra{-Y\log \rbra{Y}}
    \geq 2c_{\log}c_0\frac{M}{K^2}.
\end{align*}
\end{lemma}

\begin{proof}
By \cref{fact:mm,lemma:log-approximation}, there is a $K$-moment matching $\rbra{\mu_0, \mu_1}$ on \([\eta,1]\) such that
\begin{equation}
    \E_{X\sim\mu_1}\sbra{\log\rbra{1/X}}
    -\E_{X\sim\mu_0}\sbra{\log\rbra{1/X}}
    \geq 2c_{\log}.
    \label{eq:entropy-log-gap}
\end{equation}
Define $\nu_b=\mathsf T_{\eta,\alpha}\rbra{\mu_b}$ for $b\in\cbra{0,1}$.
Since \(\alpha/\eta=M\), $\nu_b$ is on \([0,M]\).  
By \cref{lemma:moment-lifting}, $\rbra{\nu_0, \nu_1}$ is a $\rbra{K+1}$-moment matching with mean $\alpha$.
Using the definition of $\mathsf T_{\eta,\alpha}$, 
\begin{align*}
    \E_{Y\sim\nu_b}\sbra{-Y\log \rbra{Y}}
    &=\E_{X\sim\mu_b}\sbra*{
      \frac{\eta}{X}
      \left[-\frac{\alpha X}{\eta}
      \log\rbra*{\frac{\alpha X}{\eta}}\right]}\\
    &=\alpha\log\rbra*{\frac{\eta}{\alpha}}
      +\alpha\E_{X\sim\mu_b}\sbra{\log\rbra{1/X}}.
\end{align*}
The first term is independent of \(b\).  Subtracting the two expectations and using \cref{eq:entropy-log-gap} gives
\[
    \E_{Y \sim \nu_1}\sbra{-Y\log \rbra{Y}}
    -\E_{Y \sim \nu_0}\sbra{-Y\log \rbra{Y}}
    \geq 2c_{\log}\alpha
    =2c_{\log}c_0\frac{M}{K^2}.
\]
\end{proof}

\begin{theorem}[Sample lower bound for von Neumann entropy estimation]
\label{cor:entropy-lower-bound}
There is a universal constant \(\varepsilon_0>0\) such that, for all sufficiently large \(d\) and every \(0<\varepsilon\leq\varepsilon_0\), estimating \(\mathrm S\rbra{\rho} = -\tr\rbra{\rho \log\rbra{\rho}}\) to additive error \(\varepsilon\), with success probability at least \(2/3\), requires
\begin{equation*}
    \Omega\rbra*{
    \frac{d^2}{\varepsilon\log^2\rbra d\,
    \max\cbra{1,\varepsilon\log^2\rbra d}}
    +\frac{\log^2\rbra{d}}{\varepsilon^2}}
\end{equation*}
samples of $\rho$.
\end{theorem}

\begin{proof}
    Let $C,c>0$ be the universal constants in \cref{thm:master-functional}.  The second term follows from the classical Shannon-entropy lower bound in \cite{JVHW15,WY16}.  We prove the first term by applying \cref{thm:master-functional}\eqref{main-item2}.
    Set
    \[
        d'=\floor*{\frac d2},
        \qquad
        K=\ceil*{C\log\rbra{d+2}},
        \qquad
        L=K+1,
        \qquad
        \eta=\frac{c_0}{K^2},
    \]
    \[
        M=\max\cbra*{1,\frac{32\varepsilon K^2}{c_{\log}c_0}},
        \qquad
        \alpha=\eta M,
        \qquad
        p=\frac{16\varepsilon K^2}{c_{\log}c_0M}, \qquad \varepsilon_0\leq
        \min\cbra*{\frac{c_{\log}c_0}{8},\frac{c_{\log}}{32}}.
    \]
    Then $1\leq M\leq4K^2$, $0<p\leq1/2$, and $\alpha\leq1$.
    Let $(\nu_0,\nu_1)$ be the $L$-moment matching from \cref{lemma:distri-entropy}, and apply \cref{thm:master-functional} in dimension $d'$ with $I=[0,M]$ and $\phi\rbra{x}=\phi_{\mathrm{S}}\rbra{x}=-x\log\rbra{x}$.
    It can be verified that the conditions in \cref{thm:master-functional}, \cref{cond1,cond2,cond3}, hold. 

    Let $\Pi_0, \Pi_1$ and $I_0, I_1$ be the probability distributions and intervals, respectively, specified in \cref{thm:master-functional}. 
    For $\rho\sim\Pi_b$, define the $d$-dimensional state
    \[
        \widetilde\rho
        =p\rho\oplus
        \rbra{1-p}\frac{\mathbb{I}_{d-d'}}{d-d'}.
    \]
    Since $\mathcal L_{\phi_{\mathrm{S}}}\rbra\rho
        =\mathrm S\rbra\rho-\log d'$,
    we have
    \[
        \mathrm S\rbra{\widetilde\rho}
        =p\rbra*{\mathcal L_{\phi_{\mathrm{S}}}\rbra\rho + \log \rbra{d'}} + h_2\rbra p
        +\rbra{1-p}\log\rbra{d-d'},
    \]
    where $h_2(p)=-(1-p)\log(1-p)-p\log p$.  Hence, with
    \[
        J_b=h_2\rbra p
        +p\rbra{\log d'+I_b}
        +\rbra{1-p}\log\rbra{d-d'},
    \]
    it can be seen that $\Pr_{\rho\sim\Pi_b}\sbra{\mathrm S\rbra{\widetilde\rho}\in J_b}
    \geq0.99$ and $\operatorname{dist}\rbra{J_0,J_1}=p\operatorname{dist}\rbra{I_0,I_1} \geq {pg_\phi}/{8}>2\varepsilon$.
    By \cref{thm:master-functional}\eqref{main-item2}, estimating $\mathrm{S}\rbra{\widetilde{\rho}}$ to within additive error $\varepsilon$ requires
    \begin{align*}
        \Omega\rbra*{\frac{d'^2}{pM^2}}
        =\Omega\rbra*{
        \frac{d^2}{\varepsilon\log^2d\,
        \max\cbra{1,\varepsilon\log^2d}}}
    \end{align*}
    samples of $\widetilde\rho$. 
\end{proof}

\subsection{R\'enyi and Tsallis Entropy Estimation}

\subsubsection{Moment matching for R\'enyi and Tsallis entropies}
\label{sec:power-certificates}

Let $\Delta_{h}\sbra{f}$ be the forward difference of a function $f$ with spacing $h$, defined by 
\[
\Delta_{h}\sbra{f}\rbra{x} = f\rbra{x+h} - f\rbra{x}.
\]
Denote $\Delta_h^{n}\sbra{f}\rbra{x} = \Delta_h \sbra{\Delta_h^{n-1}\sbra{f}}\rbra{x}$.
Note that
\begin{align*}
    \Delta_h^{n}\sbra{f}\rbra{x} = \sum_{k=0}^n \rbra{-1}^{n-k} \binom{n}{k} f\rbra{x + kh}.
\end{align*}

\begin{lemma}
\label{lemma:power-approximation-input}
For every non-integer $\alpha > 0$, there exist constants \(c_\alpha>0\) and
\(\gamma_\alpha\in(0,1/8]\) such that, for every integer \(K\ge2\) with $\eta={\gamma_\alpha}/{K^2}$, $E_K\rbra{x^{\alpha-1},\sbra{\eta,1}}
    \ge c_\alpha\eta^{\alpha-1}$.
\end{lemma}

It is noted that the special case of \cref{lemma:power-approximation-input} where $1 < \alpha < 3/2$ was given in {\cite[Lemma 13]{JVHW15}}.
Here, we generalize the result in \cite{JVHW15} to all non-integers $\alpha > 0$.

\begin{proof}[Proof of \cref{lemma:power-approximation-input}]
Let $\beta = \alpha - 1$, $s = \floor{\alpha}$, and $f\rbra{x} = x^\beta$. 
Then, $-1 < \beta - s < 0$. 
Let $G\rbra{x} = \Delta_1^{s}\sbra{f}\rbra{x}$, which is positive and strictly decreasing on $\rbra{0, +\infty}$ and thus $G\rbra{1} > 0$ and $G\rbra{1} - G\rbra{2} > 0$. 
Take
\[
\gamma_\alpha = \min\cbra*{\frac{1}{8}, \frac{1}{2\rbra{s+2}}, \frac{G\rbra{1} - G\rbra{2}}{16G\rbra{1}}}, \qquad c_\alpha = \frac{G\rbra{1} - G\rbra{2}}{2^{s+2}}.
\]
For $K \geq 2$ and $\eta = \gamma_\alpha/K^2$, we have $E_K\rbra{x^{\alpha-1},\sbra{\eta,1}} = \eta^\beta E_K\rbra{f, \sbra{1, \eta^{-1}}}$.
Therefore, it suffices to show that $E_K\rbra{f, \sbra{1, \eta^{-1}}} \geq c_\alpha$. 
For any polynomial $P$ of degree at most $K$, let $e = \Abs{f - P}_{\infty, \sbra{1, \eta^{-1}}}$, $Q\rbra{x} = \Delta_1^{s}\sbra{P}\rbra{x}$, and $I = \sbra{1, \eta^{-1}-s}$. 
The choice of $\gamma_\alpha$ ensures that $\abs{I} \geq \eta^{-1}/2$ and $\sbra{1, 2} \subseteq I$. 
Since $\Delta_1^{s}$ gives binomial coefficients whose absolute values sum to $2^s$, $\Abs{G - Q}_{\infty, I} \leq 2^s e$, which gives $\Abs{Q}_{\infty, I} \leq 2^s e + \Abs{G}_{\infty, I} = 2^s e + G\rbra{1}$.
Note that $Q$ is of degree at most $K$.
Then, Markov brothers' inequality \cite{Mar90,Mar92} gives 
\[
\Abs{Q'}_{\infty, I} \leq \frac{2K^2}{\abs{I}} \Abs{Q}_{\infty, I} \leq 4\gamma_\alpha \rbra*{G\rbra{1} + 2^s e}.
\]
On the other hand,
\begin{align*}
    G\rbra{1} - G\rbra{2}
    & \leq 2^{s+1} e + \abs{Q\rbra{1} - Q\rbra{2}} \\
    & \leq 2^{s+1} e + \Abs{Q'}_{\infty, I} \\
    & \leq \rbra{2 + 4\gamma_\alpha} 2^s e + 4 \gamma_\alpha G\rbra{1} \\
    & \leq \frac{G\rbra{1} - G\rbra{2}}{4} + 3 \cdot 2^s e,
\end{align*}
which gives $e \geq \rbra{G\rbra{1} - G\rbra{2}}/2^{s+2} = c_\alpha$ and thus yields the proof. 
\end{proof}

\begin{lemma}[Moment matching for power traces]
\label{lemma:pow-tr-moment-matching}
Fix a non-integer $\alpha > 0$, and let
\(c_\alpha,\gamma_\alpha\) be the constants in
\cref{lemma:power-approximation-input}.  
For every integer \(K\ge2\) and $1 \leq M \leq K^2$,
there is a $\rbra{K+1}$-moment matching $\rbra{\nu_0, \nu_1}$ on \([0,M]\) with mean
\(\gamma_\alpha M/K^2\), satisfying
\begin{equation*}
    \E_{Y \sim \nu_1}\sbra*{Y^\alpha}-\E_{Y \sim \nu_0}\sbra*{Y^\alpha}
    \ge\Delta_\alpha \rbra*{\frac{M}{K^2}}^\alpha, \qquad \Delta_\alpha = 2c_\alpha\gamma_\alpha^\alpha>0.
\end{equation*}
\end{lemma}

\begin{proof}
    Set $\eta=\gamma_\alpha/K^2$ and $a = \eta M$.  
    By \cref{fact:mm,lemma:power-approximation-input}, there is a $K$-moment matching $(\mu_0,\mu_1)$ on $[\eta,1]$ such that
    \[
        \E_{X\sim\mu_1}\sbra{X^{\alpha-1}}
        -\E_{X\sim\mu_0}\sbra{X^{\alpha-1}}
        \geq2c_\alpha\eta^{\alpha-1}.
    \]
    For $b\in\{0,1\}$, define $\nu_b=\mathsf T_{\eta,a}\rbra{\mu_b}$ on $\sbra{0, a/\eta} = \sbra{0, M}$.
    By \cref{lemma:T-eta-alpha}, $(\nu_0,\nu_1)$ is a $(K+1)$-moment matching with mean $a = \gamma_\alpha M/K^2$.  
    Moreover,
    \[
        \E_{Y\sim\nu_b}\sbra{Y^\alpha}
        =a^\alpha\eta^{1-\alpha}
        \E_{X\sim\mu_b}\sbra{X^{\alpha-1}}.
    \]
    Taking the difference yields the proof.
\end{proof}

\begin{corollary}[Moment matching for Tsallis entropy]
\label{lemma:power-moment-matching}
Fix a non-integer $\alpha > 0$.  
For every integer \(K\ge2\),
there is a $\rbra{K+1}$-moment matching $\rbra{\nu_0, \nu_1}$ on \([0,K^2]\) with mean
\(\gamma_\alpha\), satisfying
\begin{equation*}
    \E_{Y \sim \nu_1}\sbra*{Y^\alpha}-\E_{Y \sim \nu_0}\sbra*{Y^\alpha}
    \ge\Delta_\alpha.
\end{equation*}
\end{corollary}

\begin{proof}
    This is the special case of \cref{lemma:pow-tr-moment-matching} with $M = K^2$.
\end{proof}

\begin{lemma}[Relative polynomial approximation for power functions] \label{lemma:rel-EK-power}
    Fix a non-integer $\alpha > 1$. 
    Let $s = \ceil{\alpha}$ and $\beta = s - \alpha$. 
    There are constants $c_{\alpha, 0}, \delta_{\alpha, 0} > 0$ such that for every integer $K \geq \alpha$, setting $\eta = c_{\alpha, 0} K^{-2s/\beta} \leq 1/\rbra{s+1}$ gives
    \[
    E_{K}^{\mathrm{rel}}\rbra{x^\alpha, \sbra{\eta, 1}} \coloneqq \inf_{\deg \rbra{P}\leq K}
    \sup_{x\in[\eta,1]}
    \frac{\abs{P\rbra{x}-x^\alpha}}{x^\alpha} \in \sbra{\delta_{\alpha, 0}, 1}.
    \]
\end{lemma}

\begin{proof}
    Taking $P\rbra{x} = 0$ gives $E_{K}^{\mathrm{rel}}\rbra{x^\alpha, \sbra{\eta, 1}} \leq 1$. 
    It remains to prove the lower bound.
    We take $\delta_{\alpha, 0} \coloneqq D_{\alpha, s} / \rbra{4 A_{\alpha, s}}$, where
    \begin{align*}
        D_{\alpha,s} & = \Delta_{1}^s\sbra{x^\alpha}\rbra{1} = \sum_{j=0}^s \rbra{-1}^{s-j} \binom{s}{j} \rbra{1+j}^{\alpha} > 0, \\
        A_{\alpha,s} & = \sum_{j=0}^s \binom{s}{j} \rbra{1+j}^{\alpha}.
    \end{align*}
Suppose, toward a contradiction, that a polynomial \(P\) of
degree at most \(K\) satisfies $\abs{P\rbra{x}-x^\alpha} < \delta_{\alpha,0}x^\alpha$ for every $\eta\leq x\leq1$.
Taking $x_j = \rbra{1+j}\eta$, we have
\[
\begin{aligned}
    \abs*{\Delta_\eta^s\sbra{P}\rbra{\eta}}
    &\geq
    \abs*{\Delta_\eta^s\sbra*{x^\alpha}\rbra{\eta}} - \abs*{\sum_{j=0}^s \rbra{-1}^{s-j} \binom{s}{j} \rbra*{P\rbra{x_j} - x_j^\alpha}} \\
    & > \abs*{\Delta_\eta^s\sbra*{x^\alpha}\rbra{\eta}}
    -
    \delta_{\alpha,0}
    \sum_{j=0}^s
    \binom{s}{j}x_j^\alpha\\
    &=
    D_{\alpha,s} \eta^\alpha
    -
    \delta_{\alpha,0}A_{\alpha,s} 
    \eta^\alpha\\
    &\geq
    \frac{D_{\alpha,s}}{2}\eta^\alpha.
\end{aligned}
\]
By the Schwarz mean-value theorem for divided differences \cite{Sch82} (cf.\ \cite[Equation (44)]{dB05}),
there exists
\(\xi\in[\eta,(s+1)\eta]\) such that $\Delta_\eta^s\sbra{P}\rbra{\eta} = \eta^sP^{(s)}\rbra{\xi}$.
Consequently,
\begin{equation}
    \Abs*{P^{(s)}}_{\infty,[\eta,1]} \coloneqq \sup_{x \in \sbra{\eta, 1}} \abs*{ P^{\rbra{s}}\rbra{x} }
    \geq
    \frac{D_{\alpha,s}}{2}
    \eta^{\alpha-s}
    =
    \frac{D_{\alpha,s}}{2}
    \eta^{-\beta}.
    \label{eq:relative-derivative-lower}
\end{equation}
On the other hand, the $s$-th Markov brothers' inequality \cite{Mar90,Mar92} on
\([\eta,1]\) gives
\[
    \Abs*{P^{(s)}}_{\infty,[\eta,1]}
    \leq
    \frac{2^s}{\rbra{2s-1}!!}\cdot
    \frac{K^{2s}}{\rbra{1-\eta}^s}
    \Abs P_{\infty,[\eta,1]}.
\]
Note that $\Abs P_{\infty,[\eta,1]} \leq 1 + \delta_{\alpha, 0}$. 
Hence, for
\(\eta\leq1/\rbra{s+1}\), we have $\Abs{P^{(s)}}_{\infty,[\eta,1]} \leq C_\alpha K^{2s}$ for some constant $C_\alpha > 0$, which contradicts \cref{eq:relative-derivative-lower} for sufficiently small $c_{\alpha, 0} > 0$.
\end{proof}

\begin{lemma}[Moment matching for R\'enyi entropy of order $\alpha > 1$] \label{lemma:relative-power-matching}
    Fix a non-integer $\alpha > 1$. 
    Let $\beta = \ceil{\alpha} - \alpha \in \rbra{0, 1}$. 
    There are constants $C_\alpha, c_\alpha, \delta_\alpha > 0$ such that for every integer $K \geq \alpha$, there is a $K$-moment matching $\rbra{\nu_0, \nu_1}$ on $\sbra{c_\alpha K^{-2\ceil{\alpha}/\beta}, C_\alpha K^{2\ceil{\alpha}/\beta}}$ with mean $1/2$, satisfying
    \[
    \delta_\alpha \rbra*{ \E_{Y \sim \nu_0}\sbra*{Y^\alpha} + \E_{Y \sim \nu_1}\sbra*{Y^\alpha} } \leq \abs*{ \E_{Y \sim \nu_1}\sbra*{Y^\alpha} - \E_{Y \sim \nu_0}\sbra*{Y^\alpha} } \leq \frac{1}{2} \rbra*{ \E_{Y \sim \nu_0}\sbra*{Y^\alpha} + \E_{Y \sim \nu_1}\sbra*{Y^\alpha} }.
    \]
\end{lemma}

\begin{proof}
    For convenience, let $s=\ceil{\alpha} \leq K$, $\eta=c_{\alpha,0}K^{-2s/\beta} < 1$, and $I=\sbra{\eta,1}$, such that $E_K^{\mathrm{rel}}\rbra{x^\alpha,I} \in \sbra{\delta_{\alpha,0}, 1}$,
    where \(c_{\alpha,0}, \delta_{\alpha,0}>0\) are the constants in
    \cref{lemma:rel-EK-power}. 

    Consider \(C\rbra{I}\) equipped with the weighted norm
    \[
        \Abs{f}_{x^\alpha,I}
        \coloneqq
        \sup_{x\in I}
        \frac{\abs{f\rbra{x}}}{x^\alpha},
    \]
    which is a Banach space. 
    Let \(\mathcal P_K\) denote the set of all real polynomials of
    degree at most \(K\), and consider the quotient Banach space $\mathcal X = C\rbra{I}/\mathcal P_K$.
    For each coset $\sbra{f} = f+\mathcal P_K = \set{f+P}{P\in\mathcal P_K}$,
    define the norm
    \[
        \Abs{\sbra{f}}_{\mathcal X}
        \coloneqq
        \inf_{P\in\mathcal P_K}
        \Abs{f-P}_{x^\alpha,I}.
    \]
    In particular, $\Abs{\sbra{x^\alpha}}_{\mathcal X} = E_K^{\mathrm{rel}}\rbra{x^\alpha,I}$.

    Define a linear functional $\ell \colon \spanspace\cbra{\sbra{x^\alpha}} \to \mathbb R$
    by $\ell\rbra*{a\sbra{x^\alpha}} = aE_K^{\mathrm{rel}}\rbra{x^\alpha,I}$ for $a\in\mathbb R$.
    Its operator norm is
    \[
    \begin{aligned}
        \Abs{\ell}
        =\sup_{a\neq0}\frac{\abs*{\ell\rbra{a\sbra{x^\alpha}}}}{\Abs{a\sbra{x^\alpha}}_{\mathcal X}}
        =\sup_{a\neq0}\frac{\abs a E_K^{\mathrm{rel}}\rbra{x^\alpha,I}} { \abs a E_K^{\mathrm{rel}}\rbra{x^\alpha,I}}
        = 1.
    \end{aligned}
    \]
    By the Hahn--Banach theorem, \(\ell\) can be extended to a
    linear functional \(L\colon\mathcal X\to\mathbb R\) with
    \(\Abs L=1\).  
    Define  $\widetilde L\rbra f = L\rbra{\sbra f}$.
    Then, $\widetilde L\rbra P=0$ for every $P\in\mathcal P_K$, and $\widetilde L\rbra{x^\alpha} = E_K^{\mathrm{rel}}\rbra{x^\alpha,I}$.
    Let $\Gamma\rbra g = \widetilde L\rbra{x^\alpha g}$ for $g \in C\rbra{I}$ with operator norm $\Abs{\Gamma} = 1$, since $\Abs{x^\alpha g}_{x^\alpha,I} = \Abs g_{\infty,I}$. 
    By the Riesz--Markov--Kakutani representation theorem, there is a finite signed
    Borel measure \(\Lambda\) on \(I\) such that
    \[
        \Gamma\rbra g
        =
        \int_I g\rbra{x}\,
        \mathrm d\Lambda\rbra{x},
        \qquad
        \abs{\Lambda}\rbra I \coloneqq \int_{I} \mathrm{d} \abs{\Lambda} \rbra{x} =1.
    \]
    Define another finite signed measure \(\omega\) by $\mathrm d\omega\rbra{x} = x^{-\alpha}\, \mathrm d\Lambda\rbra{x}$.
    Then, for every \(f\in C\rbra I\),
    \[
    \begin{aligned}
        \widetilde L\rbra f
        =\Gamma\rbra*{\frac{f}{x^\alpha}}
        =\int_I\frac{f\rbra{x}}{x^\alpha}\,\mathrm d\Lambda\rbra{x}
        =\int_I f\rbra{x}\,\mathrm d\omega\rbra{x}.
    \end{aligned}
    \]
    Therefore, since \(x^j\in\mathcal P_K\) for \(0\leq j\leq K\),
    \[
        \int_Ix^j\,\mathrm d\omega\rbra{x}
        =
        \widetilde L\rbra{x^j}
        =
        0,
        \qquad
        0\leq j\leq K,
    \]
    \[
        \int_Ix^\alpha\,\mathrm d\omega\rbra{x}
        =
        \widetilde L\rbra{x^\alpha}
        =
        E_K^{\mathrm{rel}}\rbra{x^\alpha,I}.
    \]

    Write the Jordan decomposition as $\omega=\omega_+-\omega_-$.
    Since \(1\in\mathcal P_K\), $\omega\rbra{I} = \widetilde{L}\rbra{1} = 0$.
    Hence, $\omega_+\rbra I = \omega_-\rbra I = u$ for some \(u>0\).  
    Define $\mu_1={\omega_+}/{u}$ and $\mu_0={\omega_-}/{u}$.
    Then \(\mu_0\) and \(\mu_1\) are probability distributions on
    \(I\).  Moreover, for every \(0\leq j\leq K\),
    \[
    \begin{aligned}
        \E_{X\sim\mu_1}\sbra{X^j}
        -
        \E_{X\sim\mu_0}\sbra{X^j}
        =
        \frac1u
        \int_Ix^j\,\mathrm d\omega\rbra{x}=
        0.
    \end{aligned}
    \]
    Thus \(\rbra{\mu_0,\mu_1}\) is a \(K\)-moment matching.
    For \(b\in\cbra{0,1}\), let $B_b = \E_{X\sim\mu_b}\sbra{X^\alpha}$.
    Since
    \(\abs{\omega}=\omega_++\omega_-\), we have
    \[
    \begin{aligned}
        u\rbra{B_1+B_0}
        &=
        \int_Ix^\alpha\,
        \mathrm d\abs{\omega}\rbra{x}
        =
        1,\\
        u\rbra{B_1-B_0}
        &=
        \int_Ix^\alpha\,
        \mathrm d\omega\rbra{x}
        =
        E_K^{\mathrm{rel}}\rbra{x^\alpha,I}.
    \end{aligned}
    \]
    Consequently,
    \begin{equation} \label{eq:diff-over-sum}
        \frac{\abs{B_1-B_0}}{B_1+B_0}
        =
        E_K^{\mathrm{rel}}\rbra{x^\alpha,I}
        \in \sbra{\delta_{\alpha,0}, 1}.
    \end{equation}
    Next, define
    \[
        \overline\mu
        =
        \frac{\mu_0+\mu_1}{2},
        \qquad
        \widetilde\mu_b
        =
        \frac{\mu_b+\overline\mu}{2},
        \qquad
        b\in\cbra{0,1}.
    \]
    The pair
    \(\rbra{\widetilde\mu_0,\widetilde\mu_1}\) is still a
    \(K\)-moment matching on \(I\) with mean $m \in \sbra{\eta, 1}$. 
    For $b \in \cbra{0, 1}$, let $\nu_b$ be the probability distribution of 
    \[
        Y=\frac{X}{2m},
        \qquad
        X\sim\widetilde\mu_b.
    \]
    It can be verified that $\rbra{\nu_0, \nu_1}$ is a $K$-moment matching on $\sbra{\frac{\eta}{2m}, \frac{1}{2m}} \subseteq \sbra{\frac{\eta}{2}, \frac{1}{2\eta}} = \sbra{ c_\alpha K^{-2s/\beta}, C_\alpha K^{2s/\beta} }$ with mean $\frac{1}{2}$, where $c_\alpha=\frac{c_{\alpha,0}}2$ and $C_\alpha=\frac1{2c_{\alpha,0}}$.
    Finally, \cref{eq:diff-over-sum} gives that 
    \[
    \begin{aligned}
        \frac{\delta_{\alpha,0}}2
        \leq
        \frac{
        \abs*{
        \E_{Y\sim\nu_1}\sbra{Y^\alpha}
        -
        \E_{Y\sim\nu_0}\sbra{Y^\alpha}}}
        {
        \E_{Y\sim\nu_1}\sbra{Y^\alpha}
        +
        \E_{Y\sim\nu_0}\sbra{Y^\alpha}}\leq
        \frac12,
    \end{aligned}
    \]
    where taking $\delta_\alpha \coloneqq \delta_{\alpha, 0}/2$ yields the proof. 
\end{proof}

\subsubsection{Quantum R\'enyi Entropy Estimation}
\label{sec:renyi-estimation}

\begin{theorem}[Sample lower bounds for R\'enyi entropy estimation]
\label{cor:renyi-high-accuracy}
Let $\rho$ be an unknown $d$-dimensional quantum state, where $d$
is sufficiently large.
\begin{itemize}
    \item For \(0<\alpha<1\) and any sufficiently small $\varepsilon > 0$, estimating \(\mathrm S^{\mathrm{R}}_\alpha(\rho)\) to within additive error $\varepsilon$ with success probability at least $2/3$ requires sample complexity
\begin{equation*}
    \Omega\rbra*{
    \frac{d^{1+1/\alpha}}
         {\varepsilon^{1/\alpha}\log^2\rbra{d}\max\cbra{1, \varepsilon^{1/\alpha}\log^2\rbra{d}}}+
    \frac{d^{1/\alpha-1}}
         {\varepsilon^2}}.
\end{equation*}
    \item For non-integer \(\alpha>1\) and any sufficiently small $\varepsilon > 0$, estimating
\(\mathrm S^{\mathrm{R}}_\alpha(\rho)\) to within additive error $\varepsilon$ with success probability at least $2/3$ requires sample complexity
\[
    \Omega\rbra*{
    \frac{d^2}{\varepsilon^{1/\alpha}\log^{\kappa_\alpha}\rbra{d}}+
    \frac{d^{1-1/\alpha}}{\varepsilon^2}}, \qquad \textup{where }
\kappa_\alpha = \frac{4\ceil{\alpha}}{\ceil{\alpha}-\alpha}.
\]
\end{itemize}
\end{theorem}

\begin{proof}
Let $C,c>0$ be the universal constants in
\cref{thm:master-functional}.  Fix a sufficiently large constant
$C_\alpha>0$ and a sufficiently small constant
$\varepsilon_\alpha>0$, depending only on $\alpha$.
Let $0 < \varepsilon \leq \varepsilon_\alpha$. 

For the case of $0<\alpha<1$, set
\[
    d'=d-1,
    \qquad
    K=\ceil*{C\log\rbra{d+2}},
    \qquad
    L=K+1,
\]
\[
    U
    =
    \max\cbra*{
    1,
    C_\alpha K^2\varepsilon^{1/\alpha}},
    \qquad
    p
    =
    \frac{
    C_\alpha K^2\varepsilon^{1/\alpha}}
    {\rbra{d'}^{\rbra{1-\alpha}/\alpha}U},
\]
and take $I = \sbra{a_0, b_0} = \sbra{0, U}$ and $M = \max\cbra{1, U - 1}$. 
For sufficiently small $\varepsilon_\alpha > 0$, one has $1\leq U\leq K^2$ and thus $M=\Theta\rbra U$. 
Moreover, $0<p\leq1/2$
for all sufficiently large $d$, and
\begin{equation}
    p^\alpha\rbra{d'}^{1-\alpha}
    =
    \frac{
    C_\alpha^\alpha K^{2\alpha}\varepsilon}
    {U^\alpha}
    \leq1.
    \label{eq:renyi-below-power-coefficient}
\end{equation}

Let $(\nu_0,\nu_1)$ be the $L$-moment matching on $\sbra{0, U}$ specified by
\cref{lemma:pow-tr-moment-matching}, with mean $\nu = \gamma_\alpha U/K^2 \leq 1$. 
Apply \cref{thm:master-functional} in dimension $d'$ with $K\coloneqq L$, $I=\sbra{a_0,b_0}=\sbra{0,U}$, and $\phi\rbra{x}=\phi_\alpha\rbra{x}=x^\alpha$.
The parameters in \cref{def:master-functional-parameters} satisfy
\[
    g_\phi
    \geq
    \Delta_\alpha
    \left(\frac{U}{K^2}\right)^\alpha,
    \qquad
    A_\phi\leq U^\alpha,
    \qquad
    L_\phi
    \leq
    \alpha
    \max\cbra*{
    2^{1-\alpha},
    \rbra{4/3}^{\alpha-1}}.
\]
The conditions of \cref{thm:master-functional} can be verified: \cref{cond1,cond2,cond3} hold for all sufficiently large $d$.  
First, $L\geq C\log \rbra{d'}$.
Second, $\rbra{1+M}^2\log \rbra{d'}
    =
    O\rbra{\log^5 \rbra{d}}
    = o\rbra{d'}$
for all sufficiently large $d$.  Finally,
\begin{align*}
    \frac{
    A_\phi+\rbra{b_0-a_0}L_\phi}
    {g_\phi}
    \leq
    O_\alpha\rbra*{
    K^{2\alpha}
    \rbra{1+U^{1-\alpha}}}
    \leq
    O_\alpha\rbra{K^2}
    =
    o_\alpha\rbra{\sqrt{d'}},
\end{align*}
where the second inequality uses $1\leq U\leq K^2$.
Thus \cref{cond1,cond2,cond3} hold.

Let $\Pi_0,\Pi_1$ and $I_0,I_1$ be the probability distributions
and intervals specified in \cref{thm:master-functional}.  
When $0<\alpha<1$, it holds that $\mathcal L_{x^\alpha}\rbra\rho \in \sbra{0, 1}$ and without loss of generality, we assume that $I_b \subseteq \sbra{0, 1}$. 
For $\rho\sim\Pi_b$, define the $d$-dimensional state
\[
    \widetilde\rho
    =
    p\rho
    \oplus
    \rbra{1-p}\ketbra00.
\]
Since $\tr\rbra{\rho^\alpha}
    =
    \rbra{d'}^{1-\alpha}
    \mathcal L_{x^\alpha}\rbra\rho$,
we have $\tr\rbra{\widetilde\rho^\alpha}
    =
    \rbra{1-p}^\alpha
    +
    p^\alpha\rbra{d'}^{1-\alpha}
    \mathcal L_{x^\alpha}\rbra\rho$.
For $b\in\cbra{0,1}$, define
\[
    J_b
    =
    \rbra{1-p}^\alpha
    +
    p^\alpha\rbra{d'}^{1-\alpha}I_b, \qquad
    \widetilde J_b
    =
    \cbra*{
    \frac1{1-\alpha}\log x
    :
    x\in J_b}.
\]
It can be seen that $\Pr_{\rho\sim\Pi_b}\sbra{\mathrm S_\alpha^{\mathrm R}\rbra{\widetilde\rho}\in\widetilde J_b}\geq0.99$, $J_0\cup J_1
\subseteq
\sbra{2^{-\alpha},2}$ (under the assumption that $I_0\cup I_1\subseteq[0,1]$), and $\operatorname{dist}\rbra{J_0,J_1} = p^\alpha\rbra{d'}^{1-\alpha}
\operatorname{dist}\rbra{I_0,I_1} \geq C_\alpha^\alpha \Delta_\alpha \varepsilon/8 \geq 8\rbra{1-\alpha}\varepsilon$ (where the last inequality holds by choosing $C_\alpha$ sufficiently large). 
Therefore, $\operatorname{dist}
\rbra{\widetilde J_0,\widetilde J_1} \geq 4\varepsilon$. 
By \cref{thm:master-functional}\eqref{main-item2}, estimating
$\mathrm S_\alpha^{\mathrm R}\rbra{\widetilde\rho}$ to within
additive error $\varepsilon$ requires
\begin{align*}
    \Omega\rbra*{
    \frac{\rbra{d'}^2}{pM^2}}
    =
    \Omega\rbra*{
    \frac{d^{1+1/\alpha}}
    {\varepsilon^{1/\alpha}\log^2\rbra d
    \max\cbra{
    1,\varepsilon^{1/\alpha}\log^2\rbra d}}}
\end{align*}
samples of $\widetilde\rho$. 

Now suppose that $\alpha>1$ is non-integer.  Set
\[
    s=\ceil*{\alpha},
    \qquad
    \beta=s-\alpha,
    \qquad
    d'=\floor*{\frac d2},
    \qquad
    K=\ceil*{C\log\rbra{d+2}}.
\]
Let $c_{\alpha,0},C_{\alpha,0},\delta_{\alpha,0}>0$ be the
constants in \cref{lemma:relative-power-matching}, and set
\[
    a_0=c_{\alpha,0}K^{-2s/\beta},
    \qquad
    b_0=C_{\alpha,0}K^{2s/\beta}, \qquad I=\sbra{a_0,b_0},
    \qquad
    M=\max\cbra{1-a_0,b_0-1}
    =
    \Theta_\alpha\rbra*{K^{2s/\beta}}.
\]
Let $(\nu_0,\nu_1)$ be the $K$-moment matching on $I$ with mean $1/2$ specified in
\cref{lemma:relative-power-matching}.  
For $b\in\cbra{0,1}$, set
$A_b=\E_{Y\sim\nu_b}\sbra{Y^\alpha} \geq 2^{-\alpha}$; without loss of generality, assume that $A_1>A_0$.  
Then,
\begin{equation}
    \delta_{\alpha,0}\rbra{A_0+A_1}
    \leq
    A_1-A_0
    \leq
    \frac12\rbra{A_0+A_1}.
    \label{eq:renyi-relative-power-gap}
\end{equation}
Apply \cref{thm:master-functional} in dimension $d'$ with
$\phi\rbra x=x^\alpha$.  The parameters in
\cref{def:master-functional-parameters} satisfy
\[
    g_\phi=A_1-A_0
    \geq
    \delta_{\alpha,0}2^{1-\alpha},
    \qquad
    A_\phi\leq b_0^\alpha, \qquad L_\phi
    \leq
    \alpha\rbra*{\frac43}^{\alpha-1}.
\]
Hence, $K\geq C\log\rbra{d'}$, $\rbra{1+M}^2\log\rbra{d'} = O_\alpha\rbra{K^{4s/\beta}\log d} = o\rbra{d'}$, and
\[
    \frac{
    A_\phi+\rbra{b_0-a_0}L_\phi}
    {g_\phi}
    =
    O_\alpha\rbra*{K^{2s\alpha/\beta}}
    =
    o\rbra{\sqrt{d'}}.
\]
Thus \cref{cond1,cond2,cond3} hold for all sufficiently large $d$.

Let $\Pi_0,\Pi_1$ and $I_0,I_1$ be the probability distributions
and intervals specified in \cref{thm:master-functional}, where
\[
I_b = \sbra*{\theta_b - \frac{q g_\phi}{8 d'}, \theta_b + \frac{q g_\phi}{8 d'}}, \qquad \theta_b = \frac{q}{d'} A_b + \frac{r}{d'}\rbra*{\frac{d'-q/2}{r}}^{\alpha}.
\]
Note that 
\[
\operatorname{dist}\rbra{I_0, I_1} \geq \frac{g_\phi}{8}, \qquad \sup I_1 \leq \frac{g_\phi}{4\delta_{\alpha, 0}} + \rbra*{\frac{4}{3}}^\alpha + \frac{g_\phi}{32}.
\]
Set
\begin{equation*}
    p = \rbra*{ 32\rbra{\alpha-1}H_\alpha \varepsilon }^{1/\alpha}, \qquad H_\alpha = \frac{1}{4\delta_{\alpha, 0}} + \frac{1+\rbra{4/3}^{\alpha}}{2^{1-\alpha}\delta_{\alpha, 0}} + \frac{1}{32}.
\end{equation*}
For sufficiently small $\varepsilon_\alpha>0$, one has
$0<p\leq1/2$ for every
$0<\varepsilon\leq\varepsilon_\alpha$.
For $\rho\sim\Pi_b$, define the $d$-dimensional state
\[
    \widetilde\rho
    =
    p\rho
    \oplus
    \rbra{1-p}
    \frac{\mathbb{I}_{d-d'}}{d-d'}.
\]
Since
$\tr\rbra{\rho^\alpha}
=(d')^{1-\alpha}\mathcal L_{x^\alpha}\rbra\rho$, one has
\[
    \mathrm S_\alpha^{\mathrm R}\rbra{\widetilde\rho} = \log\rbra{d'} + \frac1{1-\alpha} \log \rbra*{ p^\alpha\mathcal L_{x^\alpha}\rbra\rho + \rbra{1-p}^\alpha \rbra*{ \frac{d-d'}{d'} }^{1-\alpha} }.
\]
For $b\in\cbra{0,1}$, define
\[
    J_b
    =
    \rbra{1-p}^\alpha
    \left(
    \frac{d-d'}{d'}
    \right)^{1-\alpha}
    +
    p^\alpha I_b, \qquad
    \widetilde J_b
    =
    \cbra*{
    \log\rbra{d'}
    +
    \frac1{1-\alpha}\log z
    :
    z\in J_b}.
\]
Then, it can be verified that $\Pr_{\rho\sim\Pi_b}\sbra{\mathrm S_\alpha^{\mathrm R}\rbra{\widetilde\rho}\in\widetilde J_b}\geq0.99$ for $b\in\cbra{0,1}$ and $\operatorname{dist}
    \rbra{\widetilde J_0,\widetilde J_1}
    \geq4\varepsilon$.
To see this, note that $\operatorname{dist}\rbra{J_0, J_1} \geq p^\alpha g_\phi/8$ and $\sup J_1 \leq H_\alpha g_\phi$, which gives $\operatorname{dist}\rbra{\widetilde J_0, \widetilde J_1} \geq \frac{p^\alpha}{8H_\alpha\rbra{\alpha-1}} \geq 4\varepsilon$. 
By
\cref{thm:master-functional}\eqref{main-item2}, estimating $\mathrm S_\alpha^{\mathrm R}\rbra{\widetilde\rho}$ to within additive error $\varepsilon$ requires 
\begin{align*}
    \Omega\rbra*{
    \frac{\rbra{d'}^2}{pM^2}}
    =
    \Omega\rbra*{
    \frac{d^2}
    {\varepsilon^{1/\alpha}
    \log^{\kappa_\alpha}\rbra d}}
\end{align*}
samples of $\widetilde{\rho}$, where
\[
    \kappa_\alpha
    =
    \frac{4s}{\beta}
    =
    \frac{
    4\ceil*{\alpha}}
    {\ceil*{\alpha}-\alpha}.
\]
On the other hand, the sample lower bound for classical R\'enyi entropy estimation in \cite[Theorem 16]{AOST17} gives $\Omega\rbra{d^{1-1/\alpha}/\varepsilon^2}$.
Combining the two lower bounds yields the proof.
\end{proof}

\subsubsection{Quantum Tsallis Entropy Estimation}
\label{sec:tsallis-estimation}

\begin{lemma}
\label{cor:tsallis-unified}
Fix \(0<\alpha<2\) with \(\alpha\ne1\), and let \(d\) be sufficiently
large.
If 
\begin{equation}
    h=
    \left(
    \frac{32\abs{1-\alpha}}{\Delta_\alpha}
    \varepsilon d^{\alpha-1}
    \right)^{1/\alpha} \leq \frac14,
    \label{eq:tsallis-h-parameter}
\end{equation}
then estimating
\(\mathrm S^{\mathrm T}_\alpha(\rho)\) to within additive error
\(\varepsilon\) for a
\((d+1)\)-dimensional quantum state \(\rho\) with success probability at least $2/3$ requires
\begin{equation*}
    \Omega\rbra*{
    \frac{d^2}
    {h\log^2\rbra{d}\max\cbra{1,h\log^2\rbra{d}}}}
\end{equation*}
samples of $\rho$.
\end{lemma}

\begin{proof}
Let $C,c>0$ be the universal constants in
\cref{thm:master-functional}.  Set
\[
    K=\ceil*{C\log\rbra{d+2}},
    \qquad
    L=K+1,
    \qquad
    t=\max\cbra*{\frac1{K^2},2h},
    \qquad
    p=\frac ht,
\]
\[
    a_0=0,
    \qquad
    b_0=tK^2,
    \qquad
    M=\max\cbra{1,tK^2-1}.
\]
Since $h\leq1/4$ and $K\geq2$, it holds that $0<t\leq\frac12$, $0<p\leq\frac12$, and $pt=h$.
Moreover, $tK^2\geq1$, and hence $M=\Theta\rbra{tK^2}$.
Let $(\mu_0,\mu_1)$ be the $L$-moment matching supplied by
\cref{lemma:power-moment-matching} with parameter $K$.  For
$b\in\cbra{0,1}$, draw $Y\sim\mu_b$, set $X=tY$, and let $\nu_b$
be the distribution of $X$.  Then $(\nu_0,\nu_1)$ is an
$L$-moment matching on $I=\sbra{a_0,b_0}=\sbra{0,tK^2}$ with mean $\nu=t\gamma_\alpha\leq1$.  
We apply
\cref{thm:master-functional} in dimension $d$ with $K\coloneqq L$ and $\phi\rbra{x}=\phi_\alpha\rbra{x}=x^\alpha$.
The parameters in \cref{def:master-functional-parameters} satisfy
\[
    g_\phi\geq\Delta_\alpha t^\alpha,
    \qquad
    A_\phi\leq\rbra{tK^2}^\alpha,
    \qquad
    L_\phi
    \leq
    \alpha\max\cbra*{
    2^{1-\alpha},
    \rbra{4/3}^{\alpha-1}}.
\]
The conditions of \cref{thm:master-functional} can be verified: \cref{cond1,cond2,cond3} hold for all sufficiently large $d$.
First, $L\geq C\log \rbra{d}$. 
Second, since $M\leq tK^2$ and $t\leq1/2$, $\rbra{1+M}^2\log d =O\rbra{\log^5 d} \leq o\rbra{d}$.
Finally,
\begin{align*}
    \frac{A_\phi+\rbra{b_0-a_0}L_\phi}{g_\phi}
    \leq
    O_\alpha\rbra*{
    K^{2\alpha}+K^2t^{1-\alpha}}
    \leq
    O_\alpha\rbra*{
    K^{2\max\{1,\alpha\}}}
    =o_\alpha\rbra{\sqrt d},
\end{align*}
where the second inequality uses $t\leq1$ when $0<\alpha<1$ and
$t\geq K^{-2}$ when $1<\alpha<2$.

Let $\Pi_0,\Pi_1$ and $I_0,I_1$ be the probability distributions
and intervals, respectively, specified in
\cref{thm:master-functional}.  
For $\rho\sim\Pi_b$, define the
$(d+1)$-dimensional state
\[
    \widetilde\rho
    =p\rho\oplus\rbra{1-p}\ketbra00.
\]
Since $\tr\rbra{\rho^\alpha}
    =d^{1-\alpha}\mathcal L_{\phi_\alpha}\rbra\rho$,
we have
\[
    \mathrm S_\alpha^{\mathrm T}\rbra{\widetilde\rho}
    =
    \frac{\rbra{1-p}^\alpha-1}{1-\alpha}
    +
    \frac{p^\alpha d^{1-\alpha}}{1-\alpha}
    \mathcal L_{\phi_\alpha}\rbra\rho.
\]
For $b\in\cbra{0,1}$, set
\[
    J_b=\frac{\rbra{1-p}^\alpha-1}{1-\alpha}+\frac{p^\alpha d^{1-\alpha}}{1-\alpha} I_b.
\]
It can be seen that $\Pr_{\rho \sim \Pi_b} \sbra{\mathrm S_\alpha^{\mathrm T}\rbra{\widetilde\rho} \in J_b} \geq 0.99$ and $\operatorname{dist}\rbra{J_0, J_1} > 2\varepsilon$. 
    By \cref{thm:master-functional}\eqref{main-item2}, estimating $\mathrm S_\alpha^{\mathrm T}\rbra{\widetilde\rho}$ to within additive error $\varepsilon$ requires
    \[
    \Omega\rbra*{\frac{d^2}{pM^2}} = \Omega\rbra*{
    \frac{d^2}
    {hK^2\max\cbra{1,hK^2}}}= \Omega\rbra*{
    \frac{d^2}
    {h\log^2\rbra d\max\cbra{1,h\log^2\rbra d}}}
    \]
    samples of $\widetilde\rho$. 
\end{proof}

\begin{theorem}[Sample lower bounds for Tsallis entropy estimation]
\label{cor:tsallis-regimes}
Let $\rho$ be an unknown $d$-dimensional quantum state, where $d$
is sufficiently large.

\begin{itemize}
    \item For $0<\alpha<1$ and any sufficiently small $\varepsilon > 0$, estimating $\mathrm S_\alpha^{\mathrm T}\rbra\rho$ to within
    additive error $\varepsilon$ with success probability at least $2/3$ requires
    \begin{equation}
        \Omega\rbra*{
        \frac{d^{1+1/\alpha}}
         {\varepsilon^{1/\alpha}\log^2\rbra{d}\max\cbra{1, \varepsilon^{1/\alpha}\log^2\rbra{d}}}
        +
        \frac{d^{2-2\alpha}}{\varepsilon^2}}
        \label{eq:tsallis-below-one-general}
    \end{equation}
    samples of $\rho$.

    \item For $1<\alpha<2$ and any sufficiently small $\varepsilon > 0$, there exists a constant
    $a_\alpha>0$ such that, for $d\geq1+\floor{a_\alpha \varepsilon^{-1/\rbra{\alpha-1}}}$,
    estimating
    $\mathrm S_\alpha^{\mathrm T}\rbra\rho$ to within additive
    error $\varepsilon$ with success probability at least $2/3$ requires
    \begin{equation}
        \Omega\rbra*{
        \frac{\varepsilon^{-2/\rbra{\alpha-1}}}
        {\log^4\rbra{1/\varepsilon}}}.
        \label{eq:tsallis-between-one-two-critical}
    \end{equation}
\end{itemize}
\end{theorem}

\begin{proof}
For $0<\alpha<1$, the lower bound for R\'enyi entropy estimation in \cref{cor:renyi-high-accuracy} is also a lower bound for Tsallis entropy estimation. 
On the other hand, the sample lower bound for classical Tsallis entropy estimation in \cite[Theorem~3]{FS17} gives $\Omega\rbra{{d^{2-2\alpha}}/{\varepsilon^2}}$.
Combining the two lower bounds yields the proof. 

Now suppose that $1<\alpha<2$.  Choose
$a_\alpha, \varepsilon_\alpha>0$ sufficiently small and set $d'=\floor{a_\alpha\varepsilon^{-1/\rbra{\alpha-1}}} = \Theta\rbra{\varepsilon^{-1/\rbra{\alpha-1}}}$ for every $0<\varepsilon\leq\varepsilon_\alpha$.
Applying \cref{cor:tsallis-unified} with $d \coloneqq d'$ and noting that its
parameter $h$ now satisfies
\[
    h^\alpha
    =
    \frac{32\rbra{\alpha-1}}{\Delta_\alpha}
    \varepsilon\rbra{d'}^{\alpha-1} \leq \frac{32\rbra{\alpha-1}}{\Delta_\alpha} a_\alpha^{\alpha-1} \leq \frac{1}{4^\alpha}
\]
for sufficiently small $a_\alpha > 0$, 
estimating $\mathrm S_\alpha^{\mathrm T}\rbra\rho$ to within additive error $\varepsilon$ requires 
\begin{align*}
    \Omega\rbra*{\frac{\rbra{d'}^2}{h\log^2\rbra{d'}\,\max\cbra{1,h\log^2\rbra{d'}}}}=\Omega\rbra*{\frac{\varepsilon^{-2/\rbra{\alpha-1}}}{\log^4\rbra{1/\varepsilon}}}
\end{align*}
samples of $\rho$. 
\end{proof}

\subsection{Trace Distance Estimation}

\begin{lemma}[{\cite[Lemma~30]{JHW18}}]
\label{lemma:jhw18}
Fix \(0<a\le1/2\) and define
\[
    f_a\rbra x=\frac{\abs{x-a}-a}{x},
    \qquad x>0.
\]
There are universal constants \(d_0>1\) and \(c_0,c_1\in(0,1)\) such that, for every integer
\(K\ge2\),
\[
    E_K\rbra*{f_a,[a/d_0,1]}
    \ge
    \begin{cases}
        \dfrac{c_0}{K\sqrt a},&K^{-2}\le a\le1/2,\\[1ex]
        c_1,&0<a<K^{-2}.
    \end{cases}
\]
\end{lemma}

\begin{lemma}[Moment matching for trace distance]
\label{lemma:distri-td}
    Let $K\geq2$ and $1/K^2\leq a\leq1/2$.  There is a $\rbra{K+1}$-moment matching $\rbra{\nu_0, \nu_1}$ on $\sbra{\frac{1}{4},\frac{1}{4}+\frac{3}{4a}}$ with mean $\frac{1}{4}+\frac{3}{4d_0} < 1$,
    satisfying
    \begin{align}
        \E_{Y\sim\nu_1}\sbra*{\phi_{\mathrm T}\rbra{Y}}
        -\E_{Y\sim\nu_0}\sbra*{\phi_{\mathrm T}\rbra{Y}}
        \geq
        \Delta_{\mathrm T}
        \coloneqq\frac{3c_0}{4d_0K\sqrt a},
        \label{eq:td-matching-gap}
    \end{align}
    where $\phi_{\mathrm T}\rbra{x}=\max\cbra{1-x,0}$.
\end{lemma}

\begin{proof}
    By \cref{fact:mm,lemma:jhw18}, there is a $K$-moment matching $\rbra{\mu_0,\mu_1}$ on $\sbra{a/d_0,1}$ satisfying
    \begin{align}
        \E_{X\sim\mu_1}\sbra*{f_a\rbra{X}}
        -\E_{X\sim\mu_0}\sbra*{f_a\rbra{X}}
        \geq\frac{2c_0}{K\sqrt a}.
        \label{eq:td-base-gap}
    \end{align}
    Let $\zeta_b
        =\mathsf T_{a/d_0,a/d_0}\rbra{\mu_b}$ for $b\in\cbra{0,1}$.
    By \cref{lemma:T-eta-alpha,lemma:moment-lifting}, $\rbra{\zeta_0,\zeta_1}$ is a $\rbra{K+1}$-moment matching with mean $a/d_0$.
    Let $Z \sim \zeta_b$, define $Y=\frac{1}{4}+\frac{3Z}{4a}$
    and let $\nu_b$ be its probability distribution. 
    Since $0\leq Z\leq1$, $\nu_b$ is in $\sbra{\frac{1}{4},\frac{1}{4}+\frac{3}{4a}}$.  
    Note that $\rbra{\nu_0, \nu_1}$ is a $\rbra{K+1}$-moment matching with mean $\frac{1}{4} + \frac{3}{4d_0}$. 

    For $x>0$, it can be verified that
    \begin{equation}
        \frac{1}{x}
        \rbra*{\phi_{\mathrm T}\rbra*{\frac14+\frac{3x}{4a}}
        -\phi_{\mathrm T}\rbra*{\frac14}}
        =\frac{3}{8a}\rbra*{f_a\rbra{x}-1}.
        \label{eq:td-transform-identity}
    \end{equation}
    Using the definition of $\mathsf T_{a/d_0,a/d_0}$,
    \begin{align*}
        \E_{Y\sim\nu_b}\sbra*{\phi_{\mathrm T}\rbra{Y}}
        &=\phi_{\mathrm T}\rbra*{\frac14}
        +\frac{a}{d_0}
        \E_{X\sim\mu_b}\sbra*{
        \frac{1}{X}
        \rbra*{\phi_{\mathrm T}\rbra*{\frac14+\frac{3X}{4a}}
        -\phi_{\mathrm T}\rbra*{\frac14}}}\\
        &=\phi_{\mathrm T}\rbra*{\frac14}
        +\frac{3}{8d_0}
        \E_{X\sim\mu_b}\sbra*{f_a\rbra{X}-1},
    \end{align*}
    which gives \cref{eq:td-matching-gap} using
    \cref{eq:td-base-gap}.
\end{proof}

\begin{theorem}[Sample lower bound for trace distance estimation]
\label{cor:trace-lower-bound}\label{corollary:td}
    There is a universal constant $\varepsilon_0>0$ such that, for all sufficiently large $d$ and every $0<\varepsilon\leq\varepsilon_0$, estimating $\mathrm T\rbra{\rho,\mathbb{I}/d}$ to within additive error $\varepsilon$ with success probability at least $2/3$ requires
    \begin{equation*}
        \Omega\rbra*{
        \frac{d^2}{\varepsilon^2\log^2\rbra{d}
        \max\cbra{1,\varepsilon^2\log^2\rbra{d}}}}
    \end{equation*}
    samples of $\rho$.
\end{theorem}

\begin{proof}
    Let $c,C>0$ be the universal constants in \cref{thm:master-functional}. 
    Take $\varepsilon_0=c_0/(64d_0)$, where $c_0$ and $d_0$ are the constants in \cref{lemma:jhw18}.
    For $0 < \varepsilon \leq \varepsilon_0$, we apply
    \cref{thm:master-functional} with dimension $d$, function
    $\phi\rbra{x}=\phi_{\mathrm T}\rbra{x}=\max\cbra{1-x,0}$, and the following parameters:
    \[
        K=\ceil*{C\log\rbra{d+2}},
        \qquad
        a=\min\cbra*{
        \frac14,
        \frac{c_0^2}{4096d_0^2\varepsilon^2K^2}},
        \qquad
        \Delta_{\mathrm T}=\frac{3c_0}{4d_0K\sqrt a},
        \qquad
        t=\frac{24\varepsilon}{\Delta_{\mathrm T}}.
    \]
    It can be shown that for sufficiently large $d$, it holds that $K\geq2$, $K^{-2}\leq a\leq1/4$, and $t \leq 1/2$.

    Let $(\mu_0,\mu_1)$ be the $(K+1)$-moment matching given by
    \cref{lemma:distri-td}. 
    Draw $X_b \sim \mu_b$ and let $Y_b = 1 - t + tX_b$. 
    Denote $\nu_b$ as the distribution of $Y_b$. 
    Then, it can be shown that $\rbra{\nu_0, \nu_1}$ is a $\rbra{K+1}$-moment matching on $I = \sbra{a_0, b_0} = \sbra{1 - \frac{3t}{4}, 1+\frac{3t\rbra{1-a}}{4a}}$ with mean $\nu = 1 - t + t\rbra{\frac{1}{4}+\frac{3}{4d_0}} \leq 1$. 
    Since $t \leq 1/2$, we have $a_0 \geq 5/8$ and thus $I \subseteq [0, +\infty)$. 
    Direct calculation shows
    \[
    M \leq \frac{3t}{4a} = \Theta\rbra{\varepsilon K \max\cbra{1, \varepsilon K}}, \qquad A_{\phi_{\mathrm{T}}} \leq \frac{3t}{4}, \qquad L_{\phi_{\mathrm{T}}} \leq 1, \qquad g_{\phi_{\mathrm{T}}} \geq 24\varepsilon, \qquad b_0 - a_0 = \frac{3t}{4a},
    \]
    which ensures that the conditions in \cref{thm:master-functional}, \cref{cond1,cond2,cond3}, hold. 

    By \cref{thm:master-functional}\eqref{main-item1}, estimating $\mathcal{L}_{\phi_{\mathrm{T}}}\rbra{\rho} = \mathrm{T}\rbra{\rho, \mathbb{I}/d}$ to within additive error $g_{\phi_{\mathrm{T}}}/24 \geq \varepsilon$ requires
    \[
    \Omega\rbra*{\frac{d^2}{M^2}} = \Omega\rbra*{ \frac{d^2}{\varepsilon^2\log^2\rbra d \max\cbra{1,\varepsilon^2\log^2\rbra d}}}
    \]
    samples of $\rho$. 
\end{proof}

\subsection{Spectrum Estimation}

Let \(\lambda^\downarrow\rbra\rho\) denote the nonincreasingly ordered spectrum of \(\rho\), and let
\[
    d_{\mathrm{TV}}\rbra{p,q}=\frac12\Abs{p-q}_1.
\]

\begin{theorem}[Sample lower bound for spectrum estimation]
\label{cor:spectrum-lower-bound}
For sufficiently large $d$ and sufficiently small $\varepsilon > 0$, any estimator that, on input every $d$-dimensional quantum state $\rho$, outputs $\widehat{\lambda}$ such that $d_{\textup{TV}}\rbra{\widehat\lambda, \lambda^\downarrow\rbra\rho} \leq \varepsilon$ with success probability at least $2/3$ requires 
\begin{equation*}
    \Omega\rbra*{
    \frac{d^2}{\varepsilon^2\log^2\rbra{d}
    \max\cbra{1,\varepsilon^2\log^2\rbra{d}}}}
\end{equation*}
copies.
\end{theorem}

\begin{proof}
Let $u_d=\rbra{1/d,\ldots,1/d}$ and note that $d_{\mathrm{TV}}\rbra{\lambda^\downarrow\rbra\rho,u_d}=\mathrm T\rbra{\rho,\mathbb{I}/d}$. 
Let $\widehat{\lambda}$ be the output spectrum estimation, then $\widehat{T} = d_{\mathrm{TV}}\rbra{\widehat\lambda,u_d}$ is an estimate of $T\rbra{\rho,\mathbb{I}/d}$ within additive error $\varepsilon$. 
Therefore, the lower bound for trace distance estimation in \cref{cor:trace-lower-bound} is also a lower bound for spectrum estimation. 
\end{proof}

\subsection{Uhlmann Fidelity Estimation}

\begin{lemma}[Approximation of the inverse square root]
\label{lemma:inverse-square-root}
    Let $K\geq2$ and $\eta=1/(32K^2)$.  Then
    \begin{equation*}
        E_K\rbra*{x^{-1/2},\sbra{\eta,1}}
        \geq\frac{1}{16\sqrt\eta}.
    \end{equation*}
\end{lemma}

\begin{proof}
    Suppose, toward a contradiction, that a polynomial $P$ of degree at most $K$ satisfies $\abs{P\rbra{x}-x^{-1/2}} < 1/\rbra{16\sqrt{\eta}}$ for $x\in\sbra{\eta,1}$.
    Then, $P\rbra{\eta}>{15}/\rbra{16\sqrt\eta}$ and $P\rbra{4\eta}<{9}/\rbra{16\sqrt\eta}$.
    Hence, $P\rbra{\eta}-P\rbra{4\eta}>{3}/\rbra{8\sqrt\eta}$.
    By the mean-value theorem, there exists
    $\xi\in\rbra{\eta,4\eta}$ such that
    \begin{equation}
        \abs*{P'\rbra{\xi}}
        >\frac{1}{8\eta^{3/2}} > \frac{3K^2}{\sqrt{\eta}}.
        \label{eq:inverse-sqrt-derivative-lower}
    \end{equation}
    On the other hand, Markov brothers' inequality \cite{Mar90,Mar92} on
    $\sbra{\eta,1}$ gives
    \begin{align*}
        \sup_{x\in\sbra{\eta,1}}\abs*{P'\rbra{x}}
        \leq\frac{2K^2}{1-\eta}
        \sup_{x\in\sbra{\eta,1}}\abs*{P\rbra{x}} \leq\frac{2K^2}{1-\eta} \cdot 
        \frac{17}{16\sqrt\eta}
        <\frac{3K^2}{\sqrt\eta},
    \end{align*}
    contradicting
    \cref{eq:inverse-sqrt-derivative-lower}.
\end{proof}

\begin{lemma}[Moment matching for Uhlmann fidelity]
\label{lemma:distri-fidelity}
    Let $K\geq2$, $\eta=1/(32K^2)$,
    $1\leq M\leq4K^2$, and $\alpha=\eta M$.  There is a $\rbra{K+1}$-moment matching $\rbra{\nu_0,\nu_1}$ on $\sbra{0,M}$ with mean $\alpha \leq 1/8$ satisfying
    \begin{align*}
        \E_{Y\sim\nu_1}\sbra*{\sqrt Y}
        -\E_{Y\sim\nu_0}\sbra*{\sqrt Y}
        \geq
        \Delta_{\mathrm F}
        \coloneqq\frac{\sqrt M}{48K}.
    \end{align*}
\end{lemma}

\begin{proof}
    By \cref{fact:mm,lemma:inverse-square-root}, there is a $K$-moment matching $\rbra{\mu_0, \mu_1}$ on $\sbra{\eta,1}$ satisfying
    \begin{align*}
        \E_{X\sim\mu_1}\sbra*{X^{-1/2}}
        -\E_{X\sim\mu_0}\sbra*{X^{-1/2}}
        \geq\frac{1}{8\sqrt\eta}.
    \end{align*}
    Let $\nu_b=\mathsf T_{\eta,\alpha}\rbra{\mu_b}$.
    Since $\alpha/\eta=M$, $\nu_b$ is on $\sbra{0,M}$.  
    By \cref{lemma:moment-lifting}, $\rbra{\nu_0, \nu_1}$ is a $\rbra{K+1}$-moment matching with mean $\alpha$.
    Direct calculation shows
    \begin{align*}
        \E_{Y\sim\nu_b}\sbra*{\sqrt Y}
        =\E_{X\sim\mu_b}\sbra*{
        \frac{\eta}{X}\sqrt{\frac{\alpha X}{\eta}}}
        =\sqrt{\alpha\eta}
        \E_{X\sim\mu_b}\sbra*{X^{-1/2}},
    \end{align*}
    which gives
    \begin{align*}
        \E_{Y\sim\nu_1}\sbra*{\sqrt Y}
        -\E_{Y\sim\nu_0}\sbra*{\sqrt Y}
        \geq\sqrt{\alpha\eta}\cdot\frac{1}{8\sqrt\eta}
        =\frac{\sqrt\alpha}{8}
        =\frac{\sqrt M}{8\sqrt{32}K}
        \geq\frac{\sqrt M}{48K}.
    \end{align*}
\end{proof}

\begin{theorem}[Sample lower bound for Uhlmann fidelity estimation]
\label{cor:fidelity-lower-bound}
    There is a universal constant $\varepsilon_0>0$ such that, for all sufficiently large $d$ and every $0<\varepsilon\leq\varepsilon_0$, estimating $\mathrm F\rbra{\rho,\mathbb{I}/d}$
    to within additive error $\varepsilon$ with success probability at least $2/3$ requires
    \begin{equation*}
        \Omega\rbra*{
        \frac{d^2}{\varepsilon^2\log^2\rbra{d}
        \max\cbra{1,\varepsilon^2\log^2\rbra{d}}}}
    \end{equation*}
    samples of $\rho$.
\end{theorem}

\begin{proof}
Let $C, c > 0$ be the constants in \cref{thm:master-functional}. 
Take $\varepsilon_0=2^{-10}$.
Set
\[
    d'=\floor*{\frac d2},
    \qquad
    K=\ceil*{C\log\rbra{d+2}},
    \qquad
    L=K+1,\qquad
    \eta = \frac{1}{32K^2},
\]
\[
    M=\max\cbra{1,2^{22}\varepsilon^2K^2},
    \qquad
    \alpha = \eta M,
    \qquad
    p=\frac{2^{21}\varepsilon^2K^2}{M}.
\]
Since \(\varepsilon\leq2^{-10}\), one has $1\leq M\leq4K^2$ and $0<p\leq1/2$.
Let \(\rbra{\nu_0,\nu_1}\) be $L$-moment matching in \cref{lemma:distri-fidelity}.  
Apply \cref{thm:master-functional} with $d \coloneqq d'$, $K \coloneqq L$, and $\phi\rbra{x}\coloneqq\phi_{\mathrm{F}}\rbra{x}=\sqrt{x}$.
It can be verified that the conditions in \cref{thm:master-functional}, \cref{cond1,cond2,cond3}, hold. 

Let $\Pi_0, \Pi_1$ and $I_0, I_1$ be the probability distributions and intervals, respectively, specified in \cref{thm:master-functional}. 
For $\rho \sim \Pi_b$, define
\[
\widetilde\rho = p\rho \oplus \rbra{1-p} \frac{\mathbb{I}_{d-d'}}{d-d'}.
\]
Note that
\[
\mathrm{F}\rbra*{\widetilde{\rho}, \frac{\mathbb{I}_d}{d}} = \sqrt{\frac{pd'}{d}} \mathcal{L}_{\phi_{\mathrm{F}}}\rbra{\rho} + \sqrt{\frac{\rbra{1-p}\rbra{d-d'}}{d}}.
\]
Take
\[
J_b = \sqrt{\frac{pd'}{d}} I_b + \sqrt{\frac{\rbra{1-p}\rbra{d-d'}}{d}}.
\]
It can be seen that $\Pr_{\rho \sim \Pi_b} \sbra{\mathrm{F}\rbra{\widetilde{\rho}, \mathbb{I}_d/d} \in J_b} \geq 0.99$ and $\operatorname{dist}\rbra{J_0, J_1} > 2\varepsilon$. 
By \cref{thm:master-functional}\eqref{main-item2}, estimating $\mathrm{F}\rbra{\widetilde{\rho}, \mathbb{I}_d/d}$ to within additive error $\varepsilon$ requires
\[
\Omega\rbra*{\frac{d'^2}{pM^2}} = \Omega\rbra*{ \frac{d^2}{\varepsilon^2\log^2\rbra d \max\cbra{1,\varepsilon^2\log^2\rbra d}}}
\]
samples of $\widetilde\rho$. 
\end{proof}

\subsection{Quantum Hellinger Distance Estimation}\label{sec:hellinger-estimation}

\begin{theorem}[Sample lower bound for Hellinger distance estimation]\label{cor:hellinger-lower-bound}
There is a universal constant \(\varepsilon_0>0\) such that, for all sufficiently large \(d\) and \(0<\varepsilon\le\varepsilon_0\), estimating either \(\mathrm d_{\mathrm H}\rbra{\rho,\mathbb{I}/d}\) or \(\mathrm \mathrm d_{\mathrm H}^2\rbra{\rho,\mathbb{I}/d}\) to within additive error \(\varepsilon\), with success probability at least \(2/3\), requires
\begin{equation}
    \Omega\rbra*{
    \frac{d^2}{\varepsilon^2\log^2\rbra{d}
    \max\{1,\varepsilon^2\log^2\rbra{d}\}}}
\end{equation}
samples of $\rho$. 
\end{theorem}
\begin{proof}
    This is straightforward by \cref{cor:fidelity-lower-bound} and noting that $\mathrm d_{\mathrm H}^2\rbra{\rho, \mathbb{I}/{d}} = 2\rbra{1 - \mathrm{F}\rbra{\rho, \mathbb{I}/{d}}}$.
\end{proof}

\subsection{Rank Testing}

For an integer \(r\geq1\), define the distance of $\rho$ from quantum states of rank at most $r$ by
\begin{equation}
    \Delta_r\rbra\rho
    \coloneqq
    \inf_{\rank\rbra\sigma\leq r}
    \mathrm T\rbra{\rho,\sigma} = \sum_{j>r}\lambda_j^\downarrow\rbra\rho,
    \label{eq:rank-distance-definition}
\end{equation}
where \(\lambda_1^\downarrow\rbra\rho\geq\cdots\geq
\lambda_d^\downarrow\rbra\rho\) are the eigenvalues of $\rho$. 

\begin{lemma}[Moment matching for rank testing, cf. {\cite[Appendix A]{WY19}}]
\label{lemma:rank-support-matching}
For every integer \(K\geq2\), let \(M=100K^2\).  There is a $K$-moment matching \(\rbra{\mu_0,\mu_1}\) on $\cbra{0}\cup\sbra{1,M}$ with mean $1$
such that
\begin{equation}
    \Pr_{Y\sim\mu_1}\sbra{Y>0}
    -\Pr_{Y\sim\mu_0}\sbra{Y>0}
    \geq\frac14.
    \label{eq:rank-matching-support-gap}
\end{equation}
\end{lemma}

\begin{proof}
We first prove a concrete approximation lower bound.  Suppose that a polynomial \(P\) of degree at most \(K-1\) satisfies $\sup_{x\in[1,M]}\abs*{P\rbra x-1/x}<1/8$.
Then $P\rbra1>7/8$ and $P\rbra4<1/4+1/8=3/8$.
Hence \(P\rbra1-P\rbra4>1/2\), and the mean-value theorem gives a point \(\xi\in(1,4)\) with $\abs{P'\rbra\xi}>1/6$.
On the other hand, Markov brothers' inequality \cite{Mar90,Mar92} on \([1,M]\) gives
\begin{align*}
    \sup_{x \in \sbra{1, M}} \abs{P'\rbra{x}}
    \leq\frac{2\rbra{K-1}^2}{M-1}
    \sup_{x \in \sbra{1,M}} \abs{P\rbra{x}} 
    \leq\frac{2K^2}{100K^2-1}\cdot\frac98
    <\frac16,
\end{align*}
which gives a contradiction. Thus
\begin{equation*}
    E_{K-1}\rbra*{1/x,\sbra{1,M}}
    \geq\frac18.
\end{equation*}
By \cref{fact:mm}, there is a $\rbra{K-1}$-moment matching \(\nu_0,\nu_1\) on \([1,M]\) such that
\begin{equation*}
    \E_{X\sim\nu_1}\sbra{1/X}
    -\E_{X\sim\nu_0}\sbra{1/X}
    \geq\frac14.
\end{equation*}
For \(b\in\cbra{0,1}\), define \(Y\) by putting, conditionally on \(X\sim\nu_b\),
\[
    Y=
    \begin{cases}
        X,&\text{with probability }1/X,\\
        0,&\text{with probability }1-1/X.
    \end{cases}
\]
Let $\mu_b$ be the probability distribution of $Y$ when $X \sim \nu_b$. 
Then, $\rbra{\mu_0, \mu_1}$ is a $K$-moment matching with mean $1$, and $\Pr_{Y \sim \mu_b}\sbra{Y > 0} = \E_{X \sim \nu_b}\sbra{1/X}$ yields the proof.
\end{proof}

\begin{lemma}[Rank-testing transfer]
\label{cor:master-rank-transfer}
There exist universal constants $C,c>0$ such that the following holds. 
Let $d$ be sufficiently large, and set $q=\floor*{d/4}$.
Suppose that $(\nu_0,\nu_1)$ is a $K$-moment matching on $\cbra{0} \cup [\ell,b_0]$ with mean $\leq 1$ for some $0 < \ell \leq 1/2$ and $\gamma = \pi_1 - \pi_0 \geq 1/4$, where $\pi_b=\Pr_{X\sim\nu_b}\sbra{X>0}$. 
Denote $M=\max\cbra{1, b_0-1}$,
and suppose that $K\geq C\log \rbra{d}$ and $(1+M)^2\log \rbra{d}\leq cd$.  
Let $r_0$ be an integer satisfying
\begin{equation}
    d-q+\pi_0q+\frac{\gamma q}{4}+1
    \leq r_0
    \leq
    d-q+\pi_1q-\frac{\gamma q}{4}+1.
    \label{eq:rank-transfer-threshold}
\end{equation}
Then, for every $0<p\leq1/2$, distinguishing whether an unknown $(d+1)$-dimensional quantum state $\widetilde\rho$ satisfies $\rank\rbra{\widetilde\rho}\leq r_0$ or $\Delta_{r_0}\rbra{\widetilde{\rho}} \geq cp\ell$ requires $\Omega\rbra{{d^2}/\rbra{pM^2}}$ samples of $\widetilde\rho$.
\end{lemma}

\begin{proof}
Let $r=d-q$.  For sufficiently small $c$ and sufficiently large $d$, the assumptions imply $d/6\leq q\leq d/3$ and $b_0\leq d/12$.  Let $\Pi_0,\Pi_1$ be the probability distributions given by \cref{def:pi-01}.  The calculation in \cref{eq:second-neq} (in the proof of \cref{thm:master-functional}), which uses only the two assumptions $K\geq C\log \rbra{d}$ and $(1+M)^2\log \rbra{d}\leq cd$, gives
\begin{equation}
    \mathrm T\rbra*{
    \E_{\rho\sim\Pi_0}\sbra{\rho^{\otimes n}},
    \E_{\rho\sim\Pi_1}\sbra{\rho^{\otimes n}}}
    \leq0.01,
    \qquad
    1\leq n\leq n_0\coloneqq\floor*{\frac{cd^2}{M^2}}.
    \label{eq:rank-transfer-closeness}
\end{equation}
For $b \in \cbra{0, 1}$, let $X_1,\ldots,X_q \sim \nu_b$, and set $N_b=\sum_{i=1}^q\mathbf1_{\{X_i>0\}}$ and $F_b=\{\sum_{i=1}^qX_i\leq d/2\}$.  
Define
$E_0=F_0\cap\cbra*{N_0\leq\pi_0q+{\gamma q}/{8}}$ and $E_1=F_1\cap\cbra*{N_1\geq\pi_1q-{\gamma q}/{8}}$.
Hoeffding's inequality, together with $q=\Theta(d)$, $\gamma\geq1/4$, and $(1+M)^2\log d\leq cd$, gives $\Pr_{\rho\sim\Pi_b}\sbra{E_b}\geq0.99$.

On $F_b$, the rank of the state $\rho \sim \Pi_b$ is $\rank\rbra\rho=r+N_b$.  
By \cref{eq:rank-transfer-threshold}, on $E_0$,
\[
    \rank\rbra\rho+1
    \leq r+\pi_0q+\frac{\gamma q}{8}+1
    \leq r_0.
\]
On $E_1$, $\rank\rbra{\rho} - \rbra{r_0-1} = N_1-(r_0-1-r)\geq\gamma q/8$ and any non-zero eigenvalue of $\rho$ is at least $\ell/d$; hence, Ky Fan's maximum principle \cite{Fan49} (cf.\ \cite[Exercise II.1.13]{Bha97}) gives $\Delta_{r_0-1}\rbra{\rho} \geq \frac{\gamma q \ell}{8d} \geq c \ell$. 

For $\rho\sim\Pi_b$, define 
\[
\widetilde\rho=p\rho\oplus(1-p)\ketbra00.
\]
On $E_0$, $\rank\rbra{\widetilde\rho}\leq r_0$.  
On $E_1$, again by Ky Fan's maximum principle, we have $\Delta_{r_0}\rbra{\widetilde\rho} \geq cp\ell$. 
Suppose that a tester uses $n\leq cn_0/p$ samples and distinguishes whether an unknown quantum state $\widetilde\rho$ has rank $\leq r_0$ or $\Delta_{r_0}\rbra{\widetilde\rho} \geq cp\ell$ with success probability at least $2/3$. 
Then, using this tester, we can distinguish $b \in \cbra{0, 1}$ with success probability $p_{\textup{succ}} \geq \frac{2}{3} \cdot \Pr_{\rho \sim \Pi_b}\sbra{E_b} \geq 0.66$. 
On the other hand, \cref{eq:rank-transfer-closeness}, with Markov's inequality, gives
\begin{align}
    \mathrm T\rbra*{
    \E_{\rho\sim\Pi_0}\sbra{\widetilde\rho^{\otimes n}},
    \E_{\rho\sim\Pi_1}\sbra{\widetilde\rho^{\otimes n}}} 
    & = \E_{J\sim\operatorname{Bin}\rbra{n,p}}
    \sbra*{\mathrm T\rbra*{
    \E_{\rho\sim\Pi_0}\sbra*{\rho^{\otimes J}},
    \E_{\rho\sim\Pi_1}\sbra*{\rho^{\otimes J}}}} \nonumber \\
    &\leq0.01+\Pr_{J\sim\operatorname{Bin}(n,p)}\sbra{J>n_0}
    \leq0.01+c. \label{eq:direct-sum}
\end{align}
Hence the Helstrom--Holevo bound \cite{Hel67,Hol73} gives $p_{\mathrm{succ}}\leq\frac12(1+0.01+c)$, which contradicts $p_{\textup{succ}} \geq 0.66$ for sufficiently small $c > 0$. 
Therefore,
\[
    n=\Omega\rbra*{\frac{n_0}{p}}
    =\Omega\rbra*{\frac{d^2}{pM^2}}.
\]
\end{proof}

\begin{theorem}[Sample lower bound for rank testing]
\label{cor:rank-lower-bound}
There are universal constants \(c,\varepsilon_0>0\) such that the following holds.  
For sufficiently large \(r\) and \(0<\varepsilon\leq\varepsilon_0\), distinguishing whether $\rank\rbra\rho\leq r$ or $\Delta_r\rbra\rho\geq\varepsilon$
with success probability at least \(2/3\) requires
\begin{equation*}
    \Omega\rbra*{
    \frac{r^2}{\varepsilon\log^2\rbra r\,
    \max\cbra{1,\varepsilon\log^2\rbra r}}}
\end{equation*}
samples of $\rho$.
\end{theorem}

\begin{proof}
Let $C,c>0$ be the universal constants in Corollary~\ref{cor:master-rank-transfer}.  Take $\varepsilon_0=c/4$ and set
\[
    K=\ceil*{2C\log\rbra{r+2}},
    \qquad
    t=\max\cbra*{\frac1{100K^2},\frac{2\varepsilon}{c}},
    \qquad
    p=\frac{\varepsilon}{ct}.
\]
Then, $0<t\leq1/2$, $0<p\leq1/2$, and $cpt=\varepsilon$.

Let $(\mu_0,\mu_1)$ be the $K$-moment matching from Lemma~\ref{lemma:rank-support-matching}.  Denote $\pi_b=\Pr_{X\sim\mu_b}\sbra{X>0}$, $\gamma=\pi_1-\pi_0\geq1/4$, and $\pi_*=(\pi_0+\pi_1)/2$.  Choose an integer $q$ nearest to $(r-1)/(3+\pi_*)$ and set $D=4q=\Theta(r)$.
Then, $D+1\leq2r+1$, and $\abs*{r-1-(3+\pi_*)q}\leq2$.  Since $\gamma\geq1/4$, for all sufficiently large $r$,
\begin{equation}
    3q+\pi_0q+\frac{\gamma q}{4}+1
    \leq r
    \leq
    3q+\pi_1q-\frac{\gamma q}{4}+1.
    \label{eq:rank-threshold-location}
\end{equation}

For $b\in\cbra{0,1}$, let $\nu_b$ be the distribution of $tX$ when $X\sim\mu_b$.  Then $(\nu_0,\nu_1)$ is a $K$-moment matching on $I=[0,100tK^2]$ with mean $t\leq1$, and $\supp\rbra{\nu_0}\cup\supp\rbra{\nu_1}\subseteq\cbra0\cup\sbra{t,100tK^2}$.
Since $t\geq1/(100K^2)$, $M=\sup_{x\in I}\abs{x-1}=\Theta(tK^2)$.  Moreover, $K\geq C\log \rbra{D}$ and $(1+M)^2\log \rbra{D}=O(\log^5r)=o(D)$.
By Corollary~\ref{cor:master-rank-transfer} with $d=D$, $r_0=r$, and $\ell=t$, distinguishing whether $\rank\rbra{\widetilde\rho} \leq r$ or $\Delta_{r}\rbra{\widetilde\rho} \geq cp\ell = \varepsilon$ requires 
\begin{align*}
    \Omega\rbra*{\frac{D^2}{pM^2}}
    =\Omega\rbra*{\frac{r^2}{\varepsilon tK^4}}
    =\Omega\rbra*{
    \frac{r^2}{\varepsilon \log^2\rbra{r}
    \max\cbra{1,\varepsilon \log^2\rbra{r}}}}
\end{align*}
samples of $\widetilde\rho$. 
\end{proof}

\subsection{Schmidt Rank Testing}

The Schmidt rank of a bipartite pure state $\ket{\psi} \in \mathcal{H}_{\mathsf{A}} \otimes \mathcal{H}_{\mathsf{B}} \simeq \mathbb{C}^d \otimes \mathbb{C}^d$ is defined to be the rank of its reduced density operator, i.e., $\operatorname{SR}\rbra{\ket{\psi}} = \rank\rbra{\tr_{\mathsf{B}}\rbra{\ketbra{\psi}{\psi}}}$.

\begin{theorem}[Sample lower bound for Schmidt rank testing]
\label{cor:mps-lower-bound}
There are universal constants $c,\varepsilon_0>0$ such that the following holds.  
For sufficiently large $r$ and $0<\varepsilon\leq\varepsilon_0$, distinguishing whether an unknown bipartite pure state $\ket\psi$ has Schmidt rank $\leq r$ or $\ket{\psi}$ is $\varepsilon$-far in trace distance from any bipartite pure state of Schmidt rank at most $r$ with success probability at least \(2/3\) requires
\begin{equation*}
    \Omega\rbra*{
    \frac{r^2}{\varepsilon^2\log^2\rbra r\,
    \max\cbra{1,\varepsilon^2\log^2\rbra r}}}
\end{equation*}
samples of $\ket\psi$. 
\end{theorem}

\begin{proof}
This is directly done by recalling \cite[Corollary IV.1]{CWZ26} that the sample complexity of Schmidt rank testing with parameter $r$ and $\varepsilon$ is at least the sample complexity of rank testing with parameter $r \coloneqq r$ and $\varepsilon \coloneqq \varepsilon^2$.
Therefore, applying \cref{cor:rank-lower-bound} with $\varepsilon$ replaced by $\varepsilon^2$ gives the lower bound
\[
    \Omega\rbra*{
    \frac{r^2}{\varepsilon^2\log^2\rbra r\,
    \max\cbra{1,\varepsilon^2\log^2\rbra r}}}.
\]
\end{proof}

\subsection{Matrix Product State Testing}

For $n\geq2$, let $\mathrm{MPS}_n\rbra r$ denote the set of $n$-partite matrix product states $\ket{\psi}$ of bond dimension $r$ (with open boundary condition \cite{PGVWC07}). 
Formally, each $\ket{\psi} \in \mathrm{MPS}_n\rbra r$ can be written as 
\[
\ket{\psi} = \sum_{i_1,\ldots,i_n} A^{(1)}_{i_1}\cdots A^{(n)}_{i_n} \ket{i_1\cdots i_n},
\]
where \(A_i^{(j)} \in \mathbb{C}^{r \times r}\) for \(2\leq j\leq n-1\), $A_i^{\rbra{1}} \in \mathbb{C}^{1 \times r}$, and $A_i^{\rbra{n}} \in \mathbb{C}^{r \times 1}$. 
Define
\[
    \Delta_{\mathrm{MPS}_n\rbra r}\rbra{\ket\psi}
    \coloneqq
    \inf_{\ket\phi\in\mathrm{MPS}_n\rbra r}
    \mathrm T\rbra*{\ketbra\psi\psi,\ketbra\phi\phi}.
\]

\begin{theorem}[Sample lower bound for matrix product state testing]
\label{cor:mps-n-partite-lower-bound}
There are universal constants $c,C>0$ such that the following holds. 
Let $n,r$ be sufficiently large, $0<\varepsilon\leq1/2$ such that $r \geq C \log\rbra{n}$ and $\varepsilon^2 \log^2\rbra{nr} \leq c n$.
Then, testing whether an unknown $n$-partite pure state is in $\mathrm{MPS}_n\rbra r$ or is $\varepsilon$-far from $\mathrm{MPS}_n\rbra r$ with success probability at least \(2/3\) requires
\begin{equation*}
    \Omega\rbra*{\frac{nr^2}{\varepsilon^2\log^2\rbra{nr}}}
\end{equation*}
samples.
\end{theorem}

\begin{proof}
Choose \(C>0\) sufficiently large and then \(c>0\) sufficiently small
such that \(24C^2c\leq1/25600\).  Set
\[
    \ell=\floor*{\frac n2},
    \qquad
    K=\ceil*{C\log\rbra{nr}},
    \qquad
    \delta=1-\rbra{1-\varepsilon^2}^{1/\ell}, \qquad t=\frac1{100K^2}.
\]
For \(0<\varepsilon\leq1/2\), it holds that ${\varepsilon^2}/{\ell}\leq\delta\leq{2\varepsilon^2}/{\ell}$.
Since \(\ell\geq n/3\) and \(K\leq2C\log\rbra{nr}\) for sufficiently
large \(n,r\), we have $\delta K^2 \leq 24C^2c \leq 1/25600$. 

Let \((\mu_0,\mu_1)\) be the \(K\)-moment matching from
\cref{lemma:rank-support-matching}.  Write
\[
    \pi_b=\Pr_{Y\sim\mu_b}\sbra{Y>0},
    \qquad
    \gamma=\pi_1-\pi_0\geq\frac14,
    \qquad
    \pi_*=\frac{\pi_0+\pi_1}{2}, \qquad p=\frac{32\delta}{\gamma t} \leq \frac{1}{2}.
\]
Choose \(q\in\mathbb N\) such that $\abs*{r-1-\rbra{3+\pi_*}q}\leq2$ and set $D = 4q = \Theta\rbra{r}$. 
For sufficiently large \(r\),
\begin{equation}
    3q+\pi_0q+\frac{\gamma q}{4}+1
    \leq r
    \leq
    3q+\pi_1q-\frac{\gamma q}{4}+1.
    \label{eq:mps-rank-threshold}
\end{equation}

For \(b\in\cbra{0,1}\), let \(\nu_b\) be the distribution of \(tY\)
for \(Y\sim\mu_b\). 
Then $\supp\rbra{\nu_b} \subseteq \cbra0\cup\sbra{t,1}$, and \((\nu_0,\nu_1)\) is a \(K\)-moment matching with mean $t$. 
Let $\Pi_0$ and $\Pi_1$ be the probability distributions defined in \cref{def:pi-01} with $d = D$, $q = D/4$, $\nu = t$, and $I = \sbra{0, 1}$. 
For \(b\in\cbra{0,1}\), let \(X_1,\ldots,X_q\sim\nu_b\) be the random
variables, and let $N_b=\sum_{i=1}^q\mathbf1_{\{X_i>0\}}$.
Since \(0\leq X_i\leq1\), it holds that $\sum_{i=1}^q X_i \leq q = D/4$ with probability $1$.
Consider the events $E_0=\cbra{N_0\leq\pi_0q+{\gamma q}/{8}}$ and $E_1=\cbra{N_1\geq\pi_1q-{\gamma q}/{8}}$.
By Hoeffding's inequality,
\cref{eq:mps-rank-threshold}, and the proof of
\cref{cor:master-rank-transfer}, if $\rho_b \sim \Pi_b$ for $b \in \cbra{0, 1}$, then 
\begin{itemize}
    \item For $b \in \cbra{0, 1}$, $\Pr\sbra{E_b^{\mathsf{c}}} \leq e^{-\gamma^2 q/32}$, 
    \item On $E_0$, $\rank\rbra{\rho_0} \leq r - 1$, 
    \item On $E_1$, $\Delta_{r-1}\rbra{\rho_1} \geq \frac{\gamma q}{8}\frac{t}{D}
    =\frac{\gamma t}{32}$. 
\end{itemize}
Let \(m_0=\floor{cD^2}\). 
For \(1\leq m\leq m_0\),
\cref{eq:second-neq}, with $d \coloneqq D$ and $M \coloneqq \sup_{x\in[0,1]}\abs{x-1}=1$, gives
\begin{equation} \label{eq:td-pi0-pi1}
    \mathrm T\rbra*{
        \E_{\rho\sim\Pi_0}\sbra{\rho^{\otimes m}},
        \E_{\rho\sim\Pi_1}\sbra{\rho^{\otimes m}}}
    \leq
    q \rbra*{4^{-K}+e^{-D/384}}.
\end{equation}

For \(b\in\cbra{0,1}\), let \(\widetilde\Pi_b\) be the distribution of
\[
    \widetilde\rho_b
    =
    p\rho_b\oplus(1-p)\ketbra00,
\]
where \(\rho_b\sim\Pi_b\) is conditioned on \(E_b\).
Since \(p\leq1/2\), we have $\rank\rbra{\widetilde\rho_0}\leq r$ on $E_0$, and $\Delta_r\rbra{\widetilde\rho_1}=p\Delta_{r-1}\rbra{\rho_1}\geq\delta$ on $E_1$.
Note that
\[
\mathrm T\rbra*{\E_{\widetilde\rho \sim \widetilde\Pi_0}\sbra{\widetilde\rho^{\otimes m}}, \E_{\widetilde\rho \sim \widetilde\Pi_1}\sbra{\widetilde\rho^{\otimes m}}} 
= \E_{J \sim \operatorname{Bin}\rbra{m,p}} \sbra*{ \mathrm{T}\rbra*{ \E_{\rho \sim \Pi_0} \sbra*{ \rho^{\otimes J} \middle| E_0 }, \E_{\rho \sim \Pi_1} \sbra*{ \rho^{\otimes J} \middle| E_1 } } }.
\]
For each $0 \leq j \leq m_0$, 
\begin{align*}
\mathrm{T}\rbra*{ \E_{\rho \sim \Pi_0} \sbra*{ \rho^{\otimes j} \middle| E_0 }, \E_{\rho \sim \Pi_1} \sbra*{ \rho^{\otimes j} \middle| E_1 } } 
& \leq \mathrm{T}\rbra*{ \E_{\rho \sim \Pi_0} \sbra*{ \rho^{\otimes j} }, \E_{\rho \sim \Pi_1} \sbra*{ \rho^{\otimes j} } } + \Pr\sbra*{E_0^{\mathsf{c}}} + \Pr\sbra*{E_1^{\mathsf{c}}},
\end{align*}
Therefore, for $1 \leq m \leq m_0/\rbra{8p}$, we have
\begin{align}
    \mathrm T\rbra*{
        \E_{\widetilde\rho \sim \widetilde\Pi_0}\sbra{\widetilde\rho^{\otimes m}},
        \E_{\widetilde\rho \sim \widetilde\Pi_1}\sbra{\widetilde\rho^{\otimes m}}}
    &\leq
    \max_{0 \leq j \leq m_0} \mathrm{T}\rbra*{ \E_{\rho \sim \Pi_0} \sbra*{ \rho^{\otimes j} }, \E_{\rho \sim \Pi_1} \sbra*{ \rho^{\otimes j} } } + \Pr\sbra*{E_0^{\mathsf{c}}} + \Pr\sbra*{E_1^{\mathsf{c}}} \nonumber \\
    & \qquad + \Pr_{J \sim \operatorname{Bin}\rbra{m, p}}\sbra{J > m_0}
    \nonumber\\
    &\leq
    q \rbra*{4^{-K}+e^{-D/384}} + 2 e^{-\gamma^2 q/32} + \rbra*{\frac{e}{8}}^{m_0},
    \label{eq:mps-one-block-closeness}
\end{align}
where we use \cref{eq:td-pi0-pi1} and \(\Pr\sbra{J>m_0}\leq \rbra{e/8}^{m_0}\). 

For \(b\in\cbra{0,1}\), consider the canonical purification of the $\rbra{D+1}$-dimensional quantum state $\widetilde{\rho}_b$:
\[
\ket{\widetilde{\rho}_b} = \sum_{i=1}^{D+1} \rbra*{ \rbra{\sqrt{\widetilde{\rho}_b}}_{\mathsf{A}} \otimes I_{\mathsf{B}}} \ket{i}_{\mathsf{A}} \ket{i}_{\mathsf{B}}.
\]
Let $\widehat{\Pi}_b$ be the distribution of the pure state $\rbra{I_{\mathsf{A}} \otimes U_{\mathsf{B}}}\ket{\widetilde\rho_b}$ where $U$ is a Haar-random unitary operator on subsystem $\mathsf{B}$. 
It can be seen that $\widehat{\Pi}_b$ is unitarily invariant on subsystem $\mathsf{B}$. 
By \cite[Corollary III.8]{CWZ26}, 
\[
\mathrm T\rbra*{\E_{\ket{\psi} \sim \widehat\Pi_0}\sbra{\ketbra{\psi}{\psi}^{\otimes m}}, \E_{\ket{\psi} \sim \widehat\Pi_1}\sbra{\ketbra{\psi}{\psi}^{\otimes m}}}
= \mathrm T\rbra*{\E_{\widetilde\rho \sim \widetilde\Pi_0}\sbra{\widetilde\rho^{\otimes m}}, \E_{\widetilde\rho \sim \widetilde\Pi_1}\sbra{\widetilde\rho^{\otimes m}}}.
\]
Let $\overline{\Pi}_b$ be the distribution of $\ket{\Psi_b}$ defined by
\[
\ket{\Psi_b} = \bigotimes_{j=1}^{\ell} \ket{\psi_{b,j}}, \qquad \ket{\psi_{b,1}}, \ket{\psi_{b,2}}, \dots, \ket{\psi_{b,\ell}} \sim \widehat{\Pi}_b.
\]
The proof of \cite[Proposition 5.1]{SW22} implies that
\[
\max_{\ket{\Phi} \in \textup{MPS}_{2\ell}\rbra{r}} \abs*{\braket{\Phi}{\Psi_b}}^2 = \prod_{j=1}^{\ell} \max_{\textup{SR}\rbra{\ket{\phi}} \leq r} \abs*{\braket{\phi}{\psi_{b,j}}}^2.
\]
Then, $\ket{\Psi_0} \in \textup{MPS}_{2\ell}\rbra{r}$, and 
\begin{align*}
\rbra*{
\Delta_{\mathrm{MPS}_{2\ell}\rbra r}
\rbra{\ket{\Psi_1}}
}^{2}
&=
1-
\max_{\ket\Phi\in\mathrm{MPS}_{2\ell}\rbra r}
\abs*{\braket{\Phi}{\Psi_1}}^{2}\\
&=
1-
\prod_{j=1}^{\ell}
\max_{\operatorname{SR}\rbra{\ket\phi}\leq r}
\abs*{\braket{\phi}{\psi_{1,j}}}^{2}\\
&=
1-
\prod_{j=1}^{\ell}
\left(
1-\Delta_r\rbra{\widetilde\rho_{1,j}}
\right)\\
&\geq
1-\rbra{1-\delta}^{\ell}=
\varepsilon^2,
\end{align*}
where $\widetilde{\rho}_{1,j} = \tr_{\mathsf{B}}\rbra{\ketbra{\psi_{1,j}}{\psi_{1,j}}}$.

If there is a tester for $\textup{MPS}_{n}\rbra{r}$, then it can be used to distinguish $\ket{\Psi_0} \sim \overline{\Pi}_0$ and $\ket{\Psi_1} \sim \overline{\Pi}_1$ with success probability $p_{\textup{succ}} \geq 2/3$. 
On the other hand, after regrouping the tensor factors as
\[
\begin{aligned}
\E_{\ket\Psi\sim\overline\Pi_b}
\sbra*{\ketbra\Psi\Psi^{\otimes m}}
\cong
\left(
\E_{\ket\psi\sim\widehat\Pi_b}
\sbra*{\ketbra\psi\psi^{\otimes m}}
\right)^{\otimes\ell},
\end{aligned}
\]
we have
\begin{align*}
\mathrm{T}\rbra*{\E_{\ket{\Psi} \sim \overline{\Pi}_0} \sbra*{ \ketbra{\Psi}{\Psi}^{\otimes m} }, \E_{\ket{\Psi} \sim \overline{\Pi}_1} \sbra*{ \ketbra{\Psi}{\Psi}^{\otimes m} }} 
& \leq \ell \cdot \mathrm{T}\rbra*{ \E_{\ket{\psi} \sim \widehat{\Pi}_0} \sbra*{ \ketbra{\psi}{\psi}^{\otimes m} }, \E_{\ket{\psi} \sim \widehat{\Pi}_1} \sbra*{ \ketbra{\psi}{\psi}^{\otimes m} } } \\
& = \ell \cdot \mathrm{T}\rbra*{ \E_{\widetilde\rho \sim \widetilde{\Pi}_0} \sbra*{ \widetilde\rho^{\otimes m} }, \E_{\widetilde\rho \sim \widetilde{\Pi}_1} \sbra*{ \widetilde\rho^{\otimes m} } } \\
& \leq \ell \rbra*{ q \rbra*{4^{-K}+e^{-D/384}} + 2 e^{-\gamma^2 q/32} + \rbra*{\frac{e}{8}}^{m_0} } < \frac{1}{3},
\end{align*}
which, by the Helstrom--Holevo bound \cite{Hel67,Hol73}, gives a contradiction that $p_{\textup{succ}} < \frac{1}{2} \rbra*{1 + \frac{1}{3}} = 2/3$. 
Therefore, any tester for $\textup{MPS}_{n}\rbra{r}$ requires sample complexity at least
\[
\frac{m_0}{8p} = \Omega\rbra*{\frac{r^2}{\delta K^2}} = \Omega\rbra*{\frac{nr^2}{\varepsilon^2\log^2\rbra{nr}}}.
\]
\end{proof}

\addcontentsline{toc}{section}{References}

\bibliographystyle{alphaurl}
\bibliography{main}

\appendix

\section{Another Sample Lower Bound for R\'enyi Entropy Estimation}

We show a sample lower bound for estimating the R\'enyi entropy of order $\alpha$ of an unknown discrete distribution for $0 < \alpha < 1$. 
It is noted that this lower bound was already given in \cite[Theorem 26]{AOST17}; however, their construction of hard instances does not give the desired sample lower bound for, e.g., $\alpha = 1/2$, in which case their entropy difference is $\Theta\rbra{\delta^2/d}$ rather than the desired $\Theta\rbra{\delta}$.
Here, we present a simple proof of this sample lower bound with a different hard instance. 

\begin{lemma}
    For $0 < \alpha < 1$, sufficiently large $d \geq 1$, and sufficiently small $\varepsilon > 0$, estimating the R\'enyi entropy $\mathrm{H}_{\alpha}^{\mathrm{R}}\rbra{P} = \frac{1}{1-\alpha} \log\rbra{\sum_{i=1}^d p_i^\alpha}$ of a $d$-dimensional discrete distribution $P$ to within additive error $\varepsilon$ requires sample complexity 
    \[
    \Omega\rbra*{\frac{d^{1/\alpha-1}}{\varepsilon^2}}.
    \]
\end{lemma}

\begin{proof}
    Let $N=d-1$ and $a=N^{-\rbra{1-\alpha}/\alpha}$, and, for $0\leq t\leq a/4$, define the probability distribution
\[
    P_t
    =
    \left(
    1-a-t,
    \underbrace{
    \frac{a+t}{N},\ldots,\frac{a+t}{N}}
    _{N\text{ entries}}
    \right).
\]
Then, it can be verified that 
\[
\abs*{ \mathrm{H}_{\alpha}^{\mathrm{R}}\rbra{P_0} - \mathrm{H}_{\alpha}^{\mathrm{R}}\rbra{P_t} } \geq \Omega\rbra*{\frac{t}{a}}, \qquad
d_{\textup{H}}^2\rbra{P_0, P_t} = \Theta\rbra*{\frac{t^2}{a}}.
\]
Taking $t = \Theta\rbra{a\varepsilon}$ gives $\abs{ \mathrm{H}_{\alpha}^{\mathrm{R}}\rbra{P_0} - \mathrm{H}_{\alpha}^{\mathrm{R}}\rbra{P_t} } \geq \Omega\rbra{\varepsilon}$, whereas distinguishing $P_0$ and $P_t$ with probability at least $2/3$ requires sample complexity (cf.\ \cite[Theorem 4.7]{BY02})
\[
\Omega\rbra*{\frac{1}{d_\textup{H}^2\rbra{P_0, P_t}}} = \Omega\rbra*{\frac{a}{t^2}} = \Omega\rbra*{\frac{d^{1/\alpha-1}}{\varepsilon^2}}.
\]
\end{proof}

\end{document}